\documentclass[11pt]{article}
\usepackage{amsmath} 
\usepackage{amssymb}
\usepackage{amsfonts}
\usepackage{latexsym}
\usepackage{amsthm}
\usepackage{amsxtra}
\usepackage{amscd}
\usepackage{bbm}
\usepackage{mathrsfs}
\usepackage{bm}
\usepackage{dsfont}
\usepackage{stmaryrd}

\usepackage{cite}

\usepackage{graphicx, color}

\newcommand{\f}[1]{\bm{\mathrm{#1}}} 
\newcommand{\bb}{\mathbb} 

\newcommand{\wh}{\widehat}
\newcommand{\wt}{\widetilde}

\newcommand{\ii}{\mathrm{i}}

\newcommand{\col}{\mathrel{\vcenter{\baselineskip0.75ex \lineskiplimit0pt \hbox{.}\hbox{.}}}}
\newcommand*{\deq}{\mathrel{\vcenter{\baselineskip0.65ex \lineskiplimit0pt \hbox{.}\hbox{.}}}=}

\renewcommand{\le}{\leqslant}
\renewcommand{\leq}{\leqslant}
\renewcommand{\ge}{\geqslant}
\renewcommand{\geq}{\geqslant}

\newcommand{\ceil}[1]  {\lceil  {#1} \rceil}

\newcommand{\ind}[1]{\f 1 (#1)}

\renewcommand{\epsilon}{\varepsilon}

\renewcommand{\P}{\mathbb{P}}
\newcommand{\E}{\mathbb{E}}
\newcommand{\R}{\mathbb{R}}
\newcommand{\C}{\mathbb{C}}
\newcommand{\N}{\mathbb{N}}
\newcommand{\Z}{\mathbb{Z}}

\newcommand{\p}[1]{({#1})}
\newcommand{\pb}[1]{\bigl({#1}\bigr)}
\newcommand{\pB}[1]{\Bigl({#1}\Bigr)}
\newcommand{\pbb}[1]{\biggl({#1}\biggr)}
\newcommand{\pBB}[1]{\Biggl({#1}\Biggr)}

\newcommand{\qB}[1]{\Bigl[{#1}\Bigr]}

\newcommand{\qBB}[1]{\Biggl[{#1}\Biggr]}

\newcommand{\hbb}[1]{\biggl\{{#1}\biggr\}}
\newcommand{\hBB}[1]{\Biggl\{{#1}\Biggr\}}

\newcommand{\abs}[1]{\lvert #1 \rvert}
\newcommand{\absb}[1]{\bigl\lvert #1 \bigr\rvert}
\newcommand{\absB}[1]{\Bigl\lvert #1 \Bigr\rvert}
\newcommand{\absbb}[1]{\biggl\lvert #1 \biggr\rvert}
\newcommand{\absBB}[1]{\Biggl\lvert #1 \Biggr\rvert}

\DeclareMathOperator{\tr}{Tr}
\DeclareMathOperator{\var}{Var}

\DeclareMathOperator{\im}{Im}

\theoremstyle{plain} 
\newtheorem{theorem}{Theorem}[section]
\newtheorem*{theorem*}{Theorem}
\newtheorem{lemma}[theorem]{Lemma}
\newtheorem*{lemma*}{Lemma}
\newtheorem{corollary}[theorem]{Corollary}
\newtheorem*{corollary*}{Corollary}
\newtheorem{proposition}[theorem]{Proposition}
\newtheorem*{proposition*}{Proposition}

\theoremstyle{definition} 
\newtheorem{definition}[theorem]{Definition}
\newtheorem*{definition*}{Definition}
\newtheorem{example}[theorem]{Example}
\newtheorem*{example*}{Example}

\newtheorem*{remark*}{Remark}
\newtheorem*{remarks*}{Remarks}

\newtheorem*{convention*}{Convention}

\newcommand{\nc}{\normalcolor}

\newcommand{\fn}{{\mathfrak n}}

\newcommand{\bD}{ {\bf  D}}

\newcommand{\fa}{{\frak a}} 

\newcommand{\bl}{{\boldsymbol \lambda}}
\newcommand{\ttau}{\vartheta}
\newcommand{\CC}{{\mathbb C }}
\newcommand{\RR}{{\mathbb R }}
\newcommand{\NN}{{\mathbb N}}

\newcommand{\bt}{{\boldsymbol \theta}}

\newcommand{\beqa}{\begin{eqnarray}}
\newcommand{\eeqa}{\end{eqnarray}}
\newcommand{\e}{\varepsilon}

\newcommand{\pt}{\partial}
\newcommand{\rd}{{\rm d}}
\newcommand{\bR}{{\mathbb R}}

\newcommand{\non}{\nonumber}

\newcommand{\br}{{\bf{r}}}

\newcommand{\bx}{{\bf{x}}}
\newcommand{\by}{{\bf{y}}}
\newcommand{\bu}{{\bf{u}}}
\newcommand{\bv}{{\bf{v}}}
\newcommand{\bw}{{\bf{w}}}
\newcommand{\bz}{{\bf {z}}}

\newcommand{\bh}{{\bf{h}}}
\newcommand{\bS}{\bf{S}}

\newcommand{\bla}{\mbox{\boldmath $\lambda$}}

\newcommand{\al}{\alpha}

\newcommand{\be}{\begin{equation}}
\newcommand{\ee}{\end{equation}}

\newcommand{\la}{\lambda}

\newcommand{\om}{{\omega}}
\newcommand{\si}{\sigma}
\renewcommand{\th}{\theta}

\newcommand{\cL}{{\mathscr L}}
\newcommand{\cE}{{\mathcal E}}
\newcommand{\cG}{{\mathcal G}}

\newcommand{\cK}{{\mathcal K}}
\newcommand{\cA}{{\mathcal A}}
\newcommand{\cR}{{\mathcal R}}
\newcommand{\cB}{{\mathcal B}}
\newcommand{\cW}{{\mathcal W}}
\newcommand{\cS}{{\mathcal S}}

\newcommand{\cH}{{\mathcal H}}

\newcommand{\ov}{\overline}

\numberwithin{equation}{section}
\numberwithin{theorem}{section}
\numberwithin{figure}{section}

\title{Universality for random matrices and log-gases \\ Lecture Notes for Current Developments in
Mathematics, 2012}

\author{
L\'aszl\'o Erd\H os\thanks{Partially supported
by SFB-TR 12 Grant of the German Research Council}
 \\\\
Institute of Mathematics, University of Munich, \\
Theresienstr. 39, D-80333 Munich, Germany \\ lerdos@math.lmu.de 
}

\begin{document}
\date{Dec 4, 2012}

\maketitle

\begin{abstract}

Eugene Wigner's revolutionary vision predicted that the energy levels of
large complex quantum systems exhibit a universal behavior: the statistics
of energy gaps depend only on the basic symmetry type of the model. 
These universal statistics  show strong correlations in the
form of level repulsion and they seem to represent a new paradigm
of point processes that  are characteristically different
from the Poisson statistics of independent points.  

Simplified models of
Wigner's thesis have recently become mathematically
accessible. For mean field models represented by large random matrices
with independent entries,
the  celebrated Wigner-Dyson-Gaudin-Mehta  (WDGM) conjecture 
asserts that the local eigenvalue 
statistics  are universal.
For invariant matrix models, the eigenvalue distributions are given 
by a log-gas with  potential $V$ and inverse temperature $\beta = 1, 2, 4$.
corresponding to  the orthogonal, unitary and symplectic ensembles. 
For $\beta \not  \in \{1, 2, 4\}$, there is no natural random matrix  ensemble  behind this model,
but the analogue of the WDGM conjecture asserts that  
the local  statistics are independent of $V$.  

In these lecture notes  we review the recent solution to these conjectures for both  
 invariant and non-invariant ensembles. We will discuss two different notions of
universality in the
sense of (i) local correlation functions and (ii) gap distributions.
 We will  demonstrate that the  local ergodicity 
of the  Dyson Brownian motion   is   the intrinsic  mechanism behind the 
universality. In particular,  we review the solution of Dyson's conjecture on the local relaxation time 
of the   Dyson Brownian motion. Additionally, the gap distribution requires
a De Giorgi-Nash-Moser type H\"older regularity analysis 
for a discrete parabolic equation with random coefficients.
  Related questions  such as the local version of Wigner's semicircle law
and delocalization of eigenvectors 
will also be discussed. We will also explain how these results can be
extended beyond the mean field models, especially to random band matrices.
\end{abstract}

{\bf AMS Subject Classification (2010):} 15B52, 82B44

\medskip

{\it Keywords:}  $\beta$-ensemble, local semicircle law,
Dyson Brownian motion. De Giorgi-Nash-Moser theory.

\setcounter{tocdepth}{4}

\newpage

\tableofcontents

\section{Introduction}

\subsection{The pioneering vision of Wigner}

\begin{minipage}[c]{5.5in}
{\it ``Perhaps I am now too courageous when I try to guess the distribution of the distances between
successive levels (of energies of heavy nuclei).   Theoretically, the 
situation is quite simple if one attacks the problem  in a 
simpleminded fashion.  The question is simply what are the 
 distances  of the characteristic values of  a symmetric
matrix with random coefficients."  }
\end{minipage}

\medskip 
\centerline{\qquad\qquad\qquad\qquad\qquad\qquad Eugene Wigner on 
the Wigner surmise, 1956 }   

\bigskip

Large complex systems often exhibit
remarkably simple universal patterns as the number of degrees of freedom
increases. The simplest example is the central limit theorem:
the fluctuation of the sums of independent  random scalars, irrespective of their distributions,
follows the Gaussian distribution. 
The other cornerstone of  probability theory identifies the
Poisson point process as 
the universal limit of many independent point-like events
in space or time.  These mathematical descriptions assume that
the original system has independent (or at least weakly dependent) constituents.
What if independence is not a realistic approximation and strong correlations
need to be modelled? Is there a universality for strongly correlated
models?

At first sight this seems an impossible task. While independence is a unique 
concept, correlations come in many different forms;
a-priori there is no reason to believe that they all behave similarly.
 Nevertheless  they do, according to the  pioneering vision of 
Wigner \cite{W} at least if they originate from certain
 physical systems and if the ``right'' question is asked.
 The actual correlated system
he studied was the energy levels of heavy nuclei. Looking at
spectral measurement data, it is obvious that the eigenvalue
density (or density of states, as it is called in physics)
heavily depends on the system. But Wigner asked a different
question: what about the distribution of the rescaled energy gaps?
 He discovered that the difference of consecutive energy levels,
after rescaling with the local density, shows a surprisingly universal
behavior. He even predicted a universal law, given by the simple
formula (called the {\it Wigner surmise}), 
\be\label{surmise}
   \P\Big( \wt E_j-\wt E_{j-1} 
= s+\rd s\Big) \approx 
\frac{\pi s}{2}
\exp\big( -  \frac{\pi}{4} s^2\big)\rd s,
\ee
where $\wt E_j = \varrho E_j$ denote the rescaling
of the actual energy levels $E_j$ by the density of states $\varrho$ near
the energy $E_j$.
This law is characteristically different from the gap
distribution of the Poisson process which is the exponential distribution, $e^{-s}\rd s$.
The prefactor $s$ in \eqref{surmise} indicates a {\it level repulsion}
 for the point process $\wt E_j$,
i.e. the eigenvalues are strongly correlated.

Comparing measurement data from various experiments, Wigner's pioneering vision
was that the energy gap distribution \eqref{surmise} 
of complicated quantum systems is essentially universal;
it depends only  on the basic symmetries of model
(such as time-reversal invariance). This thesis 
has never been rigorously proved for any realistic 
physical system but experimental data and extensive 
numerics leave no doubt on its correctness (see \cite{M} for an overview).

Wigner not only predicted universality in complicated systems,
but he also discovered a remarkably simple mathematical model for 
this new phenomenon: the eigenvalues of large random matrices.  
For practical purposes,  Hamilton operators of quantum models are often
approximated by large matrices 
that are obtained from some type of discretization 
of the original continuous model. These matrices
have specific forms dictated by physical  rules.
 Wigner's bold step was to neglect 
all details and consider the simplest random matrix
whose entries are independent and identically distributed.
The only physical property he retained
was the basic symmetry class of the system; time reversal physical
models were modelled by real symmetric matrices,
while systems without time reversal symmetry (e.g.
with magnetic fields) were modelled by complex
Hermitian matrices. As far as the
gap statistics are concerned, this simple-minded model
reproduced the behavior of the complex quantum systems!
The universal behavior extends to the joint statistics
of several consecutive gaps which are essentially equivalent
to the local correlation functions of the point process $\wt E_j$.
From mathematical point of view,  a universal 
strongly correlated point process was found. The natural
representatives of these universality classes
are the random matrices with independent identically distributed
Gaussian entries.
These are called the {\it Gaussian orthogonal ensemble (GOE)}
and the {\it Gaussian unitary ensemble (GUE)} in case of
real symmetric and complex Hermitian matrices, respectively.

 Since Wigner's discovery random matrix statistics are found 
everywhere in physics and beyond, wherever nontrivial correlations  prevail.
 Among many other applications,  random matrix theory (RMT) is present in chaotic quantum systems
in physics, in principal component analysis in statistics,
in communication theory and even in number theory. In particular,
the zeros of the Riemann zeta function on the critical line are
expected to follow RMT statistics due to a spectacular result of Montgomery \cite{Mont}.

In retrospect, Wigner's idea should have received even more attention.
For centuries, the primary territory of probability theory was to model uncorrelated
or weakly correlated systems.  The surprising ubiquity of random matrix statistics is
a strong evidence that it plays a similar fundamental role for correlated
systems as Gaussian distribution and Poisson point process play for
uncorrelated systems. RMT seems to provide essentially
 the only universal and generally   computable pattern  for complicated
 correlated systems.

In fact, a few years after Wigner's seminal paper \cite{W},
Gaudin \cite{Gau} has discovered another remarkable
property of this new point process:  the correlation
functions have a determinantal structure, at least
if the distributions of the matrix elements are Gaussian.
The algebraic identities within the determinantal form
 opened up the route to calculations and to obtain explicit formulas
for local correlation functions.
In particular, the gap distribution for the complex
Hermitian case is given by a Fredholm determinant
involving Hermite polynomials. In fact,  Hermite polynomials were 
first  introduced in the context of random matrices by Mehta and Gaudin \cite{MG} earlier.
 Dyson and Mehta \cite{M2, Dy1, Dy2}
have later extended this exact calculation 
to correlation functions and to other symmetry classes.
When compared with the exact formula,
the Wigner surmise \eqref{surmise}, based upon a simple $2\times 2$
matrix model, turned out to be quite accurate.
While the determinantal structure is present only in
Gaussian Wigner matrices, the paradigm of local universality
predicts that the formulas for the local eigenvalue statistics
obtained in the Gaussian case hold for general distributions as well.

\subsection{Physical models}

The ultimate mathematical goal is to prove Wigner's vision for 
a very large class of realistic quantum mechanical models.
This is extremely hard, since the local statistics involve
tracking individual eigenvalues in the bulk spectrum.
Wigner's original model, the energy levels of heavy nuclei,
is a strongly interacting many-body quantum system. The
rigorous analysis of such model with the required
precision is beyond the reach of current mathematics.

A much simpler question is to neglect all interactions and to study the
natural one-body quantum model, the  Schr\"odinger
operator $-\Delta +V$ with a potential $V$ on $\R^d$.
The complexity comes from assuming that 
$V$ is generic in some sense, in particular to exclude
models with additional symmetries that may lead to 
non-universal eigenvalue correlations. Two well-studied
examples are (i) the random Schr\"odinger operators where
$V=V(x)$ is a random field with a short range correlation,
and (ii) quantum chaos models, where $V$ is generic but fixed
and the statistical ensemble is generated by sampling
the spectrum in small spectral windows at high energies
(an alternative formulation uses the semiclassical limit).

Unfortunately, there are essentially no
rigorous results on local spectral universality  even  
in these one-body models. 
Random Schr\"odinger operators are conjectured to
exhibit a metal-insulator transition that
was  discovered by Anderson \cite{A}.
The high disorder regime is relatively well understood since the seminal
work of Fr\"ohlich and Spencer \cite{FS} (an alternative proof is given by
Aizenman and Molchanov \cite{AM}). However, in this regime the eigenfunctions are localized
and thus eigenfunctions belonging to  neighboring eigenvalues are typically spatially
separated, hence uncorrelated. Therefore, due to localization, the system does not
have sufficient correlation to  fall into the RMT universality class;
in fact  the local eigenvalue statistics  follow the Poisson
process \cite{Min}. In contrast, in the low disorder regime, starting from three 
spatial dimension and away from the spectral edges, the eigenfunctions
are conjectured to be delocalized ({\it extended states conjecture}).
Spatially overlapping eigenfunctions introduce correlations
among eigenvalues and it is expected that the local statistics are given by RMT.
In the theoretical physics literature, the existence of the delocalized regime and its
RMT statistics are considered as facts, supported both by non-rigorous arguments
and numerics. One of the most intriguing approach is via supersymmetric (SUSY) functional integrals
that remarkably reproduce all formulas obtained by the determinantal calculations
in much more general setup but in a non-rigorous way due to neglecting highly oscillatory
terms.  The rigorous mathematics seriously lags behind these developments;
even the existence of the delocalized regime is not proven, let alone detailed spectral statistics.

Judged from the horizons of theoretical physics,
rigorous mathematics does not fare much better in the  quantum chaos models either. 
The grand vision is that the quantization of an integrable classical
Hamiltonian system exhibits Poisson eigenvalue statistics and a
chaotic classical system gives rise to RMT statistics \cite{BGS,BT}.
While Poisson statistics have been shown to emerge some specific integrable models
\cite{Si, RS, Mar},
there is no rigorous result on the RMT statistics. Recently there has  been a 
remarkable mathematical progress in {\it  quantum unique ergodicity (QUE)}
that predicts that {\it all} eigenfunctions of chaotic
systems are uniformly distributed all over the space, at least in some 
macroscopic sense. For arithmetic domains QUE has been proved in  \cite{Li}.
For general manifolds much less is known, but a  lower bound on
the topological entropy of the support of the limiting densities
of eigenfunctions excludes that eigenfunctions are supported only 
on a periodic orbit \cite{An}. Very roughly, QUE can
be considered as the analogue of the extended states for random Schr\"odinger
operators.  Theoretically, the overlap of eigenfunctions should again lead to correlations
between neighboring eigenvalues, but their direct quantitative analysis would require
a much more precise understanding of the eigenfunctions.

\subsection{Random matrix ensembles}

In these lectures we consider even simpler models to test Wigner's universality
hypothesis, namely the random
matrix ensemble itself. The main goal is to show that their eigenvalues follow the local statistics
 of the Gaussian Wigner matrices which have earlier been computed explicitly
by Dyson, Gaudin and Mehta. The statement that the local eigenvalue 
statistics  is independent of the law of the matrix elements
 is generally referred to as the {\it universality conjecture 
 of  random matrices}  and we will call it the {\it Wigner-Dyson-Gaudin-Mehta conjecture}.
 It was first formulated in Mehta's treatise on
random matrices \cite{M} in 1967 and has remained a  key question
in the subject ever since. The goal of these lecture notes 
is to review the recent progress that has led 
to the proof of this conjecture and we sketch some important ideas.
We will, however, not be able to present all aspects of random matrices
and  we refer
the reader to recent comprehensive books  \cite{De1, DG1, AGZ}.

\subsubsection{Wigner ensembles}\label{sec:wigg}

To make the problem simpler, 
we restrict ourselves to either real symmetric or complex Hermitian matrices 
so that the eigenvalues are real.  The standard model 
consists of   $N\times N$ square  matrices 
$H= (h_{ij})$ with matrix elements having mean zero and variance $1/N$, i.e., 
\be
   \E\, h_{ij} = 0,   \quad \E |h_{ij}|^2= \frac{1}{N}
 \qquad i,j =1,2,\ldots, N.
\label{centered}
\ee
The matrix elements  $h_{ij}$,   $i,j=1, \ldots, N$,  are 
real or  complex  independent random variables 
subject to 
the symmetry constraint $h_{ij}= \ov h_{ji}$.   These ensembles of random matrices 
are called {\it (standard) Wigner matrices.} We will always consider
the limit 
as the matrix size goes to infinity, i.e., $N\to\infty$.
Every quantity related to $H$ depends on $N$, so we should have
used the notation $H^{(N)}$ and $h_{ij}^{(N)}$, etc., but
for simplicity we will omit  $N$    in the notation.

In Section~\ref{sec:lsc} we will also consider
generalizations of these ensembles, where we allow the matrix elements $h_{ij}$
to have different distributions (but retaining independence). The main motivation
is to depart from the mean-field character of the standard Wigner matrices,
where the quantum transition amplitudes $h_{ij}$ between any two sites $i, j$ 
have the same statistics.
The most prominent example is the random band matrix ensemble (see Example~\ref{ex})
that naturally interpolates between standard Wigner matrices and 
random Schr\"odinger operators with a short range hopping mechanism
(see \cite{Spe} for an overview).

The first rigorous result about the spectrum of a random matrix of this type
is the famous {\it Wigner semicircle law}
 \cite{W} which
  states that the  empirical density of the eigenvalues, $\lambda_1, \lambda_2, \ldots,
\lambda_N$, under the normalization \eqref{centered}, is given by 
\be\label{sc}
\varrho_N (x) := \frac{1}{N}\sum_{j=1}^N \delta(x-\lambda_j)\rightharpoonup
\varrho_{sc} (x) := \frac 1 { 2 \pi} \sqrt {(4 - x^2)_+}
\ee
in the weak limit as $N \to \infty$. The limit density is
independent  of the details of the distribution of $h_{ij}$.

The Wigner surmise \eqref{surmise} is a much finer problem since it concerns
individual eigenvalues and not only their  behavior on macroscopic scale.
To understand it, we introduce correlation functions.
If  $p_N(\lambda_1, \lambda_2, \ldots , \lambda_N)$ denotes the
joint probability density of the  (unordered) eigenvalues, then
the $n$-point correlation functions (marginals) are defined by
\be
  p^{(n)}_N(\la_1, \la_2, \ldots, \la_n):
 = \int_{\bR^{N-n}} p_N( \la_1, \ldots,\la_n, \la_{n+1},
\ldots \la_N) \rd\la_{n+1} \ldots \rd\la_{N}.
\label{pk}
\ee
 To keep this introduction simple, 
we state the corresponding results in terms of the
eigenvalue correlation functions for Hermitian 
$N\times N $  matrices.
In the Gaussian case (GUE) the joint probability
density of the eigenvalues can be expressed explicitly as
\be
    p_N(\lambda_1, \lambda_2, \ldots , \lambda_N) 
  = \mbox{const.} \prod_{i<j} (\lambda_i-\lambda_j)^2 \prod_{j=1}^N
   e^{- \frac{1}{2}N  \lambda_j^2},
\label{expli28}
\ee
where the normalization constant  can
be computed explicitly.
The Vandermonde determinant structure allows one to compute
the $k$-point  correlation functions in the large $N$ limit via 
Hermite polynomials that are the
orthogonal polynomials with respect to the Gaussian weight function.

The result of Dyson, Gaudin and Mehta  asserts that  for 
any fixed energy $E$ in the bulk of the spectrum, i.e., $|E| < 2$,  
the small scale behavior of $ p^{(n)}_N$ is given explicitly by 
\be
  \frac{1}{[\varrho_{sc}(E)]^n}
 p_N^{(n)}\Big( E+ \frac{\al_1}{N\varrho_{sc} (E)}, E + \frac{\al_2}{N\varrho_{sc}(E)},
 \ldots ,E+ \frac{\al_n}{N\varrho_{sc}(E)}\Big)  \to 
\det \big( K(\al_i - \al_j)\big)_{i,j=1}^n
\label{sine}
\ee
where $K$ is the  celebrated sine kernel
\begin{align}
K(x,y)
  =  \frac{\sin  \pi (x-y)}{\pi(x-y)}.
\end{align}
Note that the limit in \eqref{sine} is independent of the energy $E$
as long as it lies in the bulk of the spectrum. 
 The rescaling by a factor $N^{-1}$ of the correlation functions in \eqref{sine}
corresponds to the typical distance between consecutive eigenvalues 
and we will refer to  the law under  such scaling
 as  {\it local  statistics}. Note that the correlation functions do not factorize, i.e.
the eigenvalues are strongly correlated despite that the matrix elements
are independent. 
Similar but  more complicated  formulas were  obtained 
for symmetric matrices and also for the self-dual quaternion random
matrices which is the  third symmetry class of 
random matrix ensembles.

The convergence in \eqref{sine} holds
for each fixed $|E|<2$ and uniformly in $(\al_1, \ldots , \al_n)$
in any compact subset of $\bR^n$. Fix now $k$ compact subsets $A_1, \ldots A_k$
in $\bR$. From \eqref{sine} one can compute the
distribution of the number $n_j$ of the rescaled eigenvalues $\wt \lambda_\al: =  
N(\lambda_\al-E)\varrho_{sc}(E)$ in  $A_j$ around a fixed energy $|E|<2$.
The limit of the joint probabilities
\be\label{PA}
  \P\Big( \#\big\{ \wt\lambda_\al \in A_j\big\}=n_j, \; j=1,2,\ldots ,k\Big)
\ee
is given as derivatives of a Fredholm determinant involving the sine kernel. 
Clearly \eqref{PA} gives a complete local description of the rescaled eigenvalues
as a point process around a fixed energy $E$. In particular it describes the distribution of the 
eigenvalue gap that contains a {\it fixed energy $E$}. However, \eqref{PA}
does not determine the distribution of the gap with a {\it fixed label,}
e.g. the gap $\lambda_{N/2+1}-\lambda_{N/2}$. Only the cumulative statistics
of many consecutive gaps can be deduced, see \cite{De1} for a precise formulation.
The slight discrepancy between the statements at fixed energy and  with fixed label 
leads to involved technical complications.

\subsubsection{Invariant ensembles}

The explicit formula \eqref{expli28} is special for Gaussian Wigner matrices;
if $h_{ij}$ are independent but non-Gaussian, then no analogous explicit formula is
known for the joint probability density. Gaussian Wigner matrices have this special property because their
distribution is {\it invariant} under base transformation. The derivation
of \eqref{expli28} relies on the fact that in the diagonalization $H=U\Lambda U^*$
of $H$, where $\Lambda$ is diagonal and $U$ is unitary, the distributions
of $U$ and $\Lambda$ decouple. The Gaussian measure of $h_{ij}$ with the
normalization \eqref{centered} can also be expressed as 
\be\label{expl3}
\exp\Big(-\frac{1}{2}N\tr H^2\Big)\rd H = \exp\Big(-\frac{1}{2}N\tr \Lambda^2\Big)\rd (U\Lambda U^*),
\ee
where $\rd H$ is the Lebesgue measure on hermitian matrices.
The Vandermonde determinant in \eqref{expli28} originates from the integrating the
 Jacobian $\rd (U\Lambda U^*)/\rd \Lambda$ over the unitary group.
Similar argument holds for real symmetric matrices with orthogonal conjugations, the only 
difference is the exponent 2 of the Vandermonde determinant becomes 1.
The exponent is 4 for
 the third symmetry class of Wigner matrices, the self-dual quaternion
matrices with symmetry group being the symplectic matrices ({\it Gaussian symplectic ensemble, GSE}).

 Starting from \eqref{expli28}, there are
two natural generalizations of Gaussian Wigner matrices. One direction
is the Wigner matrices with non-Gaussian but independent entries that we have
already introduced in Section~\ref{sec:wigg}.
Another direction is to consider a more general real function  $V(H)$ of $H$  instead of the quadratic
$H^2$ in \eqref{expl3}. Since invariance still holds, 
$\tr V(H)= \tr V(U\Lambda U^*) = \tr V(\Lambda)$, the
same argument gives \eqref{expli28}, with $V(\lambda_i)$ instead of $\lambda_j^2/2$,
for the correlation functions of $\exp(-N\tr V(H))$. These
are called {\it invariant ensembles} with potential $V$. Their matrix elements
are in general correlated except in the Gaussian case.

Invariant ensembles in all three symmetry classes can be given simultaneously
by the probability measure
$$
 Z^{-1} e^{- \frac{1}{2}N \beta  { \rm Tr} V(H)} \rd H,
$$ 
where $N$ is the size of the matrix $H$, 
 $V$ is a real valued potential and $Z=Z_N$ is the normalization constant.  
The positive parameter $\beta$ is determined by the symmetry
class, its value is 1, 2 or 4, for real symmetric, complex hermitian and
self-dual quaternion matrices, respectively. The Lebesgue 
measure $\rd H$ is understood over the matrices in the same class.
The probability distribution of the eigenvalues $\bla=(\lambda_1, \ldots , \lambda_N)$
is given by the explicit formula (c.f. \eqref{expli28})
\begin{equation}\label{01}
\mu^{(N)}_{\beta, V}(\bla)\rd \bla \sim e^{- \beta N \cH(\lambda)} \rd\bla
 \quad \mbox{with Hamiltonian} \quad
\cH(\bla) :=   \sum_{k=1}^N  \frac{1}{2}V(\lambda_k)- 
\frac{1}{N} \sum_{1\leq i<j\leq N}\log (\lambda_j-\lambda_i).
\end{equation}
The key structural ingredient
 of this formula, the logarithmic interaction that gives
rise to the the Vandermonde determinant, 
is the same as in the Gaussian case, \eqref{expli28}. Thus all previous computations, 
developed for the Gaussian
case, can be carried out for $\beta=1, 2, 4$,
provided that the Gaussian weight function  for the orthogonal polynomials 
is replaced with the function $e^{- \beta V(x) /2}$. 
The analysis of the correlation functions depends 
critically  on the the asymptotic properties of the 
 corresponding orthogonal polynomials.

While the asymptotics of the Hermite polynomial for the Gaussian case are well-known,
the extension of  the necessary analysis to  a general potential is a 
demanding task;  important  progress was  made since  the late 1990's by 
Fokas-Its-Kitaev \cite{FIK},   Bleher-Its \cite{BI}, Deift {\it et. al.}
\cite{De1,  DKMVZ1, DKMVZ2},  Pastur-Shcherbina \cite{PS:97, PS} 
and more recently by
Lubinsky \cite{Lub}. { These results concern the simpler $\beta=2$ case.}
   For $\beta =1, 4$, the universality was established only quite recently 
for analytic $V$ with additional assumptions \cite{DG, DG1, KS, Sch} 
using earlier ideas of Widom  \cite{Wid}. 
 The final outcome  of these sophisticated  analyses is that  universality holds for the 
measure \eqref{01} in the sense 
that the short scale behavior of the correlation functions is independent of the potential
 $V$ (with appropriate assumptions) 
provided that  $\beta$ is one of the classical values,
 i.e., $\beta \in \{ 1,2,4\}$, that corresponds
to an underlying matrix ensemble.

 Notwithstanding  matrix ensembles or orthogonal polynomials,
the measure \eqref{01} on $N$ points $\lambda_1, \ldots , \lambda_N$
 is perfectly well defined for any $\beta>0$.
It can be interpreted as the Gibbs measure for a system of
particles with external potential $\frac{1}{2}V$ and with a logarithmic interaction (log-gas)
 at inverse temperature $\beta$. From this point of view $\beta$ is a continuous
parameter and the classical values $\beta=1,2,4$ play apparently no distinguished role.
 It is therefore natural 
to extend the universality problem  to all non-classical $\beta$ but
 the orthogonal polynomial methods are difficult to apply for this case.
  For any $\beta > 0$ the local statistics for the Gaussian case $V(x)=x^2/2$ 
is given by a point process, denoted by $\mbox{Sine}_\beta$. It can be obtained
from a rescaling of the $\mbox{Airy}_\beta$ process 
as $\lim_{a\to\infty} \sqrt{a}(\mbox{Airy}_\beta + a) = \mbox{Sine}_\beta$.
The Airy process itself is the low lying eigenvalues of 
the one dimensional Schr\"odinger operator $-\frac{\rd^2}{\rd x^2} + x + \frac{2}{\sqrt{\beta}}b'_x$
on the positive half line, where $b_x'$ is the white noise. 
The relation between Gaussian random matrices and random Schr\"odinger operators
is  derived from  a tridiagonal matrix representation
 \cite{DumEde}.  Another convenient representation of
the $ \mbox{Sine}_\beta$ process is given by the
  ``Brownian carousel''  \cite{RRV, VV}.

Beyond random matrices, the log-gas can also be viewed
as the only interacting particle model with a scale-invariant
interaction and with a single relevant parameter, the inverse
temperature $\beta$. It is believed to be
the canonical model for strongly correlated systems and thus
to play a similarly fundamental role in probability theory
as the Poisson process or the Brownian motion.
Nevertheless, we still
have very little information about its properties. 
Unlike the universality problem that is inherently analytical,
many properties of the log-gas are destined, at the first sight, to be
revelead by  smart algebraic identities. Despite many trials
by physicists and mathematicians,
 the log-gas with 
a general  $\beta$ seems to defy all algebraic attempts.
We do not really understand why
the algebraic approach is suitable for $\beta=2$, and to a lesser extent for $\beta=1,4$,
but it fails for any other $\beta$, while 
from an analytical point of view there is no difference
between various values of $\beta$.  To understand this
fascinating ensemble, a main goal  is to develop
general analytical methods that work for any $\beta$.

\subsection{Universality of the local statistics: the main results}

All universality results reviewed in the previous sections
rely  on some version of the explicit formula \eqref{01}
that is not available for Wigner matrices with non-Gaussian matrix elements. The only result 
prior 2009 towards universality for Wigner matrices
was the proof of Johansson \cite{J} (extended by Ben Arous-P\'ech\'e \cite{BP})
for complex Hermitian Wigner matrices with a  substantial Gaussian
component. The hermiticity is necessary, since the proof
still relies on an algebraic formula, a  modification of
 the Harish-Chandra/Itzykson/Zuber integral observed first 
by Br\'ezin and Hikami in this context \cite{BH}.

To indicate the restrictions imposed by the usage of
explicit formulas, we note that previous methods were not suitable
to deal even with very small perturbations of the Gaussian Wigner case.
For example, universality was already not known if only a few matrix
elements of $H$ had a distribution different from Gaussian.

Given this background, the main challenge a few years ago was
to develop a new approach to universality that does not rely
on any algebraic identity. We believe that the genuine reason behind
Wigner's universality is of analytic nature. Algebraic computations
may be used to obtain explicit formulas for the most convenient representative
of a universality class (typically the Gaussian case), but only
analytical methods have the power to deal with the general case.
In light of the two main classes of random matrix ensembles, we
set the following two main problems.

\medskip
\noindent
{\it Problem 1}: Prove the Wigner-Dyson-Gaudin-Mehta conjecture, i.e.
the universality for Wigner matrices
 with a general distribution for the matrix elements.

\medskip 
\noindent 
{\it Problem 2}:  Prove the universality of  the local statistics
for the log-gas \eqref{01} for all $\beta > 0$.

\medskip

We were able to  solve  Problem 1
 for a very general
class of distributions. As for  Problem 2, we solved it for the case of
  real analytic   potentials $V$
 assuming that the equilibrium measure is supported on a single interval,
which, in particular, holds for any convex potential. We will give
a historical overview of related results in Section~\ref{sec:hist}.

The original universality conjectures, as formulated in Mehta's book \cite{M},
do not specify the type of convergence in \eqref{sine}. 
We focus on two types of results for both problems. First we
show that universality holds in the sense that local correlation functions
around an energy $E$ converge weakly if  $E$ is averaged on a small
interval of size $N^{-1+\e}$. Second, we prove the universality of the
joint distribution of consecutive gaps with fixed labels. 

We note that universality of the {\it cumulative} statistics of $N^\e$ gaps
directly follows from the weak convergence of the correlation functions
but our result on a single gap requires a quite different approach.
{F}rom the point of view of Wigner's original vision on the ubiquity
of the random matrix statistics in seemingly disparate ensembles and physical systems,
the issue of cumulative gap statistics versus single gap statistics is
minuscule. Our main reason of pursuing the single gap
universality is less for the result itself; more importantly, we
develop new methods to analyze the structure of the log-gases,
which seem to represent the universal statistics of strongly
correlated systems. In the next two sections we state the results precisely.

\subsubsection{Generalized Wigner matrices}

Our main results hold for a larger class of ensembles than the standard Wigner matrices, which
we will call {\it generalized Wigner matrices}.

\begin{definition}\label{def:genwig}(\cite{EYY})
 The real symmetric or complex Hermitian
matrix ensemble $H$ with centred and independent matrix elements $h_{ij}=\ov{h}_{ji}$, $i\le j$,
is called generalized Wigner matrix if 
the following  assumptions hold on the variances of the matrix
elements $s_{ij}= \E |h_{ij}|^2$:
\begin{description}
\item[(A)] For any $j$ fixed
\be
   \sum_{i=1}^N s_{ij} = 1 \, .
\label{sum}
\ee

\item[(B)]   There exist two positive constants, $C_1$ and $C_2$,
independent of $N$ such that
\be\label{1.3C}
\frac{C_1}{N} \le s_{ij}\leq \frac{C_2}{N}.
\ee
\end{description} 
\end{definition}

The result on the correlation functions is the following theorem:

\begin{theorem}[Wigner-Dyson-Gaudin-Mehta conjecture for averaged correlation functions]\cite[Theorem 7.2]{EKYY2}
 \label{bulkWigner}
Suppose that $H = (h_{ij})$ is a complex Hermitian (respectively, real symmetric) generalized Wigner matrix.
 Suppose that for some constants $\e>0$, $C>0$,  
\begin{equation} \label{4+e}
\E \left | \sqrt N  h_{ij}  \right | ^{4 + \e}  \;\leq\; C.
\end{equation}
Let $n \in \N$ and $O : \R^n \to \R$ be compactly supported and continuous.
Let $E$ satisfy $|E|<2$ and let $\xi > 0$. 
Then for any sequence $b_N$ satisfying
 $N^{-1 + \xi} \leq b_N 
\leq \left | { |E|- 2}  \right |/2 $ we have
\begin{multline}\label{avg}
\lim_{N \to \infty} \int_{E - b_N}^{E + b_N} \frac{\rd x}{2 b_N} \int_{\bR^n} \rd \alpha_1 \cdots
 \rd \alpha_n\, O(\alpha_1, 
\dots, \alpha_n) 
\\ 
{}\times{}  \frac{1}{\varrho_{sc}(E)^n} \left ( {p_{N}^{(n)} - p_{{\rm G}, N}^{(n)}} \right ) 
 \left ( {x +
\frac{\alpha_1}{N\varrho_{sc}(E)}, \dots, x + \frac{\alpha_n}{N\varrho_{sc}(E)}}\right )  \;=\; 0\,.
\end{multline}
Here $\varrho_{sc}$ is the semicircle law
 defined in \eqref{sc}, $p_N^{(n)}$ is the $n$-point correlation function
 of the eigenvalue 
distribution of $H$ \eqref{pk}, and $p^{(n)}_{{\rm G},N}$ is the $n$-point correlation function
of an $N \times N$ GUE (respectively, GOE) matrix.
\end{theorem}
The condition \eqref{1.3C} can be  relaxed, see Corollary 8.3 \cite{EKY3}. For example, 
the lower bound can be changed to  $N^{-9/8+\e}$. Alternatively,
the upper bound $C_2N^{-1}$ can replaced with $N^{-1+\e_n}$ for some  $\e_n>0$.  For band matrices,
the upper and lower bounds can be simultaneously relaxed.

We remark that for the complex Hermitian case the convergence of
the correlation functions can be strengthened to a
convergence at each fixed energy, i.e. 
for any fixed $|E|<2$ we have  that
\be\label{pointwise}
 \int_{\bR^n} \rd \alpha_1 \cdots
 \rd \alpha_n\, O(\alpha_1, 
\dots, \alpha_n)
 \frac{1}{\varrho_{sc}(E)^n} \left ( {p_{N}^{(n)} - p_{{\rm G}, N}^{(n)}} \right ) 
 \left ( {E +
\frac{\alpha_1}{N\varrho_{sc}(E)}, \dots,
 E + \frac{\alpha_n}{N\varrho_{sc}(E)}}\right )  \;=\; 0\,.
\ee
The main ideas leading to the results \eqref{avg} and \eqref{pointwise} have been 
developed in a series of papers. We will give a short overview of the key methods in
Section~\ref{sec:str1} and of the related results in Section~\ref{sec:hist}.

\medskip

The second result on generalized Wigner matrices
asserts that the local gap statistics in the bulk of the spectrum
are universal for any general Wigner matrix, in particular
they coincide with those of the Gaussian case.
To formulate the statement,
we need to introduce the notation  $\gamma_{j}$
for the $j$-th quantile of the semicircle density, i.e. $\gamma_j=\gamma_j^{(N)}$ is defined by
\be\label{defgamma}
   \frac{j}{N} = \int_{-2}^{\gamma_{j}}\varrho_{sc}(x) \rd x.
\ee
We also introduce the notation $\llbracket A, B\rrbracket : = \{ A, A+1, \ldots, B\}$
for any integers $A<B$.

\begin{theorem}  [Gap universality for Wigner matrices]\cite[Theorem 2.2]{EYsinglegap}
 \label{thm:sg} Let $H$ be a generalized real symmetric or complex Hermitian
Wigner matrix with subexponentially decaying matrix elements, i.e. we assume that
 \be\label{subexp}
 \P \big( \sqrt{N}|h_{ij}|\ge x\big) \le C_0\exp (-x^\vartheta)
\ee
holds for any $x>0$ with some $C_0, \vartheta$ positive constants.
 Fix a positive number $\al>0$,
an integer $n\in \N$ and a smooth, compactly supported function
$O:\bR^n\to \bR$. There exists an $\e>0$ and $C>0$,   depending only on $C_0, \vartheta$, $\alpha$ and  $O$
 such that 
\be\label{EEO}
   \Big| \big[\E -\E^{\mu}\big] O\big( N(x_j-x_{j+1}), N(x_j-x_{j+2}), \ldots , N(x_j - x_{j+n})\big)\Big|
    \le CN^{-\e}, 
\ee
for any $j\in \llbracket \al N, (1-\al)N\rrbracket$ and
for any sufficiently large  $N\ge N_0$, where $N_0$ depends
on all parameters of the model, as well as on $n$ and $\alpha$.  
Here $\E$ and $\E^\mu$ denotes the expectation with respect to the
Wigner ensemble $H$ and the Gaussian
equilibrium measure (see \eqref{expli28}
for the Hermitian case), respectively.

More generally, for any $k, m \in \llbracket \al N, (1-\al)N\rrbracket$ we have
\begin{align} \label{EEO1}
   \Big| \E  O\big( & (N\varrho_k) (x_k-x_{k+1}),   (N\varrho_k)(x_k-x_{k+2}), \ldots , (N\varrho_k)(x_k - x_{k+n})\big) \\
& -\E^{\mu} O\big( (N\varrho_m)(x_m-x_{m+1}), (N\varrho_m)(x_m-x_{m+2}), \ldots , (N\varrho_m)(x_m - x_{m+n})\big)\Big|
    \le CN^{-\e},  \non
\end{align}
where the local density $\varrho_k$ is defined by
$\varrho_k : =\varrho_{sc}(\gamma_k)$.
\end{theorem}

As it was already mentioned, the gap universality with a certain local averaging,
i.e.  for the cumulative statistics of $N^{\e}$ 
consecutive gaps, follows directly from the universality of the correlation functions,
Theorem~\ref{bulkWigner}.  The gap distribution for  Gaussian random matrices,
with a local averaging,  can then be explicitly
expressed via a Fredholm determinant, 
see \cite{De1, DG, DG1}. 
The first result for a single gap, i.e. without local averaging,  
was only achieved recently  in  the special 
case of the Gaussian unitary ensemble (GUE) in \cite{Taogap}, which statement  then  easily  implies
 the same results for  complex Hermitian Wigner  matrices satisfying  the  
  four moment matching condition.

\subsubsection{Log-gases}

In the case of  invariant ensembles, it is well-known that for 
 $V$ satisfying certain mild conditions the sequence
of one-point correlation functions, or
 densities, associated with  $\mu=\mu^{(N)}$ from \eqref{01} has a limit
as $N\to\infty$
and the limiting  equilibrium density  $\varrho_V(s)$
can be obtained as the unique minimizer  of the
functional
$$
I(\nu)=
\int_\bR V(t) \nu(t)\rd t-
\int_\bR \int_\bR \log|t-s| \nu(s) \nu(t) \rd t \rd s.
$$
We assume that 
$\varrho=\varrho_{V}$ is supported on a single compact interval, $[A,B]$ and $\varrho\in C^2(A,B)$. 
Moreover, we assume that  $V$ is {\it regular} in the sense that $\varrho$ is strictly positive on $(A, B)$
and vanishes as a square root at the endpoints, see (1.4) of \cite{BEY2}.
It is known that these condition are satisfied if, for example, $V$ is strictly convex.
In this case $\varrho_V$ satisfies the equation 
\be
   \frac{1}{2}V'(t) = \int_\bR \frac{\varrho_V(s)\rd s}{t-s}
\label{equilibrium}
\ee
for any $t\in(A,B)$.  For the Gaussian case, $V(x)=x^2/2$, the equilibrium density
is given by the semicircle law, $\varrho_V=\varrho_{sc}$, see \eqref{sc}.

The following result was proven in Corollary  2.2 of \cite{BEY}
for convex potential $V$ and it was generalized in Theorem 1.2 of \cite{BEY2} 
for the non-convex case.

\begin{theorem}[Bulk universality of $\beta$-ensemble] \label{bulkbeta}
 Assume
$V$ is  real analytic   with 
$\inf_{x\in\R}V''(x) > -\infty$. 
Let $\beta> 0$.
 Consider the $\beta$-ensemble $\mu=\mu_{\beta, V}^{(N)}$ given in \eqref{01} 
and let
 $p_N^{(n)}$ denote the $n$-point correlation functions
of $\mu$, defined analogously to \eqref{pk}.
 For the Gaussian case, $V(x) =x^2/2$, the correlation
functions are denoted by $p_{G,N}^{(n)}$.
Let $E\in (A,B)$ lie in the interior of the support of $\varrho$
 and
similarly let $E'\in (-2,2)$ be inside the support of $\varrho_{sc}$. 
Let $O:\R^n\to \R$ be a smooth, compactly supported function.
Then for $b_N = N^{-1+\xi}$ with any $0<\xi \le 1/2$ we have
\begin{align}
\lim_{N \to \infty} \int  & \rd \alpha_1 \cdots \rd \alpha_n\, O(\alpha_1, 
\dots, \alpha_n) \Bigg [ 
   \int_{E - b_N}^{E + b_N} \frac{\rd x}{2 b_N}  \frac{1}{ \varrho (E)^n  }  p_{N}^{(n)}   \Big  ( x +
\frac{\alpha_1}{N\varrho(E)}, \dots,   x + \frac{\alpha_n}{N\varrho(E)}  \Big  ) \\ \nonumber
&
 -   \int_{E' - b_N}^{E' + b_N} \frac{\rd x}{2b_N}  \frac{1}{\varrho_{sc}(E')^n} p_{{\rm G}, N}^{(n)}      \Big  ( x +
\frac{\alpha_1}{N\varrho_{sc}(E')}, \dots,   x + \frac{\alpha_n}{N\varrho_{sc}(E')}  \Big  ) \Bigg ]
 \;=\; 0\,,
\end{align}
 i.e. the appropriately normalized  correlation functions  of  the
measure $\mu_{\beta, V}^{(N)}$ at the level $E$ in the bulk of the limiting density 
 asymptotically coincide with those
of the Gaussian case. In particular,  they are independent of the value of $E$.
\end{theorem}

For the corresponding theorem on the single gap we need to define the classical location  of the $j$-th particle 
$\gamma_{j,V}$ by
\be\label{defgammagen}
   \frac{j}{N} = \int_A^{\gamma_{j,V}}\varrho_V(x) \rd x,
\ee
similarly to the quantiles $\gamma_j$ of the semicircle law, see \eqref{defgamma}.
We set
\be\label{rhov}
   \varrho_j^V := \varrho_V(\gamma_{j,V}),\qquad
\mbox{and}\qquad \varrho_j : = \varrho_{sc} (\gamma_{j})
\ee
to be the limiting densities at the  classical  location of the $j$-th particle.
 Our main theorem on the 
$\beta$-ensembles is the following. 

\begin{theorem}  [Gap universality for $\beta$-ensembles]\cite[Theorem 2.3]{EYsinglegap}
 \label{thm:beta} Let   $\beta \ge 1$   and $V$ be a real analytic potential
with $\inf V''> -\infty$,  such that $\varrho_V$ is  supported on a single compact interval,
 $[A,B]$, $\varrho_V\in C^2(A,B)$, and that  $V$ is regular.
 Fix a positive number $\al>0$,
an integer $n\in \N$ and a smooth, compactly supported function
$O:\bR^n\to \bR$.  Let $\mu=\mu_V = \mu_{\beta, V}^{(N)}$ be given by \eqref{01}
and let $\mu_G$ denote the same measure for the Gaussian case, $V(x)=\frac{1}{2}x^2$.
Then there exist an $\e>0$, depending only on $\alpha, \beta$ 
and  the potential $V$,
and a constant $C$ depending on $O$ such that
\begin{align}\label{betaeq}
 \Bigg| \E^{\mu_V} & O\Big( (N\varrho_k^V) (x_k-x_{k+1}), (N\varrho_k^V) (x_{k}-x_{k+2}), \ldots , 
(N\varrho_k^V) (x_k - x_{k+n})\Big) \\ \nonumber
  &  - \E^{\mu_{G}} 
  O\Big( (N\varrho_m) (x_m-x_{m+1}), (N\varrho_m) (x_{m}-x_{m+2}), \ldots , 
(N\varrho_m)(x_m - x_{m+n})\Big)
 \Bigg| 
  \le CN^{-\e}
\end{align}
for any $k,m\in \llbracket \al N, (1-\al)N\rrbracket$ and
for any sufficiently large   $N\ge N_0$, where $N_0$ depends
on $V$, $\beta$, as well as on $n$ and $\alpha$.
In particular, the distribution of the rescaled gaps
 w.r.t. $\mu_V$ 
does not depend on the index $k$ in the bulk.
\end{theorem}

We point out that Theorem~\ref{bulkbeta} holds for any $\beta>0$, but Theorem~\ref{thm:beta}
requires $\beta\ge 1$. Most likely this is only a technical restriction 
related to a certain condition in the De Giorgi-Nash-Moser regularity theory
that is the backbone of our proof.

\subsection{Some remarks on the general strategy and on related results}\label{sec:genrem}

\subsubsection{Strategy for the universality of correlation functions}\label{sec:str1}

The proof  of Theorem \ref{bulkWigner} 
   consists of  the following three steps,
discussed in Sections~\ref{sec:lsc}, \ref{sec:DBM} and \ref{sec:4mom}, respectively.
 This  three-step  strategy was first introduced in \cite{EPRSY}.

\bigskip 
\noindent
{\bf Step 1.}  {\it Local semicircle law and delocalization of eigenvectors:}
It states that the density of eigenvalues
is  given by
the semicircle law not only as a weak limit on macroscopic scales \eqref{sc}, but
also in a strong sense and down to short scales containing 
only $N^\e$ eigenvalues for all $\e> 0$.  
This will imply the  {\it rigidity of eigenvalues}, 
 i.e., that   the eigenvalues are near their classical location
  in the sense to be made 
clear in Section~\ref{sec:DBM}. 
We also obtain precise estimates on the matrix elements of the Green function
 which in particular imply complete delocalization of eigenvectors.

\bigskip 
\noindent
{\bf Step 2.}
{\it Universality for
Gaussian divisible ensembles:}  The Gaussian divisible ensembles are complex or real 
Hermitian matrices of the form 
$$
 H_t= e^{-t/2} H_0+ \sqrt {1 - e^{-t}} U,
$$
 where $H_0$ is a Wigner matrix and
$U$ is an independent GUE/GOE matrix. The parametrization of  $H_t$ reflects that $H_t$ is most
 conveniently obtained by an Ornstein-Uhlenbeck  process.
  There are two methods  and both methods imply 
 the bulk  universality of $H_t$ for $t = N^{-\tau}$ for the entire  range of $0< \tau<1$ with
  different estimates.

\begin{description}
\item[2a.] {\it Proposition 3.1 of \cite{EPRSY} which uses an extension of Johansson's formula \cite{J}.}
\item[2b.] {\it  Local ergodicity of the Dyson Brownian motion (DBM).}
\end{description}

The approach in 2a yields a slightly stronger estimate 
(no local averaging in the energy)
than  the approach in 2b,  but it works only in the complex Hermitian case.  
In these notes, we will focus on 2b.  As time evolves,
the eigenvalues of $H_t$ evolve according to a system of stochastic differential equations,
the Dyson Brownian motion. The distribution
of the eigenvalues of $H_t$ will be written as $f_t\mu$, where $\mu$ is the equilibrium 
measure \eqref{expli28}.  We will study the evolution equation $\pt_t f_t=\cL f_t$,
where $\cL$ is the generator to the Dirichlet form $\int |\nabla f|^2\rd\mu$.
 As time goes to infinity, $f_t$ converges to constant, i.e. 
to equilibrium. The key technical question is the speed to local equilibrium.

\medskip 
\noindent
{\bf Step 3.}  {\it Approximation by Gaussian divisible ensembles:}
 It is a simple 
 density argument  in the space
of matrix ensembles which 
shows that for any probability distribution of the matrix
 elements there exists a Gaussian divisible distribution
with a small Gaussian component, as in Step 2, such 
that the two  associated Wigner ensembles
have asymptotically identical local eigenvalue statistics. 
 The first implementation  of this approximation scheme was via  a reverse heat flow 
argument  \cite{EPRSY}; it was later replaced by  the
  {\it Green function comparison theorem} \cite{EYY} that was
motivated by the four moment matching condition of \cite{TV}.

\medskip

\noindent
The proof of Theorem \ref{bulkbeta}  consists of  the following two steps
that will be presented in Sections~\ref{beta} and \ref{sec:loceq}.

\medskip
\noindent
{\bf Step 1.  Rigidity of eigenvalues.}  
This establishes that the location of the eigenvalues  are not too far
from their classical locations  $\gamma_{j,V}$ determined by the equilibrium density $\varrho_V$, see
\eqref{defgammagen}. At this stage the analyticity of $V$ is necessary
since we make use of the loop equation from Johansson \cite{Joh} and Shcherbina \cite{Sch}.

\medskip
\noindent
{\bf Step 2.  Uniqueness of local Gibbs measures with logarithmic interactions.}
With the precision of eigenvalue location estimates from the Step 1 as an input,
the eigenvalue spacing distributions are
 shown to be  given by the corresponding 
Gaussian ones. (We will take the uniqueness of the spacing
distributions as our definition of the uniqueness of Gibbs state.)

\medskip 

There are several similarities and differences  between the proofs of Theorem~\ref{bulkWigner} and \ref{bulkbeta}.
Both start with rigidity estimates on eigenvalues and then
 establish that the local spacing distributions are the same
 as in the Gaussian cases. 
The Gaussian divisible ensembles, which play a  key role in our theory for noninvariant ensembles, 
are  completely absent for invariant ensembles.
The key connection between the two methods, however,  
is the usage of DBM (or its analogue) in the Steps 2. In Section~\ref{sec:DBM}, we will 
first  present this idea.

\subsubsection{Strategy for gap universality}

The proofs of Theorems~\ref{thm:sg} and \ref{thm:beta} require several new ideas. 
The focus is to analyze the local conditional measures $\mu_\by$ and $f_{t,\by}\mu_\by$
 instead of the equilibrium measure $\mu$ and the DBM evolved measure $f_t\mu$.
They are obtained by fixing all but $\cK$ consecutive points, denoted by $\by$.
The local measures are Gibbs measures on $\cK$ points, denoted by $\bx$, that are confined
to an interval $J=J_\by$ determined by the boundary points of $\by$.
The external potential, $V_\by$, of the local measure contains not only the external potential $V$ from $\mu$,
but also the interactions between $\bx$ and $\by$. 

The first step is again to establish rigidity, but this time with respect to 
the conditional measures $\mu_\by$ and  $(f_t\mu)_\by=:f_{t,\by}\mu_\by$,
at least for most boundary conditions $\by$.
Due to the logarithmic interactions, $V_\by$ is not
analytic any more and the loop equation is not available, but the rigidity information can
still be extracted from the rigidity with respect to the global measure with some
additional arguments.

In the second step, which is the key part of the argument, we establish the universality of the gap distribution
w.r.t $\mu_\by$ by interpolating between $\mu_\by$ and $\mu_{\wt\by}$ with two different boundary conditions $\by$ and $\wt\by$.
This amounts to estimating the correlation between a gap observable, say $O(x_i-x_{i+1})$, and
$V_\by-V_{\wt\by}$. The correlation between particles in log-gases  decay
only logarithmically, i.e.  extremely 
slowly:
\be\label{corrr}
  \frac{\langle x_i ; x_j \rangle}{\sqrt{\langle x_i ; x_i \rangle\langle x_j ; x_j \rangle}} \sim \frac{1}{\log |i-j|}
\ee
at least if $i,j$ are far from the boundaries. Here $\langle\cdot \; ; \; \cdot \rangle$ denotes the
covariance with respect to $\mu_\by$.  The key observation is that correlation between a {\it gap} $x_i-x_{i+1}$
and a particle $x_j$
decays much faster
\be
  \frac{\langle x_i -x_{i+1} ; x_j \rangle}{\sqrt{\langle x_i -x_{i+1}; x_i-x_{i+1}
 \rangle\langle x_j ; x_j \rangle}} \sim \frac{1}{|i-j|},
\label{gappart}\ee
because it is essentially the derivative in $i$ of \eqref{corrr}. The decay of the gap-gap correlation
is even faster. 

While the formulas  \eqref{corrr}--\eqref{gappart} are plausible, their rigorous proof
is extremely difficult due to the very strong correlations in $\mu_\by$. We are able
to prove a much weaker version of \eqref{gappart}, practically a decay of
order $|i-j|^{-\e}$ for some small $\e>0$, which is sufficient for our purposes.
Even the proof of this weaker decay requires quite heavy tools.

We start with a classical observation by Helffer and Sj\"ostrand
\cite{HS} that the
covariance of any two observables $f,g$ with respect to
 a Gibbs measure $\mu = \exp(-\cH(\bx))\rd \bx$ can
be expressed as
\be\label{HS}
   \langle f(\bx); g(\bx)\rangle_\mu = \int_0^\infty \langle \bh_t(\bx), \nabla g(\bx)\rangle_\mu \rd t,
  \qquad \partial_t \bh_t = -(\cL + \cH'')\bh_t, \quad \bh_0= \nabla f, 
\ee
where $\cL\ge0$ is the generator to the Dirichlet form $\int |\nabla f|^2\rd\mu$
and $\cH''$ is the Hessian of the Hamiltonian.
The generator $\cL$
in the heat equation in \eqref{HS} creates a time dependent random environment
$\bx (t)$ that makes the matrix entries $(x_i-x_j)^{-2}$ of $\cH''$  time dependent. 
The solution $h_t$ to the equation in \eqref{HS} can  be thus
 represented as a random walk in a time dependent random environment, where the
jump rate from site $i$ to $j$ is given by  $(x_i(t)-x_j(t))^{-2}$
at time $t$. On large scales and for typical realizations of $\bx(t)$,
this jump rate is close to a discretization of the $\sqrt{-\Delta}$ operator.
A  discrete version of Di Giorgi-Nash-Moser partial regularity theory 
\cite{C} then guarantees
that the neighboring components of $\bh_t$ are close, which renders the
covariance $\langle \bh_t(\bx), \nabla g(\bx)\rangle_\mu$ small, assuming that
$g$ is a function of $x_i-x_{i+1}$. In more general terms, the 
correlation decay \eqref{gappart} with $|i-j|^{-\e}$ is equivalent to 
the H\"older regularity a discrete parabolic PDE with random coefficients. 
This approach has a considerable potential to study
log-gases since it connects the problem with one of the deepest
phenomena in PDE.

Finally, in the third step, we pass the information on the universality of the gap w.r.t. 
local measures to the global ones. For the invariant ensemble this step is fairly straighforward,
while for the Wigner ensemble we need to use an approximation step similar to Step 3
in Section~\ref{sec:str1}.

\subsubsection{Historical remarks}\label{sec:hist}

The method of the proof of Theorem \ref{bulkWigner} is extremely general
and the result holds for a much larger class of 
 matrix ensembles with independent
entries. Adjacency matrices of the Erd{\H o}s-R\'enyi graphs
are also covered as long as the matrix is not too sparse,
namely more than $N^{2/3}$ entries of each row are non-zero
on average \cite{EKYY1, EKYY2}.
Although Theorem \ref{bulkWigner} in its current form was proved in \cite{EKYY2},
the key ideas have been developed through several important 
steps in \cite{EPRSY, ESY4, EYY, EYY2, EYYrigi}. 
 In particular, the Wigner-Dyson-Gaudin-Mehta (WDGM)  conjecture
for complex Hermitian matrices  in the form of \eqref{pointwise} was first proved in Theorem 1.1 of \cite{EPRSY}.
This result holds  whenever the distributions of the matrix elements are smooth.   The smoothness requirement 
for \eqref{pointwise} was  
 partially removed in \cite{TV} and completely removed in 
\cite{ERSTVY} but only in the averaged convergence sense \eqref{avg}.
 For a general distribution \eqref{pointwise} was proved in Theorem 5 in \cite{TV5}. 
Although the proof in \cite{TV5} took   a  slightly different path, 
this generalization  is an immediate  corollary  of  previous 
results \cite{EY}.
These arguments are restricted to the complex Hermitian case since they still use some explicit formula.

The  WDGM conjecture 
for real symmetric matrices in the averaged form of \eqref{avg} was resolved  in \cite{ESY4}
(a special case, under a restrictive third moment matching condition, was treated in \cite{TV}).
 In \cite{ESY4},  
a novel idea based on Dyson Brownian motion was introduced. The most difficult case, 
the real symmetric Bernoulli matrices, was solved in \cite{EYY2}, where a ``Fluctuation Averaging 
 Lemma" (Theorem \ref{thm: averaging} of the current paper)  
exploiting cancellation of matrix elements of the Green function was first introduced.
A more detailed historical review on  Theorem \ref{bulkWigner}  was given in Section 11 of \cite{EYBull}.

For $\beta=2$, Theorem~\ref{bulkbeta} was proved for very general potentials, the best results
for $\beta=1,4$ \cite{DG, KS, Sch} are still restricted to analytic $V$ with additional conditions. 
Prior to Theorem~\ref{bulkbeta} there was no result for general $\beta$, except for the Gaussian case \cite{VV}.

Given the historical importance of the Wigner surmise, it is somewhat surprising
that single gap universality did not receive much attention until very recently. This
is  probably because our understanding of the Wigner-Dyson-Gaudin-Mehta
 universality became sufficiently sophisticated only in the last few years
to realize the subtle difference between {\it fixed energy} and {\it fixed label}
universality. In fact,  even the GUE case was not known
until the very recent paper by Tao \cite{Taogap}. In this work, the complex  Hermitian Wigner
case was also covered under the condition that the distribution
matches that of the GUE to fourth order. Theorem~\ref{thm:sg} is considerable
more general, as it applies to any symmetry classes and does not require moment matching.
Finally, the single gap universality of the invariant
ensembles has not been considered before Theorem~\ref{thm:beta}.

\subsubsection{What will not be discussed}

In these lecture notes we focus on the four universality results, Theorem~\ref{bulkWigner}--\ref{thm:beta},
and the necessary background material. There are many related questions on random matrix universality and
several of them can be studied with the methods we present here. Here we just list them
and give a few relevant references.

\begin{itemize}
\item  Edge universality for Wigner matrices.
See  Section 9 of \cite{EYBull} for a summary and also the recent paper \cite{LY} that
gives the  the optimal moment condition.
\item Universality of eigenvectors. See \cite{KY}.
\item Universality for sample covariance matrices. See \cite{ESYY, PY1, PY2}.
\item Sparse matrices and adjancency matrices of Erd{\H o}s-R\'enyi graphs. See Section 10  of \cite{EYBull}.
\end{itemize}

\subsubsection{Structure of the lecture notes}

A large part of presentation in these lecture notes is borrowed from other papers and
reviews written on the subject \cite{Etucs, EYBull, EKY3} and sometimes whole paragraphs
of the original articles are verbatim taken over. The overlap is especially large
with the review paper \cite{EYBull}; Sections~\ref{sec:DBM}--\ref{sec:loceq} on the Dyson Brownian
motion, the Green function comparison theorem and on the analysis of the $\beta$-ensemble are repeated
 without much changes.
The local semicircle law (Section~\ref{sec:lsc}) 
is presented here more generally than in \cite{EYBull}, following the recent paper \cite{EKY3}. 
For pedagogical reasons,  we will give  the proof in a simplified form
in Section~\ref{sec:nogap} and we only comment on the general proof in Section~\ref{sec:withgap}.
Sections~\ref{sec:mat}--\ref{sec:appl}  cover new results on random band matrices based upon the recent work \cite{EKYY3}. 
Section~\ref{sec:sg} presents an extensive outline of  the proofs of Theorems~\ref{thm:sg} and
\ref{thm:beta} on the single gap universality following the very recent paper \cite{EYsinglegap}.

\smallskip

We will use the convention that $C$ and $c$ denote generic positive constants whose
actual values are irrelevant and may change from line to line. For two $N$-dependent
quantities $A_N$ and $B_N$ we use the notation $A_N\asymp B_N$ to express that
$c\le A_N/B_N\le C$.
\medskip

{\it Acknowledgement.}
 The results in these lecture notes were  obtained in collaboration
with Horng-Tzer Yau, Benjamin Schlein, Jun Yin, Antti Knowles and Paul Bourgade and in some work, also with
Jose Ramirez and Sandrine P\'eche.
 This article  reports the joint progress 
with these authors.

\section{Local semicircle law for general Wigner-type matrices}\label{sec:lsc}

\subsection{Setup and the main results}

Let $(h_{ij} \, : \, i \leq j)$ be a family of independent, complex-valued random variables satisfying
$\E h_{ij} = 0$  and $h_{ii} \in \R$ for all $i$.
For $i > j$ we define $h_{ij} \deq \bar h_{ji}$, and denote by $H = (h_{ij})_{i,j = 1}^N$
 the $N \times N$ matrix with entries $h_{ij}$. By definition, $H$ is Hermitian: $H = H^*$.
 (Note that this setup also includes the case of a real symmetric matrix $H$.)
Such ensembles will be called {\it general Wigner-type matrices}. Note that we allow for the matrix elements
having different distributions. This class of matrices is a natural generalization of the
{\it standard real symmetric Wigner matrices} for which $h_{ij}\in R$ are identical distributed,
and the {\it standard complex Hermitian Wigner matrices} for which the off-diagonal
elements $h_{ij}\in \C$ are identically distributed and the diagonal elements
$h_{ii}\in \R$ have their own, but still identical distribution.

The fundamental data of the model is the $N\times N$ matrix of variances $S=(s_{ij})$, where
$$
   s_{ij} : = \E \, |h_{ij}|^2.
$$
We introduce the parameter $M \;\deq\; \big[\max_{i, j} s_{ij}\big]^{-1}$
that expresses the maximal size of $s_{ij}$:
\begin{equation} \label{s leq W}
s_{ij} \;\leq\; M^{-1}
\end{equation}
for all $i$ and $j$. We regard $N$ as the fundamental parameter and $M=M_N$ as a function of $N$.
\begin{equation} \label{lower bound on W}
N^\delta \;\leq\; M \;\leq\; N
\end{equation}
for some fixed $\delta > 0$. 
 We assume that $S$ is (doubly) stochastic:
\begin{equation} \label{S is stochastic}
\sum_j s_{ij} \;=\; 1
\end{equation}
for all $i$. 
For standard Wigner matrices, $h_{ij}$ are identically distributied, hence
 $s_{ij}=\frac{1}{N}$ and  $M=N$. In this presentation,  we allow for the matrix elements
having different distributions but independence (up to the Hermitian symmetry)
is always assumed.

\begin{example}\label{ex}  Random band matrices are characterized by
translation invariant variances of the form
\be\label{bandvar}
   s_{ij} = \frac{1}{W} f\Big(\frac{|i-j|_N}{W}\Big)
\ee
where $f$ is a smooth, symmetric probability density on $\bR$, $W$ is a large parameter,
called the {\it band width}, 
and $|i-j|_N$ denotes the periodic distance on the discrete torus $\bb T$ of length $N$.
The generalization in higher spatial dimensions is straighforward, in this case
the rows and columns of $H$ are labelled by a discrete $d$ dimensional torus ${\bb T}_L^d$ of length $L$
with $N=L^d$. 
\end{example}

For convenience we assume that the normalized entries 
\be\label{def:zeta}
\zeta_{ij}: = s_{ij}^{-1/2}h_{ij}
\ee
 have a polynomial
decay of arbitrary high degree, i.e. for all $p \in \N$ there is a constant $\mu_p$ such that
\begin{equation} \label{polydecay}
\E \abs{\zeta_{ij}}^p \;\leq\; \mu_p
\end{equation}
for all $N$, $i$, and $j$. We make this assumption to streamline notation, but in fact, 
 our results hold, with the same proof, provided \eqref{polydecay} is valid for some large but fixed $p$.
 If we  strengthen it  to uniform subexponential decay, \eqref{subexp},
then certain estimates will become stronger.
In this paper we  work with \eqref{polydecay}  for simplicity, but we remark that
most of our previous work used \eqref{subexp}.

Throughout the following we use a spectral parameter $z \in \C$ satisfying $\im z > 0$. We shall use the notation
\begin{equation*}
z \;=\; E + \ii \eta\
\end{equation*}
without further comment. The eigenvalues of $H$ in the $N\to \infty$ limit are distributed
by the celebrated Wigner semicircle law,
\be
\varrho(x)= \varrho_{sc}(x) \;\deq\; \frac{1}{2 \pi} \sqrt{(4 - x^2)_+}\,.
\ee
and its Stieltjes transform at spectral parameter $z$ is defined by
\begin{equation} \label{definition of msc}
m(z) \;\deq\; \int_\bR \frac{\varrho(x)}{x - z} \, \rd x\,.
\end{equation}
To avoid confusion, we remark that $m$ was denoted by $m_{sc}$ and $\varrho$ by $\varrho_{sc}$ in most of our previous papers.
In this section we drop the subscript referring to ``semicircle''.
 It is well known that the Stieltjes transform $m$ is the unique solution of
\begin{equation} \label{identity for msc}
m(z) + \frac{1}{m(z)} + z \;=\; 0
\end{equation}
with $\im m(z) > 0$ for $\im z > 0$. Thus we have
\begin{equation} \label{explicit m}
m(z) \;=\; \frac{-z + \sqrt{z^2 - 4}}{2}.
\end{equation}

We define the \emph{resolvent} of $H$ through
\begin{equation*}
G(z) \;\deq\; (H - z)^{-1}\,,
\end{equation*}
and denote its entries by $G_{ij}(z)$.
The Stieltjes transform of the empirical spectral measure 
$$
\varrho_N(\rd x)=\frac{1}{N}\sum_\al \delta(\lambda_\al-x)\rd x
$$
for the eigenvalues $\lambda_1\le \lambda_2\le\ldots \le \lambda_N$ of $H$ is
\begin{equation}\label{mNdef}
m_N(z) \;\deq\; \int_\bR \frac{\varrho_N(\rd x)}{x - z} = \frac{1}{N} \tr G(z)\,.
\end{equation}

An important parameter of the model is 
\begin{equation}\label{def:rho}
\Gamma(z) \; \deq \; \Big\| \frac{1}{1-m^2(z) S} \Big\|_{\infty\to\infty}.
\end{equation}
Note that $S$, being a stochastic matrix, satisfies $-1\le S\le 1$, and 1 is
an eigenvalue with eigenvector ${\bf e} = N^{-1/2}(1,1,\ldots 1)$, $S{\bf e} = {\bf e}$.
We assume that $1$ is simple for convenience.
Another important parameter is
\begin{equation}\label{def:rhohat}
\wt\Gamma(z) \; \deq \; \Bigg\| \frac{1}{1-m^2(z) S} \bigg|_{{\bf e}^\perp }\Bigg\|_{\infty\to\infty},
\end{equation}
i.e.\ the norm of $1-m^2S$ restricted to the subspace orthogonal to the constants.
Clearly $\wt\Gamma \le \Gamma$.

For standard Wigner matrices we easily obtain that
\be\label{Gammawigner}
   \Gamma(z) = \frac{1}{|1-m^2(z)|} \asymp \frac{1}{\sqrt{ \kappa_E +\eta}}, 
\qquad \wt\Gamma(z) =1,
\ee
where $\kappa_E:= \big| |E|-2\big|$ denotes the distance of $E$ to the
spectral edges.
For generalized Wigner matrices (Definition~\ref{def:genwig})
essentially the same relations hold:
\be\label{Gammawigner1}
   \Gamma(z) = \frac{1}{|1-m^2(z)|} \asymp \frac{1}{\sqrt{ \kappa_E +\eta}}, 
\qquad \wt\Gamma(z) \asymp 1.
\ee

The following definition introduces a notion of a high-probability bound that is suited for our purposes.

\begin{definition}[Stochastic domination]\label{def:stocdom}
Let
\begin{equation*}
X = \pb{X^{(N)}(u) \, : \, N \in \N, u \in U^{(N)}} \,, \qquad
Y = \pb{Y^{(N)}(u) \, : \, N \in \N, u \in U^{(N)}}
\end{equation*}
be two families of nonnegative random variables, where $U^{(N)}$ is a possibly $N$-dependent parameter set. 
We say that $X$ is \emph{stochastically dominated by $Y$, uniformly in $u$,} if for all (small)
 $\e > 0$ and (large) $D > 0$ we have
\begin{equation*}
\sup_{u \in U^{(N)}} \P \qB{X^{(N)}(u) > N^\e Y^{(N)}(u)} \;\leq\; N^{-D}
\end{equation*}
for large enough $N\ge N_0(\e, D)$. Unless stated otherwise, 
throughout this paper the stochastic 
domination will always be uniform in all parameters apart from the parameter $\delta$ in \eqref{lower bound on W}
 and the sequence of constants $\mu_p$ in \eqref{polydecay}; thus, $N_0(\e, D)$ also depends on $\delta$ and $\mu_p$.
If $X$ is stochastically dominated by $Y$, uniformly in $u$, we use the notation $X \prec Y$. Moreover, if
 for some complex family $X$ we have $\abs{X} \prec Y$ we also write $X = O_\prec(Y)$.
\end{definition}

For example, using Chebyshev's inequality and \eqref{polydecay} one easily finds that 
\begin{equation}\label{hsmallerW}
h_{ij} \;\prec\; (s_{ij})^{1/2} \;\prec\; M^{-1/2}, 
\end{equation}
uniformly in $i$ and $j$, so that we may also write $h_{ij} = O_\prec((s_{ij})^{1/2})$.
An easy exercise shows that
the relation $\prec$ satisfies the familiar algebraic rules of order relations, e.g.
such relations can be added and multiplied.
The definition of $\prec$ with the polynomial factors $N^{-\e}$ and $N^{-D}$ are taylored 
for the assumption \eqref{polydecay}. We remark that if \eqref{subexp} is assumed, a stronger
form of stochastic domination can be introduced but we will not pursue this direction here.

Since 
$$
  \lim_{\eta\to 0^+} \im m(E+i\eta) = \pi \, \varrho(E),
$$
the convergence of the Stieltjes transform $m_N(z)$ to $m(z)$ as $N\to \infty$ will show that the
empirical local density of the eigenvalues around the energy $E$ in a window of 
size $\eta$ converges to the semicircle law $\varrho(E)$. Therefore the key
task is to control $m_N(z)$ for small $\eta=\im z$. 

 We now define
the lower threshold for $\eta$ that depends on the energy  $E \in [-10,10]$
\begin{equation} \label{def wt eta E}
\wt \eta_E \;\deq\; \min \hBB{\eta \;\col\; \frac{1}{M \eta} \leq
 \min \hbb{\frac{M^{-\gamma}}{\wt \Gamma(z)^3} \,,\, 
\frac{M^{-2 \gamma}}{\wt \Gamma(z)^4 \im m(z)}} \text{ for all }
 z \in [E + \ii \eta, E + 10 \ii] }\,.
\end{equation}
Here $\gamma>0$ is a parameter that can be chosen arbitrarily small;
for all practical purposes the reader can neglect it.
For generalized Wigner matrices,  $M\asymp N$, from \eqref{Gammawigner1} we have 
$$
   \wt\eta_E \le CN^{-1+2\gamma},
$$
i.e. we will get the local semicircle law on the smallest possible
 scale $\eta\gg N^{-1}$, modulo a 
polynomial correction with an arbitrary small exponent.
We remark that if we assume subexponential decay \eqref{subexp} instead
of the polynomial decay \eqref{polydecay}, then the small
polynomial correction can be replaced with a logarithmic correction factor.

Finally we define  our fundamental control parameter
\begin{equation}\label{def:Pi}
\Pi(z) \;\deq\; \sqrt{\frac{\im m(z)}{M \eta}} + \frac{1}{M \eta}\,.
\end{equation}

We can now state the main result of this section, which in this form
appeared in \cite{EKY3}.

\begin{theorem}[Local semicircle law] \label{thm: with gap}
Uniformly  in the energy $|E|\le 10$ and  $\eta\in [\wt\eta_E, 10]$
we have the bounds
\begin{equation}\label{Gijest 2}
\absb{G_{ij}(z) - \delta_{ij} m(z)} \;\prec\; \Pi(z)= \sqrt{\frac{\im m(z)}{M \eta}} + \frac{1}{M \eta}, \qquad z=E+\ii \eta
\end{equation}
uniformly in $i,j$,  as well as
\begin{equation}\label{m-mest 2}
\absb{m_N(z) - m(z)} \;\prec\; \frac{1}{M \eta}.
\end{equation}
\end{theorem}

We point out two remarkable features of these bounds.
The error term for the resolvent entries behaves essentially as $(M\eta)^{-1/2}$,
with an improvement near the edges where $\im m$ vanishes. 
 The error bound for the Stieltjes
transform, i.e. for the average of the diagonal resolvent entries, is one order better,
$(M\eta)^{-1}$, but without improvement near the edge. 

Various local semicircle laws have a long history.
For standard Wigner matrices (i.e. $M=N$ and $\wt\Gamma =1$),
 the optimal threshold for the smallest possible $\eta\gg 1/N$ has first been achieved
in \cite{ESY2} in the bulk after an intermediate result on scale $\eta\gg N^{-2/3}$ in \cite{ESY1}.
The first effective result near the edge was given in \cite{ESY3}. 
The optimal power $(N\eta)^{-1}$ for $m_N-m$ in the bulk has first been obtained in
\cite{EYY2} where the first version of the {\it fluctuation averaging mechanism} has appeared.
The optimal behavior near the edge was first derived in \cite{EYYrigi}.
The case of $M\ll N$ has first been studied in \cite{EYY} where the threshold $\eta\gg 1/M$ in the bulk 
spectrum, $|E|<2$, has been
achieved.  The optimal power $(M\eta)^{-1}$ is proved in \cite{EKY3}, where the technique
of \cite{EYY} was combined with the fluctuation averaging mechanism.
 The edge behavior, i.e. the
deterioration of the threshold $\wt\eta_E$ near the edge has also been extensively studied in \cite{EKY3}
and it led to the power ${-3}$ of $\wt \Gamma$ in the definition of $\wt\eta_E$,
but it is yet unclear whether this power is optimal.

In Section~\ref{sec:three} we demonstrate that all proofs of the local semicircle
law rely on some version of a self-consistent equation. At the beginning this was a scalar
equation for $m_N$. The  self-consistent vector equation for $v_i = G_{ii}-m$ (see Section~\ref{sec:vec}) 
first appeared in  \cite{EYY}. This allowed us
to deviate from the identical distributions for $h_{ij}$ and opened up the route
to estimates on individual resolvent matrix elements. Finally, the
self-consistent matrix equation for $\E |G_{xy}|^2$ first appeared in \cite{EKYY3}
and it yielded diffusion profile for the resolvent (Section~\ref{sec:appl}).

\medskip

We now list a few consequences of the local semicircle law.
It is an elementary property of the Stieltjes transform that
once $m_N \approx m$ is established for all spectral parameter $z$ with $\im z\ge \wt\eta_E$, then
 $\varrho_N$ and $\varrho$ coincide on scales larger than $\wt\eta_E$ 
near the energy $E$. This  means that
$\int f\varrho_N \to\int f\varrho$ for test functions on scale $\wt\eta_E$
i.e. $|f'|\ll 1/\wt\eta_E$.  In particular, if $m_N(z)\to m(z)$
uniformly on the half planes $\{ z \, : \; \im z\ge \e\}$ for
any fixed $\e$, then $\varrho_N$ converges to $\varrho$ weakly.
The bound \eqref{m-mest 2} asserts much more: it identifies $\varrho_N$
with $\varrho$ on scales of order $1/M$. In the standard Wigner case, $M=N$,
it is basically the optimal scale since below scale $1/N$ the empirical measure
$\varrho_N$ strongly fluctuates due to individual eigenvalues. 

Once the local density is identified, we can deduce results 
on the location of individual eigenvalues and on the counting function.
Here we formulate the corresponding statements only for the simpler case when $s_{ij}$ is comparable with
$1/N$ as they were stated in \cite{EYYrigi} (apart from the fact that in \cite{EYYrigi}
a subexponential decay was assumed).
  The precise results in the general case are somewhat more complicated and they 
can be found in \cite{EKY3}. 

 Let $\gamma_\al =\gamma_{\al,N}$ denote  the location of the $\al$-th point
under the semicircle law, i.e., $\gamma_\al$ is defined by
\be\label{def:gamma}
 N \int_{-\infty}^{\gamma_\al} \varrho(x) \rd x = \al, \qquad 1\leq \al\le N. 
\ee
We will call $\gamma_\al$ the {\it classical location} of the $\al$-th point.
Furthermore, for any real energy $E$, let
$$
   \fn_N(E) := \frac{1}{N}\, \# \big\{ \lambda_\al\le E\big\}, \qquad
n(E): = \int_{-\infty}^E \varrho(x)\rd x
$$
be the {\it empirical counting function} of the eigenvalues and its classical counterpart.

\begin{corollary}[Rigidity of eigenvalues and limit of the counting function]\cite[Theorem 2.2]{EYYrigi}
\label{cor:rig} For generalized Wigner matrices (Definition~\ref{def:genwig}) we have
\be\label{rigwig}
   |\lambda_\al- \gamma_\al|\prec \frac{1}{N^{2/3} \wh\al^{1/3}}, \qquad \wh\al: = \min\{ \al, N+1-\al\},
\ee
 uniformly in $\al\in\{1, 2,\ldots, N\}$.
Furthermore,
\be\label{count}
   |\fn_N(E)-n(E)|\prec \frac{1}{N}.
\ee
uniformly in $E\in \bR$.
\end{corollary}
We remark that under the stronger decay condition \eqref{subexp} 
instead of \eqref{polydecay}, the $N^\e$ factors implicitly present in the
notation $\prec$ can be improved to logarithmic factors, see \cite{EYYrigi}.

Corollary~\ref{cor:rig}
is a simple consequence of the Helffer-Sj\"ostrand formula which 
  translates information on the  Stieltjes
transform of the empirical measure first to the counting function and
then to the locations of eigenvalues. 
The formula yields the representation 
\be
   f(\lambda) =\frac{1}{2\pi}\int_{\bR^2}
\frac{\partial_{\bar z} \wt f(x+iy)}{\lambda-x-iy} \rd x  \rd y 
=\frac{1}{2\pi}\int_{\bR^2}
\frac{iy f''(x)\chi(y) +i(f(x) + iyf'(x))\chi'(y) }{\lambda-x-iy} \rd x 
\rd y
\label{1*}
\ee
for any real valued $C^2$ function $f$
on $\R$, where $\chi(y)$ is any smooth cutoff function with bounded
derivatives and   
 supported in $[-1,1]$ with $\chi(y)=1$ for  $|y|\leq 1/2$. 
In the applications, $f$ will be a smoothed version of the
characteristic functions of spectral intervals so that $\sum_j f(\lambda_j)$
 counts eigenvalues in that interval. From \eqref{1*} we have
$$
   \frac{1}{N}\sum_j f(\lambda_j) = \frac{1}{2\pi}\int_{\bR^2}
\Big(iy f''(x)\chi(y) +i(f(x) + iyf'(x))\chi'(y) \Big)  m_N(x+iy) \rd x 
\rd y,
$$
and then $m_N$ can be approximated by $m(x+iy)$.
The details of the argument can be found in  \cite{ERSY}.
Once \eqref{count} is established, it is an elementary argument to
translate it into the rigidity of the eigenvalues \eqref{rigwig}. The
powers $2/3$ and $1/3$  in \eqref{rigwig} stem from the fact that $\varrho(x)$ has
a square root singularity near the spectral edges $x\approx \pm 2$,
therefore $n(x) \sim (x+2)^{3/2}_+$ for $x$ near $-2$.

Although Wigner's semicircle law and its local version only  concerns $m_N = \frac{1}{N}\tr G$, 
we remark that the resolvent matrix elements, $G_{ii}$ and $G_{ij}$ also carry important information.
For example a good bound on  $G_{ii}$ implies delocalization of the eigenvectors.
Indeed, by the spectral decomposition, we have
$$
    \im G_{ii}(z) = \eta \sum_\al \frac{|u_\al(i)|^2}{(\lambda_\al-E)^2 + \eta^2},
$$
where $\bu_\al = (u_\al(1), \ldots  u_\al(N))$ is the (normalized) eigenvector
belonging to the eigenvalue $\lambda_\al$. Choosing the energy $E$ in the $\eta$-vicinity of $\lambda_\al$, we
obtain $|u_\al(i)|^2 \le \eta  \im G_{ii}(z)$. Therefore, if $|G_{ii}(z)|\le C$ can be shown uniformly for
any $z$ with $\im z\ge \eta(N)$ for some $N$-dependent threshold $\eta(N)$, then we conclude
$$
 \max_\al \|\bu_\al\|_\infty^2 \le C\eta(N).
$$
In the Wigner case, the threshold $\eta(N)$ is almost $1/N$, thus we obtain the complete delocalization
of the eigenvectors:

\begin{corollary}
For the $\ell^2$-normalized eigenvectors $\bu_\al$, $\al=1,2\ldots, N$,
of the standard Wigner matrix, we have
\be\label{delocbound}
    \| \bu_\al\|_\infty \prec N^{-1/2}.
\ee
\end{corollary}
This result was first proven in \cite{ESY2} without bounding $G_{ii}$. In contrast to the argument
in \cite{ESY2}, the 
proof  via $G_{ii}$ can also be easily  extended to a general class of Wigner-type
 matrices. For example, an elementary argument shows that if there are two
positive constants $c$ and $C$ such that $c\le Ns_{ij}\le C$ for all $i,j$, then
$\wt\Gamma(z)\le C$ uniformly in $z$, thus $\wt\eta_E\le CN^{-1+\gamma}$ and
\eqref{delocbound} holds.

\subsection{Tools}

In this subsection we collect some basic definitions and facts.

\begin{definition}[Minors] \label{def: minors}
For $ {\bb T} \subset \{1, \dots, N\}$ we define $H^{( {\bb T})}$ by
\begin{equation*}
(H^{( {\bb T})})_{ij} \;\deq\; \ind{i \notin  {\bb T}} \ind{j \notin  {\bb T}} h_{ij}\,.
\end{equation*}
Moreover, we define the resolvent of $H^{( {\bb T})}$ and its normalized trace through
\begin{equation*}
G^{( {\bb T})}_{ij}(z) \;\deq\;  (H^{( {\bb T})} - z)^{-1}_{ij}\,, \qquad
m^{(\bb T)}(z) \;\deq\; \frac{1}{N}\tr G^{(\bb T)}(z).
\end{equation*}
We also set
\begin{equation*}
\sum_i^{( {\bb T})} \;\deq\; \sum_{i \; : \; i \notin  {\bb T}}\,.
\end{equation*}
\end{definition}

\begin{definition}[Partial expectation and independence] \label{definition: P Q}
Let $X \equiv X(H)$ be a random variable. For $i \in \{1, \dots, N\}$ define the operations $P_i$ and $Q_i$ through
\begin{equation*}
P_i X \;\deq\; \E(X | H^{(i)}) \,, \qquad Q_i X \;\deq\; X - P_i X\,.
\end{equation*}
We call $P_i$ \emph{partial expectation} in the index $i$.
Moreover, we say that $X$ is \emph{independent of ${\bb T} \subset \{1, \dots, N\}$} if $X = P_i X$ for all $i \in {\bb T}$.
\end{definition}

We shall frequently make use of Schur's well-known complement formula, which we write as
\begin{equation} \label{schur}
\frac{1}{G_{ii}^{({\bb T})}} \;=\; h_{ii} - z - \sum_{k,l}^{({\bb T} i)} h_{ik} G_{kl}^{({\bb T} i)} h_{li}\,,
\end{equation}
where $i \notin {\bb T} \subset \{1, \dots, N\}$.

The following resolvent identities form the backbone of all of our calculations. 
The idea behind them is that a resolvent  matrix element $G_{ij}$
depends strongly on the $i$-th and $j$-th columns of $H$, but weakly on all other columns. 
The first identity determines how to make
a resolvent matrix element $G_{ij}$ independent of an additional index $k \neq i,j$.
The second identity expresses the dependence of a resolvent matrix element $G_{ij}$ on
 the matrix elements in the $i$-th or in the $j$-th column of $H$.
We added a third identity that relates sums of off-diagonal resolvent entries with a diagonal one.
 The proofs are elementary.

\begin{lemma}[Resolvent identities] \label{lemma: res id}
For any real or complex Hermitian matrix $H$ and ${\bb T} \subset \{1, \dots, N\}$ the following identities hold.
If $i,j,k \notin {\bb T}$ and $i,j \neq k$ then
\begin{equation} \label{resolvent expansion type 1}
G_{ij}^{({\bb T})} \;=\; G_{ij}^{({\bb T}k)} + \frac{G_{ik}^{({\bb T})} G_{kj}^{({\bb T})}}{G_{kk}^{({\bb T})}}\, ,
\qquad \frac{1}{G_{ii}^{(\bb T)}} \;=\; \frac{1}{G_{ii}^{(\bb Tk)}} 
- \frac{G_{ik}^{(\bb T)} G_{ki}^{(\bb T)}}{G_{ii}^{(\bb T)} G_{ii}^{(\bb Tk)} G_{kk}^{(\bb T)}}\,.
\end{equation}
If $i,j \notin {\bb T}$ satisfy $i \neq j$ then
\begin{equation} \label{resolvent expansion type 2}
G_{ij}^{({\bb T})} \;=\; - G_{ii}^{({\bb T})} \sum_{k}^{({\bb T}i)} h_{ik} G_{kj}^{({\bb T}i)} \;=\; - G_{jj}^{({\bb T})} \sum_k^{({\bb T}j)} G_{ik}^{({\bb T}j)} h_{k j}\,. 
\end{equation}
Moreover, we have 
\be\label{ward}
    \sum_j |G_{ij}^{({\bb T})}|^2 = \frac{1}{\eta}\im G_{ii}^{({\bb T})},
\ee
which is sometimes called the Ward identity.
\end{lemma}

Finally, in order to estimate large sums of independent random variables as in \eqref{schur} and 
\eqref{resolvent expansion type 2}, we will need a large deviation estimate for linear
and quadratic functionals of independent
random variables:

\begin{theorem}[Large deviation bounds] \label{thm: LDE}

Let $\pb{X_i^{(N)}}$ and $\pb{Y_i^{(N)}}$ be independent families of random variables and 
$\pb{a_{ij}^{(N)}}$ and $\pb{b_i^{(N)}}$ be deterministic; here $N \in \N$ and $i,j = 1, \dots, N$. Suppose
 that all entries $X_i^{(N)}$ and $Y_i^{(N)}$ are independent and satisfy
\begin{equation} \label{cond on X}
\E X \;=\; 0\,, \qquad \E \abs{X}^2 \;=\; 1\,, \qquad \p{\E \abs{X}^p}^{1/p} \;\leq\; \mu_p
\end{equation}
for all $p \in \N$ and some constants $\mu_p$.
 Then we have the bounds
\begin{align} \label{lde 1}
\sum_i b_i X_i &\;\prec\; \pbb{\sum_i \abs{b_i}^2}^{1/2}\,,
\\ \label{lde 2}
\sum_{i,j} a_{ij} X_i Y_j &\;\prec\; \pbb{\sum_{i,j} \abs{a_{ij}}^2}^{1/2}\,,
\\ \label{lde 3}
\sum_{i \neq j} a_{ij} X_i X_j &\;\prec\; \pbb{\sum_{i \neq j} \abs{a_{ij}}^2}^{1/2}\,.
\end{align}
\end{theorem}

{\it Sketch of the proof.}
The estimates \eqref{lde 1}, \eqref{lde 2}, and \eqref{lde 3} follow from estimating
high moments of the left hand sides  combined with Chebyshev's inequality. The high moments
of \eqref{lde 1} directly follow from the Marcinkiewicz-Zygmund martingale inequality.
High moments of \eqref{lde 2}, and \eqref{lde 3} are computed by reducing 
them to \eqref{lde 1} with a decoupling argument. The details are found in Lemmas
 B.2, B.3, and B.4 of \cite{EKYfluc}.
\qed

\subsection{Self-consistent equations on three levels}\label{sec:three}

By the Schur complement formula \eqref{schur}, we have
\be\label{schur1}
   G_{ii} \;=\; \frac{1}{h_{ii}- z -  \sum_{k,l}^{(i)} h_{ik} G_{kl}^{(i)} h_{li}}\,.
\ee
The partial expectation with respect to the index $i$ gives
$$
   P_i  \sum_{k,l}^{(i)} h_{ik} G_{kl}^{(i)} h_{li} \;=\; \sum_{k}^{(i)} s_{ik}   G_{kk}^{(i)} 
  \;=\; \sum_{k}^{(i)} s_{ik} G_{kk} +  \sum_{k}^{(i)} s_{ik}  \frac{G_{ik} G_{ki}}{G_{ii}},
$$
where in the second step we 
used \eqref{resolvent expansion type 1}.
Introducing  the notation
\begin{equation*}
v_i \;\deq\; G_{ii} - m
\end{equation*}
and recalling \eqref{S is stochastic}, we get the following
self-consistent equation for $v_i$:
\begin{equation}\label{vself}
v_{i} \;=\; \frac{1}{-z - m - \pb{\sum_k s_{ik} v_k - \Upsilon_i}} - m\,,
\end{equation}
where
\begin{align}\label{def:Ups}
\Upsilon_i \;\deq\; A_i + h_{ii} - Z_i\,, \qquad
A_i \;\deq\; \sum_{k}s_{ik} \frac{G_{ik} G_{ki}}{G_{ii}}\,, \qquad 
Z_i \;\deq\; Q_i \sum_{k,l}^{(i)} h_{ik} G_{kl}^{(i)} h_{li}\,.
\end{align}
All these quantities depend on $z$, which fact is suppressed in the notation.
We will show that $\Upsilon$ is a lower order error term. 
This is clear about $h_{ii}$ by \eqref{hsmallerW}. The term $A_i$ will
be small since off-diagonal resolvent entries are small. Finally, $Z_i$ will
be small by a large deviation estimate Theorem~\ref{thm: LDE}.
Before we present more details, we heuristically show the power
of the self-consistent equations.

\subsubsection{A scalar self-consistent equation}

Introduce the notation
\be\label{ava}
[a]=\frac{1}{N}\sum_i a_i
\ee
for the {\it average} of a vector $(a_i)_{i=1}^N$.
Consider  the standard Wigner case, $s_{ij}=1/N$. Then
$$
  \sum_k s_{ik} v_k = \frac{1}{N}\sum_k v_k =  [v] \qquad  \big( =  m_N-m\big).
$$
Neglecting $\Upsilon_i$ in  \eqref{vself} and taking the average of
this relation for each $i$,
 we get
\be
\label{vv}
  [v] \approx \frac{1}{-z - m - [v] } - m\, .
\ee
Since 
$$
  m = \frac{1}{-z - m} 
$$
by the defining equation of $m$, see \eqref{identity for msc}, and
this equation is stable under small perturbations, at least away from the
spectral edges $z=\pm 2$, we obtain from \eqref{vv} that $[v]\approx 0$.
This means that $m_N\approx m$ and hence
 $\varrho_N \approx \varrho$, i.e.
we obtained the Wigner's original semicircle law.

Historically the semicircle law was first found via the
moment method \cite{W} by computing $\frac{1}{N}\E \tr H^k$, $k=1,2,\ldots$
in the $N\to\infty$ limit,  and identifying
them with the moments of the semicircle measure $\varrho(x)\rd x$. 
In this approach the semicircle law emerges as a result of a somewhat
tedious, albeit elementary calculation. A more direct approach is to take the average in $i$
 of the Schur's formula \eqref{schur1} which immediately gives
$$
  m_N \approx \frac{1}{-z-m_N}
$$
after neglecting the error terms $\Upsilon_i$. This identifies the limit of $m_N$
immediately with $m$, the (unique) solution to  \eqref{identity for msc}.
Taking the inverse Stieltjes transform then yields the semicircle law
in a very direct way.

\subsubsection{A vector self-consistent equation}\label{sec:vec}

If the variances $s_{ij}$ are not constant or we are interested in
individual resolvent matrix elements $G_{ii}$ instead of their
average, $\frac{1}{N} \tr G$, then
the scalar equation \eqref{vv}  discussed in the previous section is not sufficient.
We have to consider \eqref{vself} as a 
 system of equations for the components of the vector $\bv=(v_1, \ldots ,  v_N)$.

{F}rom the explicit formula for $m$ \eqref{explicit m}, we know that $|m+z|\ge 1$.
Assuming temporarily that
\be\label{smallups}
\absbb{\sum_k s_{ik} v_k - \Upsilon_i} \;\le\; \frac{1}{2}\,,
\ee
we can expand the right-hand side of \eqref{vself} around $-z-m$ up to second order and 
using the identity \eqref{identity for msc} we  obtain
\be\label{viexp}
   v_i \;=\; m^2 \pbb{\sum_k s_{ik} v_k - \Upsilon_i} 
+ O\qBB{\pbb{\sum_k s_{ik} v_k - \Upsilon_i}^2}\,.
\ee
This is the key equation to study $v_i$. Considering all higher order terms and $\Upsilon_i$
as errors we get a self-consistent
equation 
$$
    v_i =  m^2 \sum_k s_{ik} v_k + \mbox{Error}
$$ 
or, with matrix notation for the vector  $\bv$:
\be\label{vselfv}
\bv = m^2 S\bv +\cE,
\ee
where $\cE$ represent the vector of error terms. Thus
$$
    \bv = \frac{1}{1-m^2 S} \cE, \qquad \mbox{hence}\quad  \|\bv\|_\infty \le 
\Big\| \frac{1}{1-m^2 S}\Big\|_{\infty\to\infty} \|\cE\|_\infty  = \Gamma  \|\cE\|_\infty,
$$
and this relation shows how the quantity $\Gamma$ emerges. If the error term is indeed small
and $\Gamma$ is bounded, then we obtain that $\|\bv\|_\infty = \max |G_{ii}-m|$ is small.

\subsubsection{A matrix self-consistent equation}\label{sec:mat}

The off-diagonal resolvent matrix elements, $G_{ij}$, are strongly oscillating quantities and they are
not expected to have a deterministic limit. However the local averages of their squares,
$$
   T_{xy} :  =\sum_i s_{xi} |G_{iy}|^2,
$$
are expected to behave regularly. Note that in the Wigner case, using the identity \eqref{ward}, the quantity
$$
 T_{xy}=  \sum_i s_{xi} |G_{iy}|^2 = \frac{1}{N} \sum_i G_{yi} G_{iy}^* =
 \frac{1}{N} \big( |G|^2\big)_{yy} = \frac{1}{N\eta} \im G_{yy}
$$
is independent of $x$, but in the general case $T_{xy}$ carries information on the localization length 
of the eigenfunctions. In particular, if $T_{xy}$ decays only beyond a scale $\ell$, i.e. $T_{xy}$ remains comparable
with $T_{xx}$  for $|x-y|\ll \ell$, then most eigenfunctions have a localization length at least $\ell$.

The self-consistent equation for $T_{xy}$ can be derived from \eqref{resolvent expansion type 2}:
$$
   G_{iy} = - G_{ii} \sum_k^{(i)} h_{ik}G_{ky}^{(i)}.
$$
Replacing $G_{ii}$ with $m$ and taking the square, we have
$$
  |G_{iy}|^2 \approx |m|^2 \sum_{m,k}^{(i)}  h_{ik} G_{ky}^{(i)} \ov{h}_{im}\ov{G_{my}^{(i)}}.
$$
Taking partial expectation yields
$$
   P_i  |G_{iy}|^2 \approx |m|^2 \sum_k^{(i)} s_{ik}|G_{ky}^{(i)}|^2 \approx  |m|^2 \sum_k s_{ik} |G_{ky}|^2  = |m|^2 T_{iy},
$$
where in the second step we used \eqref{resolvent expansion type 1} to remove the upper index $i$
and we used that the off-diagonal elements are of smaller order. This formula holds for $i\ne y$, 
for the special case $i=y$ we have a diagonal element, i.e.
$$
   P_i  |G_{iy}|^2 \approx  |m|^2 T_{iy} + |m|^2 \delta_{iy}.
$$
Thus we have
\be\label{wtT}
   T_{xy} = \sum_i s_{xi} |G_{iy}|^2 = \sum_i s_{xi} P_i |G_{iy}|^2  + \wt T_{xy}, \qquad
 \wt T_{xy}:=\sum_i  s_{xi} Q_i |G_{iy}|^2 .
\ee
The term  $\wt T_{xy}$ is lower order by a {\it fluctation averaging mechanism},
see the explanation after Theorem~\ref{thm: averaging}.
Hence, neglecting this term we have
$$ 
   T_{xy} \approx \sum_i s_{xi} P_i |G_{iy}|^2 = |m|^2 \sum_i s_{xi} T_{iy} +|m|^2 s_{xy},
$$
i.e. the matrix $T$ satisfies the self-consistent matrix equation
\be
    T = |m|^2 ST + |m|^2 S + \cE,
\label{Tself}
\ee
where $\cE$ is an error matrix.  The solution is
\be\label{TT}
   T = \frac{|m|^2 S}{1-|m|^2 S}  + \frac{1}{1-|m|^2 S}\cE,
\ee
where the first term, given explicitly in terms of the variance matrix $S$, gives
the leading order behavior for $T$:
$$
  T_{xy}\approx \Theta_{xy}, \qquad \Theta: = \frac{|m|^2 S}{1-|m|^2 S}.
$$

\subsubsection{Application: diffusion profile for random band matrices}\label{sec:appl}

Depending on the structure of $S$, in some cases $\Theta$ can be computed.
Consider, for example,  case of the random band matrices, \eqref{bandvar}.
Here $s_{xy}$ and 
hence $\Theta_{xy}$ are translation invariant, $\Theta_{xy} = \theta_{x-y}$, and the
Fourier transform of $\theta$ is approximately given by
\be\label{thet}
  \theta(p) \approx \frac{1}{\al \eta + W^2Dp^2}, \qquad D:= \frac{1}{2} \int x^2 f(x)\rd x,
\qquad \al : = \frac{2}{\sqrt{4- E^2}}.
\ee
This means that the profile of $\Theta$ is approximately given by a diffusion profile on scale $W$
with diffusion constant $D$.

The analysis of the error term in \eqref{TT}
requires estimating the norm of $(1-|m|^2S)^{-1}$. Note that unlike 
in the analysis of \eqref{vselfv}, here $1-|m|^2S$ and not $1-m^2S$ has
to be inverted. Since $|m|^2\sim 1-C\eta$ for small $\eta$ and
$S$ has eigenvalue 1, the inverse of  $1-|m|^2S$ is very unstable.
Fortunately, one can subtract the constant mode in \eqref{Tself}
before solving the equation, thus eventually only 
the norm of $(1-|m|^2S)^{-1}$ on the subspace orthogonal to the
constants is relevant. Thus the spectral gap of $S$ plays an important role
and we use that for band matrices the gap is of order $(W/N)^2$.
The details are found in \cite{EKYY3}, where,
among other results, the following theorem was shown:

\begin{theorem}[Diffusion profile]\cite[Theorem 2.4]{EKYY3}\label{lm:withprof}
Let $H$ be a random band matrix with band width $W$, i.e. the variances are given by \eqref{bandvar}.
Suppose that $ N \ll W^{5/4}$ and $(W/N)^2 \le \eta \le 1$.
Then
\be\label{Tprec}
 |T_{xy} -\Theta_{xy}| \;\prec\; \frac{1}{N\eta}\,, \qquad
  \absB{P_x |G_{xy}|^2 - \delta_{xy}|m|^2- \abs{m}^2 \Theta_{xy}} \;\prec\; \frac{1}{N\eta}
 + \frac{\delta_{xy}}{\sqrt{W}}\,.
\ee
All estimates are uniform in $x,y \in \bb T$ and in the spectral parameter
$z=E+\ii \eta$ for $|E|\le 2-\kappa$ and 
for any fixed $\kappa>0$.
\end{theorem}

This theorem identifies $|G_{xy}|^2$ in two different senses of averaging; $T_{xy}$
averages in one of the indices, while $P_x$ takes partial expectation. In
both cases the result is essentially $\Theta_{xy}$.

\medskip

The behavior of $\theta_{x-y}=\Theta_{xy}$ can be analyzed by inverse Fourier transform
from \eqref{thet}. If $\eta \ll (W/N)^2$ then $\theta$ 
is essentially a constant, i.e.\ the profile is flat.
Conversely, if $\eta \gg (W/N)^2$ then we get 
 an exponentially decaying profile on the scale $\abs{x} \sim W \eta^{-1/2}$. 
The shape of the profile is therefore nontrivial if and only if $\eta \gg (W/N)^2$. 
The total mass of the profile 
\be\label{totmass}
\sum_{x \in \bb T} \theta_x = \frac{\im m}{\eta}(1+O(\eta)) = O(\eta^{-1}),
\ee
and the average height of the profile is of order $(N \eta)^{-1}$. The peak of the exponential profile
 has height of order $(W \sqrt{\eta})^{-1}$, which dominates over the average height if and only if 
$\eta \gg (W/N)^2$. The regime $\eta \gg (W/N)^2$ corresponds to the regime where $\eta$ is 
sufficiently large that the complete delocalization has not taken place, and
the profile is mostly concentrated in the region
$|x-y| \leq W\eta^{-1/2} \ll N$.

These scenarios are best understood in a dynamical picture in which $\eta$ is decreased down from $1$.  The ensuing dynamics of 
$\theta$ corresponds to the diffusion approximation, where the
quantum problem is replaced with a random walk of step-size of order $W$.
On a configuration space consisting of $N$ sites, such a random walk
will reach an equilibrium beyond time scales $(N/W)^2$. 
Here $\eta^{-1}$ plays essentially the role of time $t$, so that in this dynamical picture equilibrium is 
reached for $t \sim \eta^{-1} \gg (N/W)^2$. Figure \ref{figure: profile} illustrates this diffusive spreading of the profile for
different values of $\eta$.

\begin{figure}[ht!]
\begin{center}
\includegraphics{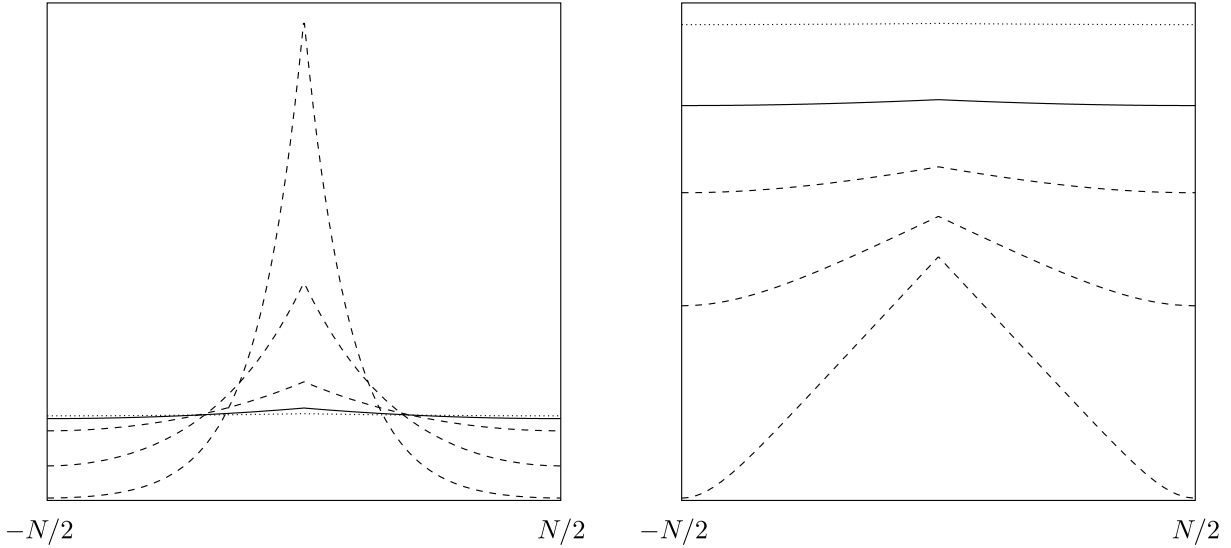}
\end{center}
\caption{A plot of the diffusion profile function at five different values of $\eta$, where the argument $x$ ranges over the torus $\bb T$. Left: the graph $x \mapsto \eta \theta_x$ (see \eqref{totmass} for the choice of normalization). Right: the graph $x \mapsto \log \theta_x$. Here we chose $N = 25 W$ and $\eta = 5^{-k}$ for $k = 1,2,3,4,5$. The cases $k = 1,2,3$ (where $\eta > (W/N)^2$) are drawn using dashed lines, the case $k = 4$ (where $\eta = (W/N)^2$) using solid lines, and the case $k = 5$ (where $\eta < (W/N)^2$) using dotted lines.
\label{figure: profile}}
\end{figure}

One important consequence of Theorem~\ref{lm:withprof} is that it proves delocalization 
for band matrices with band width $W\gg N^{4/5}$ (see Corollary 2.3 of \cite{EKYY3} for 
the precise statement). This improves the earlier result from \cite{EK, EK2}
where delocalization for $W\gg N^{6/7}$ was proved with very different methods.
We remark that from the other side it is known that narrow band matrices
with $W\ll N^{1/8}$  are in the localized regime \cite{Sche}. The conjectured
threshold for the phase transition is $W\sim \sqrt{N}$, see \cite{Fy}.

\subsection{Proof of the local semicircle law without using the spectral gap}\label{sec:nogap}

In this section we sketch the proof of a weaker version of Theorem~\ref{thm: with gap}, namely 
we replace threshold $\wt\eta_E$ with a larger threshold $\eta_E$ defined as
\begin{equation} \label{def eta E}
\eta_E \;\deq\; \min \hBB{\eta \;\col\; \frac{1}{M \eta} \leq
 \min \hbb{\frac{M^{-\gamma}}{\Gamma(z)^3} \,, \,
\frac{M^{-2 \gamma}}{\Gamma(z)^4 \im m(z)}} \text{ for all }
 z \in [E + \ii \eta, E + 10 \ii] }\,.
\end{equation}
This definition is exactly the same as \eqref{def wt eta E}, but $\wt\Gamma$ is
replaced with the larger quantity $\Gamma$, in other words we do not make use
of the spectral gap in $S$. This will pedagogically simplify the presentation
and in Section~\ref{sec:withgap} we will comment on the modifications
for the stronger result.

 We recall that here is no difference between  $\Gamma$ and $\wt\Gamma$ away
from the edges (both are of order 1), so readers interested 
in the local  semicircle law only in the bulk should be content with the simpler proof.
Near the spectral edges, however, there is a substantial difference.
Note that even in the Wigner case (see \eqref{Gammawigner}),
$\eta_E$ is much larger  near the spectral edges
 than the optimal threshold $\wt\eta_E\sim 1/N$.

\begin{definition} \label{def: domain}
 We call a deterministic nonnegative
function $\Psi \equiv \Psi^{(N)}(z)$ an \emph{admissible control parameter} if we have
\begin{equation*}
c M^{-1/2} \;\leq\; \Psi \;\leq\; M^{-c}
\end{equation*}
for some constant $c > 0$ and large enough $N$. Moreover, we call 
any (possibly $N$-dependent) subset 
$$
  \bD =\bD^{(N)} \subset \Big\{ z\; : \; |E|\le 10, \eta\ge M^{-1+\gamma}\Big\}
$$
\emph{a spectral domain}.
\end{definition}

In this section we will mostly use the
 spectral domain
$$
  {\bf S}: = \Big\{ z\, : \, |E|\le 10, \; \eta\in [\eta_E, 10]\Big\}.
$$
Define the random control parameters
\begin{equation}\label{def:Lambda}
\Lambda_o \;\deq\; \max_{i \neq j} \abs{G_{ij}}\,, \qquad
\Lambda_d \;\deq\; \max_i \abs{G_{ii} - m}\,,
\qquad \Lambda \;\deq\; \max (\Lambda_o, \Lambda_d)\,,
\qquad \Theta \;\deq\; \abs{m_N - m}\,.
\end{equation}
In the typical regime we will work, all these quantities are small.
The key quantity is $\Lambda$ and we will develop an iterative
argument to control it. The first step is an apriori bound:

\begin{proposition} \label{prop: rough bound 1}
We have $\Lambda \prec M^{-\gamma / 3} \Gamma^{-1}$ uniformly in $\bS$.
\end{proposition}

The main estimate behind the proof of Theorem~\ref{thm: with gap} for $\eta\ge \eta_E$ is the following
iteration statement:

\begin{proposition} \label{prop: optimal simple}
Let $\Psi$ be a control parameter
satisfying 
\begin{equation} \label{condition on Psi 1}
c M^{-1/2} \;\leq\; \Psi \;\leq\; M^{-\gamma/3} \Gamma^{-1}\,.
\end{equation}
 and
 fix $\e \in (0,\gamma/3)$. Then on the domain $\bS$ we have the implication
\begin{equation} \label{iteration step}
\Lambda \;\prec\; \Psi \qquad \Longrightarrow \qquad \Lambda \;\prec\; F(\Psi)\,,
\end{equation}
where we defined
\begin{equation*}
F(\Psi) \;\deq\; M^{-\e} \Psi + \sqrt{\frac{\im m}{M \eta}} + \frac{M^{\e}}{M \eta}\,.
\end{equation*}
\end{proposition}

The proofs of these two propositions are postponed, we first complete the proof of the local
semicircle law.

It is easy to check that, on the domain $\bS$,  if $\Psi$ satisfies \eqref{condition on Psi 1} then so does $F(\Psi)$.
We may therefore iterate \eqref{iteration step}. This  yields
a bound on $\Lambda$ that is essentially the fixed point of the map $\Psi \mapsto F(\Psi)$,
 which is given by
$\Pi$, defined in \eqref{def:Pi},  (up to the factor $M^\e$). More precisely,
the iteration is started with $\Psi_0 \deq M^{-\gamma/3} \Gamma^{-1}$; the initial hypothesis $\Lambda \prec \Psi_0$ is provided 
by  Proposition \ref{prop: rough bound 1}. For $k \geq 1$ we set $\Psi_{k+1} \deq F(\Psi_k)$. Hence
 from \eqref{iteration step} we conclude that $\Lambda \prec \Psi_k$ for all $k$. Choosing $k \deq \ceil{\e^{-1}}$ yields
\begin{equation*}
\Lambda \;\prec\; \sqrt{\frac{\im m}{M \eta}} + \frac{M^{\e}}{M \eta}\,.
\end{equation*}
Since $\e$ was arbitrary, we have proved that
\begin{equation} \label{Gijest in proof}
\Lambda \;\prec\; \Pi\,,
\end{equation}
which is \eqref{Gijest 2}.

To prove \eqref{m-mest 2}, i.e.\ to estimate $\Theta$, we
rewrite  \eqref{vself} as
\begin{equation} \label{vself for exp}
-\sum_k s_{ik} v_k + \Upsilon_i \;=\; \frac{1}{m + v_i} - \frac{1}{m}\,,
\end{equation}
and expand the right hand side. 
Since $|m|\ge c$ and $|v_i|\le \Lambda$, the expansion is possible on the event
where $\Lambda\ll 1$, which occurs with very high probability by  Proposition~\ref{prop: rough bound 1}.
On this event we get
\begin{equation} \label{expanded vself 2}
 m^2 \pbb{-\sum_k s_{ik} v_k + \Upsilon_i} \;=\; - v_i + O( \Lambda^2)\,.
\end{equation}
Averaging in \eqref{expanded vself 2}       yields
\begin{equation}\label{averr}
m^2 \pb{-[v] + [\Upsilon]} \;=\; -  [v] + O_\prec( \Lambda^2)\,.
\end{equation}

We will show in Lemma~\ref{lemma: Lambdao} in the next section that $|\Upsilon_i|\prec \Pi$, but
in fact the average $[\Upsilon]$ is one order better.  This is due 
to the fluctation averaging phenomenon, and we have
\begin{proposition} \label{lem: avg Upsilon}
Suppose that $\Lambda_o \prec \Psi_o$ for some deterministic 
control parameter $\Psi_o$ satisfying $M^{-1/2}\le \Psi\le M^{-c}$. Then $[\Upsilon] = O_\prec(\Psi_o^2)$.
\end{proposition}
We will explain the proof in Section~\ref{sec:fluc}.
Using this proposition and \eqref{Gijest in proof},  we get
\begin{equation*}
[v] \;=\; m^2  [v] + O_\prec(\Pi^2)\,.
\end{equation*}
 Therefore
\begin{equation*}
\abs{[v]} \;\prec\; \frac{\Pi^2}{\abs{1 - m^2}} \;\leq\; \pbb{\frac{\im m}{\abs{1 - m^2}} + \frac{1}{\abs{1 - m^2} M \eta}} \frac{2}{M \eta} \;\leq\; \pbb{C + \frac{\Gamma}{M \eta}} \frac{2}{M \eta} \;\leq\; \frac{C}{M \eta}\,.
\end{equation*}
Here in the third step we used the elementary explicit bound $\im m\le C|1-m^2|$,
 and the bound
 $\Gamma \geq \abs{1 - m^2}^{-1}$ which follows from the definition of $\Gamma$ by applying the
 matrix $(1 - m^2 S)^{-1}$ to the constant vector.
 The last step follows from the definition of $\bS$. Since $\Theta = \abs{[v]}$, 
this concludes the proof of \eqref{m-mest 2}, and hence of Theorem \ref{thm: with gap} in the regime $\bS$, i.e.
for $\eta\ge \eta_E$. \qed

\medskip

In the next sections we explain the proofs of the three propositions used in this argument.
We first control the off-diagonal elements, i.e. $\Lambda_o$, then we turn to
the proof of Propositions \ref{prop: rough bound 1},  \ref{prop: optimal simple} and \ref{lem: avg Upsilon}.

\subsubsection{Basic estimates for  $\Lambda_o$ and $\Upsilon_i$}

\begin{lemma} \label{lemma: Lambdao}
The following statements hold for any spectral domain $\bD$ and admissible control parameter $\Psi$.
If $\Lambda \prec \Psi$ then
\begin{equation} \label{Lambdao 1}
\Lambda_o + \abs{\Upsilon_i} \;\prec\; \sqrt{\frac{\im m + \Psi}{M \eta}}\,.
\end{equation}
Moreover, for any fixed ($N$-independent) $\eta > 0$ we have
\begin{equation} \label{Lambdao 2}
 \Lambda_o + \abs{\Upsilon_i} \;\prec\; M^{-1/2}\,,
\end{equation}
uniformly in $z\in \{w\in \bD\; : \;\im w=\eta\}$.
\end{lemma}

We remark that we could have written \eqref{Lambdao 1} as
\begin{equation} \label{Lambdao 1 lousy}
\Lambda_o + \abs{\Upsilon_i} \;\prec\; \sqrt{\frac{\im m + \Lambda}{M \eta}}\,,
\end{equation}
but this formulation, while it carries the essence, is literally incorrect 
since it holds only if $\Lambda \prec M^{-c}$ has been apriori established.

\medskip

{\it Proof of Lemma~\ref{lemma: Lambdao}.}
We first observe that $\Lambda\prec \Psi\ll 1$ and the positive lower bound $|m(z)|\ge c$
implies that
\be\label{1/G}
    \frac{1}{|G_{ii}|}\prec 1.
\ee
A simple iteration of the expansion formulas \eqref{resolvent expansion type 1} concludes
that
\be\label{GT}
   |G_{ij}^{(\bb T)}|\prec \Psi, \quad \mbox{for $i\ne j$}, \qquad |G_{ii}^{(\bb T)}|\prec 1,
\qquad \frac{1}{|G_{ii}^{(\bb T)}|}\prec 1
\ee
for any subset $\bb T$ of fixed cardinality.

We begin with the first statement in Lemma ~\ref{lemma: Lambdao}.
First we estimate $Z_i$, which we split as
\begin{equation} \label{Zi split}
\abs{Z_i} \;\le\; \absBB{\sum_{k}^{(i)} \pb{\abs{h_{ik}}^2 - s_{ik}} G_{kk}^{(i)}} + \absBB{\sum_{k \neq l}^{(i)} h_{ik} G_{kl}^{(i)} h_{li}}\,.
\end{equation}
We estimate each term using Theorem \ref{thm: LDE}  by conditioning on $G^{(i)}$ and
using the fact that the family $(h_{ik})_{k=1}^N$ is independent of $G^{(i)}$.
By \eqref{lde 1} the first term of \eqref{Zi split} is stochastically dominated by
$$
 \Big[ \sum_k^{(i)} s_{ik}^2 |G_{kk}^{(i)}|^2\Big]^{1/2}\prec M^{-1/2},
$$
where \eqref{GT}, \eqref{s leq W} and \eqref{S is stochastic} were used.
 For the second term of \eqref{Zi split}
we apply Theorem \ref{thm: LDE} (ii) with $a_{kl} = s_{ik}^{1/2} G_{kl}^{(i)} s_{li}^{1/2}$ 
and $X_k = \zeta_{ik}$ (see \eqref{def:zeta}). We find
\begin{equation} \label{estimate for sum Gkl}
\sum_{k,l}^{(i)} s_{ik} \absb{G_{kl}^{(i)}}^2 s_{li} \;\leq\; \frac{1}{M} \sum_{k,l} s_{ik} \absb{G_{kl}^{(i)}}^2 
\;=\; \frac{1}{M \eta} \sum_k^{(i)} s_{ik} \im G_{kk}^{(i)} \;\prec\; 
\frac{\im m + \Psi}{M \eta}\,,
\end{equation}
where in the last step we used \eqref{resolvent expansion type 1} and the estimate $1/G_{ii} \prec 1$.
Thus we get
\begin{equation} \label{Zi estimate 1}
\abs{Z_i} \;\prec\; \sqrt{\frac{\im m + \Psi}{M \eta}}\,,
\end{equation}
where we absorbed the bound $M^{-1/2}$ on the first term of \eqref{Zi split} into the
 right-hand side of \eqref{Zi estimate 1}, using $\im m \geq c\eta$ as follows from an explicit estimate.

Next, we estimate $\Lambda_o$. We can iterate \eqref{resolvent expansion type 2} once 
to get, for $i \neq j$,
\begin{equation} \label{iterated identity}
G_{ij} \;=\; - G_{ii}\sum_k^{(i)} h_{ik} G_{kj}^{(i)} \;=\; - G_{ii} G_{jj}^{(i)} \pBB{h_{ij} - \sum_{k,l}^{(ij)} h_{ik} G_{kl}^{(ij)} h_{lj}}\,.
\end{equation}
The term $h_{ij}$ is trivially $O_\prec(M^{-1/2})$. In order to estimate the other term, we
 invoke Theorem \ref{thm: LDE} (iii) with $a_{kl} = s_{ik}^{1/2} G_{kl}^{(ij)} s_{lj}^{1/2}$, 
$X_k = \zeta_{ik}$, and $Y_l = \zeta_{lj}$. As in \eqref{estimate for sum Gkl}, we find
\begin{equation*}
\sum_{k,l}^{(i)} s_{ik} \absb{G_{kl}^{(ij)}}^2 s_{lj} \;\prec\; \frac{\im m + \Psi}{M \eta}\,,
\end{equation*}
and thus
\begin{equation} \label{Lambdao estimate 1}
\Lambda_o \;\prec\; \sqrt{\frac{\im m + \Psi}{M \eta}}\,,
\end{equation}
where we again absorbed the term $h_{ij} \prec M^{-1/2}$ into the right-hand side.

In order to estimate $A_i$ and $h_{ii}$ in the definition of $\Upsilon_i$, we use 
\eqref{GT} to estimate 
\begin{equation*}
\abs{A_i} + \abs{h_{ii}} \;\prec\; \Lambda_o^2 + M^{-1/2} \;\leq\; 
\Lambda_o + C \sqrt{\frac{\im m}{M \eta}} \;\prec\; \sqrt{\frac{\im m + \Lambda}{M \eta}}\,,
\end{equation*}
where the second step follows from $\im m \ge c\eta$. Collecting
\eqref{Zi estimate 1}, \eqref{Lambdao estimate 1}, this completes the proof of \eqref{Lambdao 1}.

The proof of \eqref{Lambdao 2} is almost identical to that of \eqref{Lambdao 1}.
 The quantities $\absb{G^{(i)}_{kk}}$ and $\absb{G^{(ij)}_{kk}}$ are estimated by 
the trivial deterministic bound $\eta^{-1} = O(1)$. We omit the details. \qed

\medskip

\subsubsection{Sketch of the proof of Proposition \ref{prop: rough bound 1}}\label{sec:gapp}

The core of the proof is a {\it continuity argument.}
 Its basic idea is to establish a \emph{gap} in the range of  $\Lambda$ 
by proving 

\medskip

{\bf Claim 1.} On the event $\Lambda \leq M^{-\gamma/4} \Gamma^{-1}$ we
actually have the stronger bound  $\Lambda \prec M^{-\gamma/2} \Gamma^{-1}$.

\medskip

\noindent
 In other words, for all $z \in \bS$, with high probability either $\Lambda \leq M^{-\gamma/2} \Gamma^{-1}$
 or $\Lambda \geq M^{-\gamma/4} \Gamma^{-1}$.  The second step is to show that  $\Lambda \leq M^{-\gamma/2} \Gamma^{-1}$
holds for $z$ with a large imaginary part $\eta$:

\medskip

{\bf Claim 2.} We have $\Lambda \prec M^{-1/2}$ uniformly in $z \in [-10,10] + 2 \ii$.

\medskip

Thus, for large $\eta$ the parameter $\Lambda$ is below the gap. 
Using the fact that $\Lambda$ is continuous in $\eta = \im z$ and hence cannot jump from one side of the
 gap to the other, we then conclude that $\Lambda$ is below the gap for all $z \in \bS$ 
and this is Proposition~\ref{prop: rough bound 1}.
See Figure \ref{fig: gap lambda} for an illustration of this argument. 

\begin{figure}[ht!]
\begin{center}
\includegraphics{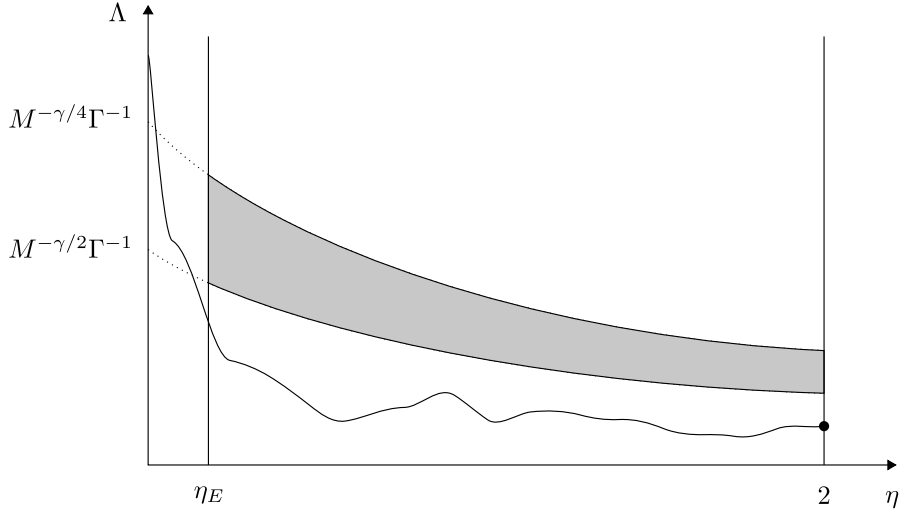}
\end{center}
\caption{The $(\eta, \Lambda)$-plane for a fixed $E$ with the graph of $\eta\to \Lambda(E+\ii \eta)$.
The shaded region is forbidden with high probability
 by Claim 1. The initial estimate, Claim 2, is marked with a black dot.
The graph of $\Lambda = \Lambda(E + \ii \eta)$ is continuous and lies beneath the shaded region.
Note that this method does not control $\Lambda(E+i\eta)$
in the regime  $\eta\le \eta_E$.  \label{fig: gap lambda}
}
\end{figure}

Now we explain Claim 1. We will work on the event $\Lambda \leq M^{-\gamma/4} \Gamma^{-1} \le M^{-c}$,
where  we may invoke \eqref{Lambdao 1}
 to estimate $\Lambda_o$ and $\Upsilon_i$.
 In order to estimate $\Lambda_d$, we expand the right-hand side of \eqref{vself for exp} in $v_i$ to get
\eqref{expanded vself 2}.
Using \eqref{Lambdao 1} to estimate $\Upsilon_i$,  we therefore have
\begin{equation}\label{toinvert}
v_i - m^2 \sum_{k} s_{ik} v_k \;=\; O_\prec \pbb{\Lambda^2 + \sqrt{\frac{\im m + \Lambda}{M \eta}}}\,.
\end{equation}
We write the left-hand side as $ [(1-m^2S)\bv]_i$ with the vector $\bv = (v_i)_{i = 1}^N$. Inverting
the operator $1-m^2 S$, we
therefore conclude that
\begin{equation}\label{Lambdadd}
\Lambda_d \;=\;  \max_i \abs{v_i} \;\prec\; \Gamma \pbb{\Lambda^2 + \sqrt{\frac{\im m + \Lambda}{M \eta}}}\,.
\end{equation}
Together with \eqref{Lambdao 1} and $\Gamma\ge c$, we therefore get
\begin{equation} \label{phi Lambda}
\Lambda \;\prec\;  \Gamma \pbb{\Lambda^2 + \sqrt{\frac{\im m + \Lambda}{M \eta}}}\,.
\end{equation}
On the event  $\Lambda \leq M^{-\gamma/4} \Gamma^{-1}$
we may estimate
\begin{equation*}
\Gamma \Lambda^2 \;\leq\; M^{-\gamma/2} \Gamma^{-1}\,.
\end{equation*}
Moreover, by definition of $\bS$, we have $\Gamma^3 \Pi \leq M^{-\gamma}$. 
Using the definition of $\Pi$ \eqref{def:Pi}, we therefore get
\begin{equation*}
 \Gamma \sqrt{\frac{\im m + \Lambda}{M \eta}} \;\leq\; 
\Gamma \Pi + \Gamma \sqrt{\Gamma^{-1} \Pi} \;\leq\; C M^{-\gamma /2} \Gamma^{-1}\,.
\end{equation*}
Plugging this into \eqref{phi Lambda} yields $ \Lambda \prec M^{-\gamma/2} \Gamma^{-1}$, which is Claim 1.

\medskip

Finally, we explain Claim 2. 
We write \eqref{vself for exp} as
\begin{equation}\label{vi}
v_i \;=\; \frac{m \pb{\sum_k s_{ik} v_k - \Upsilon_i}}{\pb{m^{-1} - \sum_k s_{ik} v_k + \Upsilon_i}}\,.
\end{equation}
In the regime $\eta=2$, all resolvent entries are bounded by $1/2$, thus $|v_i|\le 1$.
Since  $\abs{\Upsilon_i} \prec M^{-1/2}$ by \eqref{Lambdao 2} and
$\abs{m^{-1}} \geq 2$  we find 
\begin{equation*}
\absbb{m^{-1} + \sum_k s_{ik} v_k - \Upsilon_i} \;\geq\; 1 + O_\prec(M^{-1/2})\,.
\end{equation*}
Using $\abs{m} \leq 1/2$ we therefore conclude from \eqref{vi} that
\begin{equation*}
\Lambda_d \;\leq\; \frac{\Lambda_d + O_\prec(M^{-1/2})}{2 + O_\prec(M^{-1/2})} 
 \;=\;  \frac{\Lambda_d}{2} + O_\prec(M^{-1/2})\,,
\end{equation*}
i.e. $\Lambda_d= O_\prec(M^{-1/2})$. Together with 
the estimate
on $\Lambda_o$ from \eqref{Lambdao 2}, we obtained Claim 2. \qed

\subsubsection{Sketch of the proof of Proposition  \ref{prop: optimal simple}}

This argument is very similar to the proof of Claim 1 in Section~\ref{sec:gapp}.
We can work on the event  $\Lambda\le M^{-\gamma/4}$, then
the bound 
\eqref{Lambdao 1} is available to estimate $\Lambda_o$ and $\Upsilon_i$ .
Next, we estimate $\Lambda_d$.  
We expand the right-hand side of \eqref{vself for exp}, we get    
\begin{equation*}
\psi \abs{v_i} \;\leq\; C \psi \absbb{\sum_k s_{ik} v_k - \Upsilon_i} + C \psi \Lambda^2\,.
\end{equation*}
Using the fluctuation averaging estimate \eqref{averaging without Q} explained in the
next section, as well as 
\eqref{Lambdao 1}, we find
\begin{equation} \label{bound for Lambdad 1}
\Lambda_d \;\prec\; \Gamma \Psi^2 + \sqrt{\frac{\im m + \Psi}{M \eta}}\,,
\end{equation}
which, combined with \eqref{Lambdao 1},
 yields
\begin{equation} \label{bound Lambda 1}
\Lambda \;\prec\; \Gamma \Psi^2 + \sqrt{\frac{\im m + \Psi}{M \eta}} \le  
M^{-\e} \Psi + \sqrt{\frac{\im m}{M \eta}} + \frac{M^{\e}}{M \eta}= F(\Psi),
\end{equation}
where in the last step
we used the assumption $\Psi \leq M^{-\gamma/3} \Gamma^{-1}$ and $\e\le\gamma/3$.
\qed

\subsubsection{Fluctuation averaging: proof of Proposition~\ref{lem: avg Upsilon}}\label{sec:fluc}

The leading error in the self-consistent equation \eqref{expanded vself 2} for $v_i$ is  $\Upsilon_i$.
Among the three summands of $\Upsilon_i$, see \eqref{def:Ups}, typically $Z_i$ is the largest, thus
we typically have
\be\label{vZ}
   v_i\prec Z_i
\ee
in the regime where $\Gamma$ is bounded.
The large deviation bounds in Theorem~\ref{thm: LDE} show that $Z_i\prec \Lambda_o$ and
a simple second moment calculation shows that this bound is essentially optimal. On the other hand,
\eqref{ward} shows that
the typical size of the off-diagonal resolvent matrix elements is at least of order $(N\eta)^{-1/2}$, thus
the estimate
$$
   Z_i \prec \Lambda_o \prec \frac{1}{\sqrt{N\eta}}
$$
is essentially  optimal in the standard Wigner case $(M=N)$. Together with  \eqref{vZ}
this shows that the natural bound for $\Lambda$ is $(N\eta)^{-1/2}$,  which  is also
reflected in the bound \eqref{Gijest 2}.  

However, the bound \eqref{m-mest 2} for the average, $m_N-m = [v]$, is of order $\Lambda^2\sim (N\eta)^{-1}$,
i.e. it is one order better than the bound for $v_i$. For the purpose of $[v]$, it is the average $[\Upsilon]$
of the leading errors $\Upsilon_i$ that matters, see \eqref{averr}. Since $Z_i$, the leading term in
$\Upsilon_i$,  is a fluctuating quantity with zero expectation,
 the improvement comes from the fact that fluctuations cancel out
in the average. The basic example of this phenomenon is the central limit theorem.
In our case, however, $Z_i$ are not independent. In fact, their
correlations do not decay, at least in the Wigner case where all indices $i$ play
symmetric role. Thus standard results on central limit theorems for weakly correlated
random variables do not apply.

Here we formulate a version of the fluctuation averaging mechanism, taken from \cite{EKY3},
 that is the most useful 
for this discussion and we comment on the history afterwards. 

 We shall perform the averaging with respect to a family of weights $T = (t_{ik})$ satisfying
\begin{equation} \label{condition on weights}
0 \;\leq\; t_{ik} \;\leq\; M^{-1} \,, \qquad \sum_{k} t_{ik} \;=\; 1\,.
\end{equation}
Typical example  weights are $t_{ik} = s_{ik}$ and $t_{ik} = N^{-1}$. Note that in both of these cases $T$ commutes with $S$.

\begin{theorem}[Fluctuation averaging] \label{thm: averaging}
Fix a spectral domain $\bD$ and  deterministic admissible control parameters $\Psi, \Psi_o\le M^{-c}$.
Suppose that $\Lambda \prec \Psi$, $\Lambda_o\prec \Psi_o$
and the weight $T = (t_{ik})$ satisfies \eqref{condition on weights}.
Then we have
\begin{equation} \label{averaging with Q}
 \sum_{k} t_{ik} Q_k \frac{1}{G_{kk}} \;=\; O_\prec(\Psi^2_o)\,, \qquad \sum_{k} t_{ik} Q_k G_{kk} \;=\; O_\prec(\Psi^2)\,.
\end{equation}
If in addition $T$ commutes with $S$ then
\begin{equation} \label{averaging without Q}
\sum_{k} t_{ik} v_k \;=\;  O_\prec(\Gamma \Psi^2)\,, \qquad \sum_{k} t_{ik} (v_k - [v]) \;=\;  O_\prec(\wt \Gamma \Psi^2)\,.
\end{equation}
The estimates \eqref{averaging with Q} and \eqref{averaging without Q} are uniform in the index $i$.
\end{theorem}

The first version of the fluctuation averaging mechanism appeared in \cite{EYY2} for the Wigner case,
where $[Z]= N^{-1}\sum_k Z_k$ was bounded by $\Lambda_o^2$. Since $Q_k[G_{kk}]^{-1}$ is essentially
$Z_k$, see \eqref{schur}, this corresponds to the first bound in \eqref{averaging with Q}. 
A different proof (with a better bound on the constants) was given in \cite{EYYrigi}.
A conceptually streamlined  version of the original proof was extended to
sparse matrices \cite{EKYY1} and to sample covariance matrices \cite{PY1}. 
Finally, an extensive analysis in \cite{EKYfluc} treated the fluctuation averaging of general polynomials
of resolvent entries and identified the order of cancellations depending on the
algebraic structure of the polynomial.
 Moreover, in \cite{EKYfluc} an  additional cancellation effect was found 
for the quantity $Q_i|G_{ij}|^2$. This improvement plays a key role  in the proof of Theorem~\ref{lm:withprof},
see \eqref{wtT}. 

All proofs of the fluctuation averaging  theorems rely on computing expectations of high moments
of the averages and carefully estimating various terms of different combinatorial structure. 
In  \cite{EKYfluc} we have  developed a Feynman diagrammatic representation for bookkeeping the terms,
but this is necessary only for the case of general polynomials. For the special cases stated in Theorem
\ref{thm: averaging}, the proof is relatively simple and it is presented in Appendix B of \cite{EKY3}.
Here we will not repeat the proof, 
 we only indicate the main mechanism by estimating the second moment of the first term 
in \eqref{averaging with Q}. The actual proof requires estimating all moments and
then use Chebysev inequality to translate the moment estimates to probabilistic
estimates standing behind the notation $O_\prec(\Psi_o^2)$ in \eqref{averaging with Q}.

\medskip

{\it Second moment calculation.}
First we claim that
\begin{equation} \label{Xi estimate}
\absbb{Q_k \frac{1}{G_{kk}}} \;\prec\; \Psi_o.
\end{equation}
Indeed, from Schur's complement formula \eqref{schur} we get $\abs{Q_k  (G_{kk})^{-1}} \leq \abs{h_{kk}} + \abs{Z_k}$.
The first term is estimated by $\abs{h_{kk}} \prec M^{-1/2} \leq \Psi_o$. The second term 
is estimated exactly as in \eqref{Zi split} and \eqref{estimate for sum Gkl}, giving $|Z_k| \prec \Psi_o$.
 In fact, the same bound \eqref{Xi estimate} holds if $G_{kk}$ is
 replaced with  $G_{kk}^{({\bb T})}$ as
long as $|{\bb T}|$ is bounded.

Abbreviate $X_k \deq Q_k (G_{kk})^{-1}$ and compute the variance
\begin{equation} \label{variance of average of X}
\E \absbb{\sum_k t_{ik} X_k}^2 \;=\; \sum_{k,l} t_{ik} t_{il} \E X_{k} \ov{X_{l}} \;=\; \sum_{k} t_{ik}^2 \E X_{k} \ov{X_{k}}
 + \sum_{k \neq l} t_{ik} t_{il}\E X_{k} \ov{X_{l}}\,.
\end{equation}
Using the bounds \eqref{condition on weights} on $t_{ik}$ and \eqref{Xi estimate},
 we find that the first term on the right-hand
 side of \eqref{variance of average of X} is $O_\prec(M^{-1} \Psi_o^2) = O_\prec(\Psi_o^4)$, 
where we used that $\Psi_o$ is admissible.
 Let us therefore focus on the second term of \eqref{variance of average of X}.
 Using the fact that $k \neq l$, we apply \eqref{resolvent expansion type 1} to $X_{k}$ and $X_{l}$ to get
\begin{equation} \label{variance calculation}
\E X_{k} \ov{X_{l}} \;=\; \E Q_{k} \pbb{\frac{1}{G_{k k}}} Q_{l} \ov{\pbb{\frac{1}{G_{l l}}}} \;=\;
\E Q_{k} \pbb{\frac{1}{G_{kk}^{(l)}} - \frac{G_{kl} G_{lk}}{G_{kk} G_{kk}^{(l)} G_{ll}}} Q_{l} \ov{\pbb{\frac{1}{G_{ll}^{(k)}} - \frac{G_{lk} G_{kl}}{G_{ll} G_{ll}^{(k)} G_{kk}}}}\,.
\end{equation}
We multiply out the parentheses on the right-hand side. The crucial observation is that if the random variable $Y$ is independent of $i$ (see Definition \ref{definition: P Q}) then $\E Q_i(X) Y = \E Q_i(XY) = 0$. Hence out of the four terms obtained from the right-hand side of \eqref{variance calculation}, the only nonvanishing one is
\begin{equation*}
\E Q_{k} \pbb{\frac{G_{kl} G_{lk}}{G_{kk} G_{kk}^{(l)} G_{ll}}} Q_{l} \ov{\pbb{\frac{G_{lk} G_{kl}}{G_{ll} G_{ll}^{(k)} G_{kk}}}} \;\prec\; \Psi_o^4\,,
\end{equation*}
where we used that the denominators are harmless, see \eqref{1/G}.
Together with \eqref{condition on weights}, this concludes the proof of $\E \absb{\sum_{k}t_{ik} X_k}^2 \prec \Psi_o^4$,
which means that $\sum_{k}t_{ik} X_k$ is bounded by $\Psi_o^2$ in second moment sense.  \qed

\medskip

Finally, Proposition~\ref{lem: avg Upsilon} directly follows from the first estimate in
\eqref{averaging with Q} with $t_{ik}=1//N$,  since from \eqref{schur} we have
$$
 Q_k \frac{1}{G_{kk}} = h_{kk} - Z_k = \Upsilon_k - A_k = \Upsilon_k + O_\prec(\Psi_o^2).
$$
Taking the average over $k$, we get $[\Upsilon] =  O_\prec(\Psi_o^2)$. \qed

\subsection{Remark on the proof of the local semicircle law with using the spectral gap}\label{sec:withgap}

In Section~\ref{sec:nogap} we proved the local semicircle law, Theorem~\ref{thm: with gap}, uniformly
for $\eta\ge \eta_E$ instead of the larger regime $\eta\ge \wt\eta_E$. The difference between
these two thresholds stems from the difference between $\Gamma$ and $\wt\Gamma$, see
 \eqref{def wt eta E} and \eqref{def eta E}. 

The bound $\Gamma$ on the norm of $(1-m^2S)^{-1}$ entered the proof when the self-consistent equation
\eqref{toinvert} was solved.  The key  idea is
 to solve the self-consistent equation \eqref{toinvert} separately on the subspace 
of constants (the span of the vector ${\bf e} $) and on its orthogonal complement ${\bf e}^\perp$.
Roughly speaking, we obtain
\be\label{viv}
  |v_i - [v]|\prec \wt \Gamma \Big(\Lambda^2 + r(\Lambda)\Big), \qquad r(\Lambda):=
 \sqrt{\frac{\im m +\Lambda}{M\eta}},
\ee
instead of \eqref{Lambdadd}. In fact, we can improve this by using the fluctuation averaging.
Subtracting the average over $i$ from the self-consistent equation \eqref{expanded vself 2}
and estimating $|m|^2\le C$, we have
\begin{equation} \label{Lambdad estimate 2}
\absb{v_i - [v]} \;\leq\; C \absbb{\sum_k s_{ik} \pb{v_k - [v]} - \pb{\Upsilon_i - [\Upsilon]}} + O_\prec(\Lambda^2)
 \;\prec\; \wt \Gamma \Lambda^2 + r(\Lambda)\,,
\end{equation}
where in the last step we used the fluctuation averaging estimate 
\eqref{averaging without Q} with $s_{ik}= t_{ik}$, and $\abs{\Upsilon_i} \prec r(\Lambda)$ from \eqref{Lambdao 1}.

On the space of constant vectors, \eqref{vself for exp} becomes a scalar equation for
the average $[v]$, which can be expanded up to second order.
More precisely, assuming $|v_i|\ll 1$, we can expand \eqref{vself for exp} up to second order:
\begin{equation} \label{selfc expanded deg 2}
-\sum_k s_{ik} v_k + \Upsilon_i \;=\; - \frac{1}{m^2} v_i + \frac{1}{m^3} v_i^2 + O(\Lambda^3)\,.
\end{equation}
In order to get a closed equation for $[v]$, we
take the average over $i$:
\begin{equation} \label{averaged vself expanded}
(1 - m^2) [v] - m^{-1} \frac{1}{N} \sum_i v_i^2 \;=\; -m^2 [\Upsilon] + O(\Lambda^3)\, .
\end{equation}
The nonlinear term is estimated by 
\begin{equation*}
\frac{1}{N} \sum_i v_i^2 \;=\; [v]^2 + \frac{1}{N} \sum_i \pb{v_i - [v]}^2 \;=\; [v]^2 + O_\prec \pB{\pb{\wt \Gamma 
\Lambda^2 + r(\Lambda)}^2}, 
\end{equation*}
where \eqref{Lambdad estimate 2} was used. The average $[\Upsilon]$ can be estimated by $O_\prec(\Lambda_o^2)$
as in Proposition~\ref{lem: avg Upsilon}.  Moreover, $\Lambda_o$ is estimated by $r(\Lambda)$ as
in Lemma~\ref{lemma: Lambdao}. We thus obtain a quadratic equation for the scalar
quantity $[v]$:
\be\label{scl}
(1 - m^2) [v] - m^{-1} [v]^2 =  O_\prec\pB{\Lambda^3+ \pb{\wt \Gamma 
\Lambda^2 + r(\Lambda)}^2}.
\ee
The main control parameter in this proof is $\Theta = \abs{[v]}$, and the key iterative
scheme 
 is formulated in terms of $\Theta$. However, many intermediate estimates still involve $\Lambda$. In particular, 
the self-consistent equation \eqref{vself for exp} is effective only 
in the regime where $v_i$ is already small and in the calculation above
we tacitly used that $|v_i|\ll 1$.  Hence we need a preparatory
step to  prove an apriori
bound on $\Lambda$, essentially showing that $\Lambda \ll 1$, in fact
we will need  $\Lambda \ll \wt\Gamma^{-1}$ (compare with Proposition~\ref{prop: rough bound 1}). This
proof itself is a continuity argument  similar to the proof of 
Proposition \ref{prop: rough bound 1}; now, however, we have to follow $\Lambda$ and 
$\Theta$ in tandem. The main reason why $\Theta$ is already involved in
this part is that we work in larger spectral domain $\eta\ge \wt\eta_E$, defined  by using
$\wt\Gamma$. Thus, already in this preparatory step,
 the self-consistent equation has to be solved separately 
on the subspace of constants and its orthogonal complement. We will omit here these details.

The second preparatory step is to control $\Lambda$
in terms of $\Theta$, which allows us to express the error terms in the
self-consistent equation \eqref{scl} in terms of $\Theta = |[v]|$.
We first notice that \eqref{Lambdad estimate 2} implies
$$
   |v_i|\prec \Theta +  \wt \Gamma \Lambda^2 + r(\Lambda),
$$
i.e.
$$
 \Lambda \prec \Theta +  \wt \Gamma \Lambda^2 + r(\Lambda).
$$
Using the apriori bound $\Lambda \ll \wt \Gamma^{-1}$, we get
$$
  \Lambda\prec \Theta + r(\Lambda).
$$
This equation is analogous to \eqref{iteration step} and can be iterated as
in the application of
 in Proposition~\ref{prop: optimal simple} leading to \eqref{Gijest in proof}
to obtain
\be
  \Lambda \prec \Theta + \sqrt{\frac{\im m}{M\eta}} +\frac{1}{M\eta}.
\label{LaTh}
\ee
Plugging this bound into \eqref{scl}, we have a self-consistent equation for
the scalar quantity $[v]$ since $\Theta=|[v]|$. An elementary calculation, using the apriori bound
$\Lambda \ll \wt \Gamma^{-1}$,  yields
 \be\label{scl1}
(1 - m^2) [v] - m^{-1} [v]^2 =  O_\prec\pB{ p(\Theta)^2 + M^{-\gamma/4} \Theta^2}, \qquad
 p(\Theta): = \sqrt{\frac{\im m+ \Theta}{M\eta}} + \frac{1}{M\eta}.
\ee

Finally, in the main step we solve the quadratic inequality \eqref{scl1}
for $\Theta$. If we neglect the error term in \eqref{scl1}, then the equation reduces to
\be\label{unpert}
   (1-m^2)[v] \;=\; m^3 [v]^2,
\ee
which has two solutions: either $[v]=0$ or $[v] =(1-m^2)/m^3$. Away from the spectral edge
we have $|1-m^2|\ge c$
with some positive constant $c$, so the two solutions are separated and they both are  stable
under small perturbations. 
The second solution would mean that $[v]$ is strictly separated
away from zero. But this can be excluded by a continuity argument: for large $\eta$,
say $\eta= 2$, it is easy to prove that $[v]$ is small. Since $[v]$ is a continuous
function of $\eta$, we find that as $\eta$ decreases continuously, $[v]$ cannot suddenly
jump from a value near zero to a value near $(1-m^2)/m^3$.
Thus $[v]$ must remain in the vicinity of the zero solution to \eqref{unpert}. 
This completes the sketch of the proof of Theorem~\ref{thm: with gap}  for the
more general case $\eta\ge \wt\eta_E$. \qed

\section{Universality of  the correlation functions for Wigner matrices}\label{sec:uniwig}

In this section we explain  the sketch of the proof of  Theorem~\ref{bulkWigner}.

\subsection{Dyson Brownian motion and the
local  relaxation flow}\label{sec:DBM}

\subsubsection{Concept and results}

The Dyson Brownian motion (DBM) describes the evolution of the eigenvalues
 of a Wigner matrix as an interacting point process
 if each matrix element $h_{ij}$ evolves according to independent
 (up to symmetry restriction) 
Brownian motions. We will slightly alter this definition by
 generating the dynamics of the matrix elements
by an Ornstein-Uhlenbeck (OU) process  which leaves the standard Gaussian
 distribution invariant. 
 In the Hermitian case, the OU process for the rescaled matrix elements 
 $\zeta_{ij}: = N^{1/2}h_{ij}$ is given by the 
stochastic differential equation 
\be
  \rd \zeta_{ij}= \rd \beta_{ij} - \frac{1}{2} \zeta_{ij}\rd t, \qquad
i,j=1,2,\ldots N,
\label{zij}
\ee
where $ \beta_{ij}$,  $i <  j$, are independent complex Brownian
motions with variance one and $ \beta_{ii}$ are real 
Brownian motions of the same variance. 
 Denote the distribution of 
the eigenvalues $\bla=(\lambda_1, \lambda_2,\ldots, \lambda_N)$
 of $H_t$ at  time $t$
by $f_t ( \bla)\mu (\rd \bla)$
where $\mu$ is given by \eqref{01} with the Gaussian potential $V(x) = x^2/2$.

Then $f_t=f_{t,N}$ satisfies \cite{DyB}
\be\label{dy}
\partial_{t} f_t =  \cL f_t,
\ee
where 
\be
\cL=\cL_N:=   \sum_{i=1}^N \frac{1}{2N}\partial_{i}^{2}  +\sum_{i=1}^N
\Bigg(- \frac{\beta}{4} \lambda_{i} +  \frac{\beta}{2N}\sum_{j\ne i}
\frac{1}{\lambda_i - \lambda_j}\Bigg) \partial_{i}, \quad 
\partial_i=\frac{\partial}{\partial\lambda_i}.
\label{L}
\ee
The parameter $\beta$  is chosen as follows: $\beta= 2$ for complex 
Hermitian matrices  and $\beta=1$ for symmetric real matrices.
Our formulation of the problem has already taken into account  
Dyson's observation  that the invariant measure for this dynamics is $\mu$. 
 A natural question regarding the DBM is how fast the dynamics reaches equilibrium.
  Dyson had already posed this question in 1962:

\medskip

\noindent 
{\bf Dyson's conjecture \cite{DyB}:}  The global equilibrium of DBM is
 reached in  time of order one 
and  the  local equilibrium (in the bulk)  is reached in  time of order  $1/N$. 
  Dyson further remarked,

\medskip
\begin{minipage}[c]{6in}
{\it ``The picture of the gas coming into
equilibrium in two well-separated stages, with microscopic
and macroscopic time scales, is  suggested
 with the help of physical intuition. A
rigorous proof that this picture is accurate would
require a much deeper mathematical analysis."}
\end{minipage}

\bigskip

We will prove that Dyson's conjecture is correct if the 
initial data of the flow is a Wigner ensemble, which was 
Dyson's  original interest. Our result in fact is
 valid for DBM  with much more general initial data that we now survey. 
Briefly, 
it will turn out that the {\it global} equilibrium is
indeed reached within a time of order one,
but {\it local} equilibrium is achieved much faster if  an a-priori
estimate  on the location of the eigenvalues (also called points) is satisfied. 
To formulate this estimate,  
let $\gamma_j =\gamma_{j,N}$ denote  the location of the $j$-th point
under the semicircle law, i.e., $\gamma_j$ is defined by 
\eqref{def:gamma}.

\bigskip
\noindent 
{\bf A-priori Estimate:} There exists an ${\frak a}>0$ such that
\be
 Q=Q_\fa:= \sup_{t\ge N^{- 2 { \frak a}}}   \frac{1}{N}
 \int \sum_{j=1}^N(\lambda_j-\gamma_j)^2
 f_t( \lambda )\mu(\rd \lambda) \le CN^{-1-2{\frak a}}
\label{assum3}
\ee
with a constant $C$ uniformly in $N$. This condition first appeared in
\cite{ESY4}.

The main result on the local ergodicity of Dyson Brownian motion states
 that if the a-priori estimate \eqref{assum3} is satisfied  
then the local correlation functions  of the measure  $f_t \mu$ 
are the same as the corresponding ones for the Gaussian measure, 
 $\mu = f_\infty\mu$, provided that
 $t$ is larger than $N^{-2{\frak a}}$.
The $n$-point
correlation functions of the probability measure $f_t\rd\mu$ are defined,
similarly to \eqref{pk}, by
\be
 p^{(n)}_{t,N}(x_1, x_2, \ldots,  x_n) = \int_{\R^{N-n}}
f_t(\bx) \mu(\bx) \rd x_{n+1}\ldots
\rd x_N, \qquad  \bx=(x_1, x_2,\ldots, x_N).
\label{corr}
\ee
Due to the convention that one can view the locations of eigenvalues as the coordinates of particles,  we have used $\bx$, 
instead of $\bl$,  in the last equation. From now on, we will use both conventions depending 
on which viewpoint we wish to emphasize. 
Notice that the probability distribution  of the eigenvalues at the 
time $t$, $f_t \mu$, 
 is the same as that of the Gaussian divisible  matrix: 
\be\label{matrixdbm}
H_t = e^{-t/2} H_0 + (1-e^{-t})^{1/2}\, U,
\ee
where $H_0$  is the initial generalized  Wigner matrix and 
$U$ is an independent standard GUE (or GOE) matrix. This establishes  the universality of 
the Gaussian divisible ensembles. The precise statement is the following
theorem:

\begin{theorem}\label{thm:DBM} \cite[Theorem 2.1]{ESYY} Suppose that
 the a-priori estimate \eqref{assum3}  holds
 for the solution $f_t$ of the forward equation \eqref{dy} with some exponent $\fa>0$.
 Let $E\in (-2, 2) $ and $b>0$
such that $[E-b,E+b] \subset (-2, 2)$. Then for any $s>0$, 
for any integer $n\ge 1$ and for any compactly supported continuous test function
$O:\bR^n\to \bR$, we have
\be
\begin{split}
\lim_{N\to \infty} \sup_{t\ge  N^{-2\fa+s} } \;
\int_{E-b}^{E+b}\frac{\rd E'}{2b}
\int_{\R^n} &  \rd\alpha_1
\ldots \rd\alpha_n \; O(\alpha_1,\ldots,\alpha_n) \\
&\times \frac{1}{\varrho_{sc}(E)^n} \Big ( p_{t,N}^{(n)}  - p_{G, N} ^{(n)} \Big )
\Big (E'+\frac{\alpha_1}{N\varrho_{sc}(E)},
\ldots, E'+\frac{\alpha_n}{ N\varrho_{sc}(E)}\Big) =0.
\label{abstrthm}
\end{split}
\ee
\end{theorem}

We can choose  $b=b_N$ depending on $N$.  In \cite{ESYY} explicit bounds on the  speed of convergence and the
 optimal range of $b$  were also established.  
In particular, thanks to the 
optimal rigidity estimate \eqref{rigwig} which implies that  \eqref{assum3} holds with any $\fa <1/2$, the
range of the energy averaging  in \eqref{abstrthm}  can be reduced
to $b_N\ge N^{-1+\xi}$, $\xi>0$, 
 but only for $t\ge N^{-\xi/8}$ (Theorem 2.3  of \cite{EYYrigi}).

Theorem \ref{thm:DBM} is a consequence of the following theorem which identifies the averaged gap distribution
of the eigenvalues. 

\begin{theorem}[Universality of the Dyson Brownian motion for short time]\label{thmM}
\cite[Theorem 4.1]{ESYY} \\
Suppose  $\beta\ge1$ and 
let  $O:\bR \to\bR$ be a  smooth function with compact 
support. 
Then for 
any sufficiently small $\e>0$, independent of $N$, there exist
constants $C, c>0$, depending only on $\e$ and $O$ such that
 for any $J\subset \{ 1, 2, \ldots , N-1\}$  
  we have
\be\label{GG}
\Big| \int \frac 1 {|J|} \sum_{i\in J} O( N(x_i-x_{i+1})) f_t  \rd \mu -
\int \frac 1 {|J|} \sum_{i\in J} O( N(x_i-x_{i+1})) \rd\mu \Big|
\le C N^{\e} \sqrt{  \frac {N^2 Q} { |J| t }}  + Ce^{-c N^\e } .
\ee
In particular,  if the a-priori estimate \eqref{assum3} holds with some
$\fa>0$ and $|J|$ is of order $N$, then for any  $t > N^{-2 \fa + 3\e}$
the right hand side converges to zero as $N\to\infty$, i.e. the 
gap distributions for $f_t\rd\mu$ and $\rd\mu$ coincide.
\end{theorem}

The test functions can be generalized to 
\be
O\Big( N(x_i-x_{i+1}), N(x_{i+1}-x_{i+2}), \ldots, N(x_{i+{n-1}}-x_{i+n})\Big)
\label{cG}
\ee
for any $n$ fixed  which is needed to identify higher order correlation functions.
 In applications, $J$ is chosen to be the
 indices of the eigenvalues in the interval $[E-b, E+ b]$
and thus $|J| \sim N b$. 
This identifies the averaged gap distributions of eigenvalues  and thus also identifies 
the correlation functions after energy averaging.  We will not explain here in detail how to
pass information from gap distribution to correlation functions (see Section 7 in \cite{ESYY}),
but we note that this transfer is relatively easy if both statistics are averaged on a scale
larger than the typical fluctuation of a single eigenvalue (which is smaller than $N^{-1+\e}$
in the bulk by \eqref{rigwig}). This  concludes Theorem \ref{thm:DBM}.
Note that the input of this theorem, the apriori estimate \eqref{assum3}, identifies
the location of the eigenvalues only on a scale $N^{-1/2-\fa}$ which is much weaker 
than the $1/N$ precision for the eigenvalue differences in \eqref{GG}.

 By the rigidity estimate \eqref{rigwig},
  the a-priori estimate \eqref{assum3} holds for any $\fa < 1/2$
if the initial data of the DBM is a generalized Wigner ensemble.  
Therefore, Theorem \ref{thmM}  holds for any $t\ge N^{-1 + \e}$ 
for any $\e > 0$ and {\it this establishes Dyson's  conjecture }
for any generalized Wigner matrices.

\subsubsection{Main ideas behind the proof of Theorem \ref{thmM}}

The key method is to  analyze the relaxation to equilibrium of the dynamics \eqref{dy}.
This approach  was first introduced  in Section 5.1 of \cite{ESY4};
the presentation here follows \cite{ESYY}.

 We start with a short review of
the logarithmic Sobolev inequality for a general measure.
Let the probability measure $\mu$ on $\bR^N$ be given by  a general 
Hamiltonian $\cH$:
\be
  \rd\mu(\bx) = \frac{e^{-N \cH(\bx)}}{Z}\rd \bx, \qquad 
\label{convBE}
\ee
In applications $\mu$ will be the Gaussian equilibrium measure, \eqref{01} with $V(x)=x^2/2$,
so we use the same notation $\mu$, but the  statements in the beginning of this section hold
for a general measure. 
Let $\cL$ be the generator, symmetric with respect to the measure $\rd\mu$,
 defined by the associated   Dirichlet form
\be\label{D}
   D(f)= D_\mu(f) = -\int f \cL f \rd \mu := \frac{1}{2N}\sum_j 
  \int(\pt_j f)^2 \rd \mu, \qquad \pt_j = \pt_{x_j}.
\ee
Recall  the relative entropy of two  probability measures:
$$
   S (\nu| \mu ): =   
\int   \frac {\rd \nu}  {\rd \mu} \log \left ( \frac {\rd \nu}  {\rd \mu} \right ) \rd\mu.
$$
If $\rd\nu = f\rd\mu$, then we will sometimes use the notation $S_\mu(f):= S(f\mu|\mu)$.
The entropy can be used to control the total variation norm via the well known inequality 
\be
   \int |f-1|\rd \mu \le \sqrt{ 2 S_\mu(f)}.
\label{entropyneq}
\ee

Let $f_t$ be the solution to  the evolution equation
\be
  \pt_t f_t = \cL f_t,  \qquad t>0,
\label{dybe}
\ee
with a given initial condition $f_0$. The
evolution of the entropy $S_\mu(f_t)=S(f_t \mu | \mu )$
satisfies
\be
  \pt_t S_\mu(f_t )  = -4 D_\mu (\sqrt{f_t}).
\label{derS}
\ee
Following Bakry and \'Emery \cite{BakEme},  the evolution of the Dirichlet form
satisfies  the inequality 
\be\label{eq:BE}
   \pt_t D_\mu(\sqrt{f_t}) 
  \le   -  \frac{1}{2 N} \int
 (\nabla \sqrt f_t)(\nabla^2\cH)\nabla \sqrt  f_t \rd \mu.
\ee
If the Hamiltonian is convex, i.e., 
\be\label{convexham}
 \nabla^2\cH(\bx)=\mbox{Hess} \, \cH(\bx) \ge \varpi 
 \qquad \mbox{for all $\bx\in \bR^N$}
\ee 
 with some constant $\varpi>0$,
then  we have
\be
    \pt_t D_\mu (\sqrt{f_t}) \le - \varpi D_\mu(\sqrt{f_t}).
\label{derD}
\ee
Integrating \eqref{derS} and \eqref{derD} back from infinity to 0,
we obtain
the  {\it logarithmic Sobolev inequality (LSI)}
\be
 S_\mu(f) \le \frac{4}{\varpi}  D_\mu (\sqrt{f}), \quad f = f_0
\label{lsi1}
\ee
and the  {\it exponential relaxation
of the entropy and Dirichlet form on time scale $t\sim 1/\varpi$}
\be
    S_\mu(f_t )  \le e^{-t \varpi } S_\mu(f_0  ),
 \quad D_\mu (\sqrt{f_t}) \le e^{-t \varpi } D_\mu (\sqrt{f_0}) .
\label{Sdec}
\ee
As a consequence of the logarithmic Sobolev inequality, 
we also have the concentration inequality 
for any $k$ and $a>0$
\be\label{concen}
\P^\mu\left(|x_k-\E_\mu(x_k)|> a\right)\leq 2e^{-\varpi N a^2/2 }.
\ee
We will not use this inequality in this section, but it will become important in Section \ref{beta}.

 Returning to the  classical ensembles, we  assume from now on
that $\cH$ is given by \eqref{01}  with $V(x)=x^2/2$ and the
equilibrium measure $\mu$ is the Gaussian one.
We then have  the convexity inequality
\be\label{convex}
\Big\langle \bv , \nabla^2  \cH(\bx)\bv\Big\rangle
\ge   \frac{1}{2} \,  \|\bv\|^2 + \frac{1}{N}
 \sum_{i<j} \frac{(v_i - v_j)^2}{(x_i-x_j)^2} \ge  \frac{1 }{2} \,  \|\bv\|^2,
 \qquad \bv\in\bR^N.
\ee
This  guarantees that $\mu$ satisfies the LSI with $\varpi=1/2$ and the relaxation time to 
equilibrium is of order one.

The key idea is that the  relaxation time is in fact
much shorter than order one for local observables that
depend only on the eigenvalue {\it differences}. Equation \eqref{convex} shows that 
the relaxation in the direction $v_i-v_j$ is much faster than order
 one provided that $x_i- x_j$ are close.
However, this effect is hard to exploit directly  due to that all modes of different 
wavelengths are coupled.
Our idea is to add an auxiliary {strongly convex}  potential $W(\bx)$ to the Hamiltonian
to  ``speed up''
the convergence to local equilibrium. On the other hand, we will also 
show that the cost of this  speeding up 
can be effectively controlled if the a-priori estimate \eqref{assum3} holds.

The auxiliary  potential $W(\bx)$ is defined by 
\be
    W(\bx): =    \sum_{j=1}^N
W_j (x_j)   , \qquad W_j (x) := \frac{1}{2 \tau } (x_j -\gamma_j)^2,
\label{defW}
\ee
i.e. it is a quadratic confinement on scale $\sqrt \tau $ for each eigenvalue
near its classical location, where the parameter $\tau>0$ will be chosen later. 
The total Hamiltonian is given by 
\be
\wt \cH: = \cH +W,
\label{def:wth}
\ee
where $\cH$ is the Gaussian Hamiltonian given by \eqref{01}.
The measure with Hamiltonian $\wt \cH$, 
\be\label{omdef}
   \rd\om: =\om (\bx)\rd\bx, \quad \om: = e^{-N\wt\cH}/\wt Z,
\ee
will be called  the  {\it local relaxation  measure}.

The  {\it local relaxation flow} is defined to be the flow with 
the  generator characterized by  the natural Dirichlet form w.r.t. $\om$, explicitly,  
$\wt \cL$:
\be\label{tl}
\wt \cL = \cL - \sum_j b_j \partial_j, \quad
b_j = W_j'(x_j)= \frac{x_j -\gamma_j}{\tau}.
\ee
 We will choose $\tau \ll 1$ 
 so that the additional term $W$  substantially
increases the lower bound \eqref{convexham} on the Hessian, hence speeding up the
dynamics so that the relaxation  time is at most  $\tau$.

\medskip

 The idea of adding an artificial potential $W$ to speed up the convergence appears to be unnatural 
here. The current formulation is a {streamlined version} of a much more complicated approach
that   appeared in 
\cite{ESY4} and which 
 took ideas from the earlier work \cite{ERSY}.  Roughly speaking, 
in hydrodynamical limit, 
the  short wavelength modes always have shorter  relaxation times than the long wavelength modes.
A direct implementation of this idea is extremely complicated 
due to the logarithmic interaction that couples short and long wavelength modes.
 Adding a strongly convex  auxiliary  potential $W(\bx)$
shortens the relaxation time of  the long wavelength modes, 
but it does not affect the short modes, i.e. the local statistics, which are our main interest. 
The analysis of the new system is much simpler since now the relaxation is faster, uniform for
all modes. Finally, we need to compare the local statistics of the original system
with those of the modified one. It turns out that the difference is governed
by  $(\nabla W)^2$ which can be directly controlled  by the 
a-priori estimate \eqref{assum3}.

Our method for enhancing the convexity of $\cH$ is reminiscent of a standard
convexification idea concerning metastable states.
To explain the similarity, consider a particle near one of the
local minima of a double well potential separated by a local maximum,
or energy barrier. Although the potential is not convex globally,
one may still study a reference problem defined by convexifying
 the potential along with the  well in which the particle initially resides.  Before the 
particle reaches the energy barrier, there is no difference between 
these two problems. Thus questions concerning time scales
shorter than the typical escape time can be conveniently answered 
by considering the convexified problem; in particular the
escape time in the metastability problem itself can be estimated
by using convex analysis.
Our DBM problem  is already convex, but not sufficiently convex.
The modification by adding $W$ enhances convexity without altering 
the local statistics. This is similar to  the convexification in the metastability
problem which does not alter events before the escape time.

\subsubsection{Some details on the proof of  Theorem \ref{thmM}  }

The core of the proof is divided into three theorems.
For the flow with generator $\wt \cL$, 
we have the following estimates on the entropy and Dirichlet form.

\begin{theorem}\label{thm2} 
Consider the forward equation
\be
\partial_t q_t=\wt \cL q_t, \qquad t\ge 0,
\label{dytilde}
\ee
with the reversible  measure $\omega$ defined in \eqref{omdef}. The initial condition $q_0$ satisfies   $\int q_0 \rd \om=1$.
Then we have the following estimates
\be\label{0.1}
\partial_t D_{\omega}( \sqrt {q_t}) \le - \frac{1}{2\tau}D_{\omega}( \sqrt {q_t}) -
\frac{1}{2N^2}  \int   \sum_{i,j=1}^N
\frac{  ( \pt_i \sqrt{ q_t} - \pt_j\sqrt {q_t} )^2}{(x_i-x_j)^2} \rd \omega ,
\ee
\be\label{0.2}
\frac{1}{2N^2} \int_0^\infty  \rd s  \int    \sum_{ i,j=1}^N
\frac{(\pt_i\sqrt {q_s} - \pt_j\sqrt {q_s} )^2 }{(x_i-x_j)^2}\rd \omega
\le D_{\omega}( \sqrt {q_0})
\ee
and the logarithmic Sobolev inequality
\be\label{lsi}
S_\om(q)\le C \tau  D_{\omega}( \sqrt {q_0})
\ee
with a universal constant $C$.
Thus the relaxation time to equilibrium is of order $\tau$:
\be\label{Sdecay}
 S_{\omega}(q_t)\le e^{-Ct/\tau} S_\omega(q_0), \qquad D_{\omega}(q_t)\le e^{-Ct/\tau} D_\omega(q_0).
\ee
\end{theorem}

\begin{proof} Denote by $h=\sqrt{q}$ and we have the equation 
\be\label{1.5}
\pt_t D_\om ( h_t) =   \pt_t \frac{1}{2N}\int 
(\nabla h)^2 e^{- N \wt\cH} \rd \bx
  \le  - \frac{1}{2 N }\int
\nabla h (\nabla^2 \wt\cH)\nabla h  e^{- N \wt\cH} \rd\bx.
\ee
In our case, \eqref{convex} and \eqref{defW} imply  that the Hessian of $\wt \cH$ is bounded
from below as
\be\label{convex4}
\nabla h (\nabla^2 \wt\cH)\nabla h
\ge        \frac C { \tau} \sum_j  (\partial_j h)^2 +
\frac{1}{2N   }  \sum_{ i,j}  \frac 1 {(x_i - x_j)^2} (\pt_i h -
\pt_j h)^2
\ee
with some positive constant $C$.
This proves \eqref{0.1} and \eqref{0.2}. The  rest can be proved by straightforward arguments
analogously to \eqref{derS}--\eqref{Sdec}.
\end{proof}

\medskip

Notice that the estimate \eqref{0.2} is an additional information
that we extracted from the Bakry-\'Emery argument by using the second
term in the Hessian estimate \eqref{convex}. 
It plays a key role in the next theorem.

\begin{theorem}[Dirichlet form inequality]\label{thm3}
Let $q$ be a probability density $\int q\rd\om=1$ and
let  $O:\bR\to\bR$ be a  smooth function with compact support.
Then  for any $J\subset \{ 1, 2, \ldots , N-1\}$ and any $t>0$ we have
\be\label{diff}
\Big| \int \frac 1 {|J|} \sum_{i\in J} O(N(x_i - x_{i+1})) q \rd \omega -
\int \frac 1 {|J|} \sum_{i\in J} O(N(x_i - x_{i+1}) ) \rd \omega \Big|
\le C \Big( t  \frac { D_\omega (\sqrt {q})  }{|J|}  \Big)^{1/2}  + C \sqrt
{S_\om(q)} e^{-c t/\tau} .
\ee
\end{theorem}

{\it Proof.} For simplicity, we  assume that $J =  \{ 1, 2, \ldots , N-1\}$. 
Let $q_t$ satisfy
\[
\partial_t q_t = \wt \cL q_t, \qquad t\ge 0,
\]
with an initial condition $q_0=q$.
We write
\begin{align}\label{splitt}
\int \Big[  \frac 1 {|J|} & \sum_{i\in J} O(N(x_i - x_{i+1}))\Big] (q-1)\rd\om \\
&  =  \int \Big[  \frac 1 {|J|} \sum_{i\in J}O(N(x_i - x_{i+1}))\Big] (q-q_t)\rd\om
 +   \int \Big[  \frac 1 {|J|} \sum_{i\in J} O(N(x_i - x_{i+1})) \Big] (q_t-1)\rd\om.
\nonumber
\end{align}
The second term in \eqref{splitt} can be estimated by  \eqref{entropyneq},
 the decay of the entropy \eqref{Sdecay} and the boundedness of $O$; 
this gives the second term in \eqref{diff}.

To estimate the first term in \eqref{splitt},
 by the evolution equation  $\pt q_t =\wt \cL q_t$ and the definition of $\wt \cL$ we have 
\begin{align}
\int  \frac 1 {|J|} \sum_{i\in J} &O(N(x_i - x_{i+1}))q_t \rd \omega  -
\int  \frac 1 {|J|} \sum_{i\in J} O(N(x_i - x_{i+1}))q_0 \rd \omega \non\\
&= \int_0^t  \rd s \int   \frac 1 {|J|} \sum_{i\in J}
   O'( N(x_i-x_{i+1}))
[\pt_{i} q_s - \pt_{i+1}q_s]  \rd \omega. \non
\end{align}
{F}rom the Schwarz inequality and $\pt q = 2 \sqrt{q}\pt\sqrt{q}$,
the last term is bounded by
\begin{align}\label{4.1}
2 \Bigg[   \int_0^t  \rd s \int_{\bR^N}  &
 \frac {N^2} {|J|^2} \sum_{i\in J} \Big[ O' (N(x_i - x_{i +1}))\Big] ^2 
(x_{i}-x_{i+1})^2  \, q_s \rd \omega
\Bigg]^{1/2}  \\
 &\times \left [ \int_0^t  \rd s \int_{\bR^N}  \frac 1 {N^2 } \sum_i
\frac{1}{(x_{i}-x_{i+1})^2}  [ \pt_{i}\sqrt {q_s} -
\pt_{i+1}\sqrt {q_s}]^2  \rd \omega \right ]^{1/2} 
\le   \; C \Big( \frac{D_\omega(\sqrt {q_0}) t}{|J|}\Big)^{1/2}, \nonumber
\end{align}
where we have used \eqref{0.2} and that
$ \Big[ O' (N(x_i - x_{i +1}))\Big]^2
 (x_{i}-x_{i+1})^2 \le CN^{-2}$
due to  $O$ being smooth and  compactly
supported.
\qed

\medskip

Alternatively, we could have directly estimated  the left hand side 
of \eqref{diff} by using the total variation norm between $q \om$ and $\om$, 
which  in turn could be estimated by the entropy 
  \eqref{entropyneq} and the Dirichlet form using the logarithmic Sobolev
inequality,  i.e., 
by 
\be\label{entbound}
C \int | q - 1| \rd\om \le  C \sqrt { S_\om(q) } \le C \sqrt { \tau D_\om (\sqrt q)}.
\ee
However, compared with this simple bound, the estimate
  \eqref{diff} gains an extra factor $|J|\sim N$ in the denominator,
 i.e. it is in terms of Dirichlet form {\it per particle}.
The improvement
 is due to the fact that the observable in \eqref{diff} depends only on
the gap, i.e. {\it difference} of points. This allows us to exploit the additional
term \eqref{0.2} gained in the Bakry-\'Emery argument. This is a
manifestation of the general observation that gap-observables behave much
better than point-observables.

\bigskip

The final  ingredient in proving Theorem \ref{thmM} is the following entropy and
Dirichlet form estimates.

\begin{theorem}\label{thm1}
Suppose that  \eqref{convex} holds. Let $\fa>0$ be fixed and
recall the definition of $Q=Q_\fa$ from \eqref{assum3}.
Fix a constant $\tau \ge  N^{-2 \fa}$ and consider the local relaxation 
measure $\omega$ with this $\tau$. 
Set $\psi:=\om/\mu$ and
let  $g_t: = f_t/\psi$.
Suppose there is a constant $m$ such that 
\be\label{entA}
S(f_{\tau} \om | \om )\le CN^m.
\ee
Then for any $ t \ge \tau N^\e$
the entropy and the Dirichlet form satisfy the estimates:
\be\label{1.3}
S(g_t \omega | \omega) \le
 C   N^2    Q \tau^{-1}, \qquad
D_\omega (\sqrt{g_t})
\le CN^2  Q \tau^{-2}
\ee
 where the constants depend on $\e$ and $m$.
\end{theorem}

{\it Proof.}  
The evolution of the entropy $S(f_t \mu |\om)= S_{\omega} (g_t) $
 can be computed explicitly by the formula \cite{Y}
$$
\partial_t  S(f_t \mu |\om)  = -  \frac{2}{N}  \sum_{j} \int (\partial_j
\sqrt {g_t})^2  \, \psi \, \rd\mu
+\int g_t  \cL \psi \, \rd\mu.
$$
Hence
we have, by using \eqref{tl},
$$
\pt_t S(f_t \mu |\om)
= - \frac{ 2}{N}  \sum_{j} \int (\partial_j
\sqrt {g_t})^2  \, \rd\omega
+\int   \wt \cL g_t  \, \rd\omega+  \sum_{j} \int  b_j \partial_j
g_t   \, \rd\omega.
$$
Since $\om$ is $\wt \cL$-invariant and time independent,
 the middle term on the right hand side vanishes, and
from the Schwarz inequality
\be\label{1.1}
 \pt_t S(f_t \mu |\om) \le  -D_{\omega} (\sqrt {g_t})
+   C  N\sum_{j} \int  b_j^2  g_t   \, \rd\omega \le
-D_{\omega} (\sqrt {g_t})
+   C  N^2 Q \tau^{-2} .
\ee
Notice that \eqref{1.1} is reminiscent to \eqref{derS}
for the derivative of the entropy of the measure $g_t\om=f_t\mu$
with respect to $\om$.
The difference is, however, that $g_t$ does not satisfy
the evolution equation with the generator $\wt\cL$.
The last term in \eqref{1.1} expresses the error.

Together with the logarithmic Sobolev inequality  \eqref{lsi},  we have
\be\label{1.2new}
\partial_t S(f_t \mu |\om)\le  -D_{\omega} (\sqrt {g_t})
+   C  N^2 Q \tau^{-2}  \le  - C\tau^{-1}  S(f_t \mu |\om)+
 C  N^2  Q \tau^{-2}.
\ee
Integrating the last 
inequality  from $\tau$ to $t$ and using the assumption \eqref{entA}
and $t\ge \tau N^\e$, 
we have proved the first inequality of \eqref{1.3}. 
Using this result and integrating \eqref{1.1}, we have 
$$
\int_\tau^t  D_{\omega} (\sqrt {g_s})  \rd s  \le  C   N^2    Q \tau^{-1}. 
$$
By the convexity of the Hamiltonian, $ D_{\mu} (\sqrt {f_t}) $  is decreasing  in $t$. 
Since $ D_{\omega} (\sqrt {g_s}) \le  CD_{\mu} (\sqrt {f_s}) +  CN^2 Q \tau^{-2}$,
 this proves 
the second  inequality of \eqref{1.3}.
\qed

\bigskip
Finally, we complete the proof of Theorem \ref{thmM}. For any given $t>0$ we
 now choose 
$\tau: =  t N^{-\e}$ and we construct the local relaxation measure $\om$
with this $\tau$. Set $\psi= \om/\mu$ and
let $q:= g_{t}= f_{t}/\psi$  be the density $q$ in
Theorem \ref{thm3}.
 Then  Theorem \ref{thm1}, Theorem \ref{thm3}
 and an easy bound on the entropy $S_\om(q)\le CN^m$
imply  that
\be\label{diff1}
\Big| \int \frac 1 N \sum_{i\in J} O(N(x_i - x_{i+1}))(   f_t \rd \mu -
  \rd \omega) \Big| 
\le C \Big( t  \frac { D_\omega (\sqrt {q})  }{|J|}  \Big)^{1/2}  + C \sqrt
{S_\om(q)} e^{-c N^\e} .
\ee
\be\nonumber
\le C \Big( t  \frac { N^2 Q  }{|J| \tau^2 }  \Big)^{1/2}+  Ce^{-c N^{\e}}
    \le 
C N^{\e} \sqrt{  \frac {N^2 Q} { |J| t }}
+  Ce^{-c N^{\e}} , 
\ee
i.e.,  the local statistics of $f_t \mu$ and $\om$ are the same for any
 initial data  $f_\tau$ for which  
\eqref{entA} is satisfied. {Applying the same argument to 
the Gaussian initial data, $f_0=f_\tau=1$,
we can also compare $\mu$ and $\om$.}
   We have thus  proved \eqref{GG} and hence the 
universality. 
\qed

\subsection{The Green function comparison theorems and four moment matching}\label{sec:4mom}

We now state the Green function comparison theorem, Theorem~\ref{comparison}.
It will  quickly lead to Theorem~\ref{com} stating that 
the correlation functions  of eigenvalues of two matrix ensembles 
are identical  on a scale smaller than $1/N$  provided that the first four moments 
of all  matrix elements  of these two ensembles are almost the same.  We will state 
a limited version  for real Wigner matrices for simplicity of  presentation.

\begin{theorem}[Green function comparison]\label{comparison}\cite[Theorem 2.3]{EYY}  
 Suppose that we have  
two  $N\times N$ Wigner matrices, 
$H^{(v)}$ 
and $H^{(w)}$, with matrix elements $h_{ij}$
given by the random variables $N^{-1/2} v_{ij}$ and 
$N^{-1/2} w_{ij}$, respectively, with $v_{ij}$ and $w_{ij}$ satisfying
the uniform subexponential decay condition \eqref{subexp}.
We assume that the first four moments of
  $v_{ij}$ and $w_{ij}$ are  close to each other in the sense that 
\be\label{4}
    \big | \E  v_{ij}^s  -  \E  w_{ij}^s \big | \le N^{-\delta -2+ s/2},
  \qquad 1\le s\le 4,
\ee
holds for some $\delta > 0$. 
Then  there are positive constants $C_1$ and $\e$,
depending on $\ttau$ and $C_0$ from  \eqref{subexp} such that for any $\eta$ with
$N^{-1-\e}\le \eta\le N^{-1}$  and for any $z_1, z_2$ with $\im z_j = \pm \eta$, $j=1,2$,
we have 
\begin{align}\label{maincomp}
\lim_{N \to \infty} \Big [ \E \tr  G^{(v)}(z_1)  \tr G^{(v)}(z_2)    - \E \tr  G^{(w)}(z_1)  \tr G^{(w)}(z_2)\Big ] 
= 0,
\end{align}
 where $G^{(v)}$ and $G^{(w)}$ denotes the Green functions of $H^{(v)}$ and $H^{(w)}$.
\end{theorem}

Here we formulated Theorem \ref{comparison} for a product of two
traces of  the Green function, but  the result 
holds   for  a large class of
smooth functions  depending on several individual
matrix elements of the  Green functions as well,
 see \cite{EYY}  for the precise statement.  (The 
matching condition \eqref{4}  is slightly weaker than in \cite{EYY}, 
but the proof in \cite{EYY} without any change  yields this slightly stronger version.) 
This general version of
 Theorem \ref{comparison} implies the correlation functions 
  of the two ensembles at the scale $1/N$ are identical:

\begin{theorem}[Correlation function comparison]\label{com} \cite[Theorem 6.4]{EYY}
Suppose the assumptions of Theorem \ref{comparison} hold. 
Let $p_{v, N}^{(n)}$ and $p_{w, N}^{(n)}$
be the  $n-$point functions of the eigenvalues w.r.t. the probability law of the matrix $H^{(v)}$
and $H^{(w)}$, respectively. 
Then for any $|E| < 2$,  any
$n\ge 1$ and  any compactly supported continuous test function
$O:\bR^n\to \bR$ we have   
\be \label{6.3}
\lim_{N\to\infty}\int_{\R^n}  \rd\alpha_1 
\ldots \rd\alpha_n \; O(\alpha_1,\ldots,\alpha_n) 
   \Big ( p_{v, N}^{(n)}  - p_{w, N} ^{(n)} \Big )
  \Big (E+\frac{\alpha_1}{N}, 
\ldots, E+\frac{\alpha_n}{N }\Big) =0.
\ee
\end{theorem}
Notice that these comparison theorems hold for any fixed energy $E$, i.e. 
no averaging in energy is necessary.

The basic idea for proving Theorem \ref{comparison} is
similar to Lindeberg's proof of the central limit theorem, where
the random variables are replaced one by one with a Gaussian one.
We will replace  the matrix elements $v_{ij}$ with $w_{ij}$ one by one and
 estimate the effect of this change on the resolvent 
by a resolvent expansion. 
The idea of applying Lindeberg's method in random matrices
was recently used  by Chatterjee \cite{Ch}  for comparing  the traces of the Green functions;
the idea was also used by Tao and Vu \cite{TV}
in the context of comparing individual eigenvalue
distributions.

\medskip

The four moment matching condition  \eqref{4} with $\delta=0$ first appeared in \cite{TV}.
 For comparison with Theorem~\ref{comparison}, \nc we state here the main result of \cite{TV}.
Let $\lambda_1<\lambda_2 <\ldots < \la_N$ and
$\lambda_1'<\lambda_2' <\ldots < \la_N'$ denote the
eigenvalues of $H$ and $H'$, respectively. Then
 the joint distribution of any $k$-tuple of
eigenvalues on scale $1/N$ is very close to each other in the following sense:

\begin{theorem}[Four moment  theorem for eigenvalues]\cite[Theorem 15]{TV}\label{thm:TV}
Let $H$ and  $H'$ be two Wigner matrices.
 Assume that the first four moments of $h_{ij}$ and $h'_{ij}$ exactly match and 
 the subexponential decay \eqref{subexp}   holds for the single entry distributions.
 Then for any sufficiently small positive $\e$ and $\e'$
and for any function $F:\R^k\to\R$ satisfying $|\nabla^j F|\le N^\e$, $j\le 5$,
and for any selection of $k$-tuple of indices $i_1, i_2, \ldots, i_k\in[\e N,
(1-\e)N]$ away from the edge, we have
\be
   \Bigg| \E F\Big( N\la_{i_1}, N\la_{i_2}, \ldots, N\la_{i_k}\Big) 
  -  \E' F\Big( N\la'_{i_1}, N\la'_{i_2}, \ldots, N\la'_{i_k}\Big) \Bigg|\le N^{-c_0}
\label{FF}
\ee
with some $c_0>0$.
The exact moment matching condition  can be relaxed to \eqref{4},
but $c_0$ will depend on $\delta$.
\end{theorem}

Note that the arguments in \eqref{FF} are magnified by a factor $N$ so
the result is sufficiently precise to detect individual eigenvalue correlations.
Therefore Theorem~\ref{comparison} or 
\ref{thm:TV} can prove bulk universality for a Wigner matrix $H$
if another $H'$ is found, with matching four moments, for which universality is
already proved.  This will be explained in Section~\ref{sec:put}.

 Both  Theorem~\ref{comparison} and Theorem~\ref{thm:TV}  rely on some version of
the local semicircle law on the shortest possible scale. 
   There are, however, three main differences between them. 
\begin{itemize}
\item[(i)]   Theorem~\ref{comparison} compares the statistics of eigenvalues  of two different ensembles
 near {\it fixed energies} while  Theorem~\ref{thm:TV} 
 compares the statistics of the
$j_1, j_2, \ldots j_k$-th  eigenvalues for {\it fixed labels} $j_1, j_2, \ldots j_k$.
\item[(ii)] Both theorems are of perturbative nature that require some
apriori information.  Theorem~\ref{comparison} uses a bound on the
resolvent matrix entries, $|G_{ij}(z)|$, that has already been obtained in the
local semicircle law (see, e.g. \eqref{Gijest 2}). Theorem~\ref{thm:TV}
needs an apriori lower bound on the gaps 
 to exclude possible eigenvalue resonances
that may render the expansion unstable. This is achieved by a level repulsion
estimate that is the most complicated technical part of \cite{TV}.
Previously, even more precise level repulsion estimates were obtained in \cite{ESY3}
but only for smooth distributions.
\item[(iii)]  Theorem~\ref{comparison}  also compares off-diagonal Green function elements,
an information that cannot be obtained from Theorem~\ref{thm:TV}.
Hence it directly provides information on the eigenvectors as well, 
see \cite{KY} for the development.  In fact, once  Theorem~\ref{comparison} is proved
for all energies, it also implies  Theorem~\ref{thm:TV}. The reason is that we can integrate
correlation functions in energy with a precision smaller than the typical size of the gap,
hence eigenvalues with a {\it fixed label} can be identified. This was first
done near the edge in \cite{EKYY2}  and later in the bulk in \cite{KY}.
 \end{itemize}

\medskip
{\it Sketch of the proof of Theorem~\ref{comparison}.}
 We fix a bijective ordering map on the index set of
the independent matrix elements,
\[
\phi: \{(i, j): 1\le i\le  j \le N \} \to \Big\{1, \ldots, \gamma(N)\Big\} , 
\qquad \gamma(N): =\frac{N(N+1)}{2},
\] 
and denote by  $H_\gamma$  the  Wigner matrix whose matrix 
elements $h_{ij}$ follow
the $v$-distribution if $\phi(i,j)\le \gamma$ and they follow the $w$-distribution
otherwise; in particular $H^{(v)}= H_0$ and $H^{(w)}= H_{\gamma(N)}$. 

Consider the telescopic sum of 
differences of expectations (we present only  one resolvent for simplicity of the presentation):
\begin{align}\label{tel}
\E \, \left ( \frac{1}{N}\tr  \frac 1 {H^{(w)}-z} \right )   - 
 & \E \,   \left  (  \frac{1}{N}\tr  \frac  1 {H^{(v)}-z} 
\right )  \\
= & \sum_{\gamma=1}^{\gamma(N)}\left[  \E \, 
 \left (  \frac{1}{N}\tr \frac 1 { H_\gamma-z} \right ) 
-  \E \,  \left (  \frac{1}{N}\tr \frac  1 { H_{\gamma-1}-z} \right ) \right] . \non
\end{align}
Let $E^{(ij)}$ denote the matrix whose matrix elements are zero everywhere except
at the $(i,j)$ position, where it is 1, i.e.,  $E^{(ij)}_{k\ell}=\delta_{ik}\delta_{j\ell}$.
Fix a $\gamma\ge 1$ and let $(i,j)$ be determined by  $\phi (i, j) = \gamma$.
We will compare $H_{\gamma-1}$ with $H_\gamma$.
Note that these two matrices differ only in the $(i,j)$ and $(j,i)$ matrix elements 
and they can be written as
$$
    H_{\gamma-1} = Q + \frac{1}{\sqrt{N}}V, \qquad V:= v_{ij}E^{(ij)}
+ v_{ji}  E^{(ji)}, \qquad v_{ji}:= \ov v_{ij},
$$
$$
    H_\gamma = Q + \frac{1}{\sqrt{N}} W, \qquad W:= w_{ij}E^{(ij)} +
   w_{ji} E^{(ji)}, \qquad w_{ji}:= \ov w_{ij},
$$
with a matrix $Q$ that has zero matrix element at the $(i,j)$ and $(j,i)$ positions.

By the resolvent expansion, 
\[
  S_{\gamma-1}= R - N^{-1/2} RVR + \ldots  + N^{-2} (RV)^4R - N^{-5/2} (RV)^5 S,\quad 
    R := \frac{1}{Q-z}, \; \; S_{\gamma-1} :=  \frac{1}{H_{\gamma-1}-z},
\]
and a similar expression holds for the resolvent $S_\gamma$ of  by $H_{\gamma}$. 
 From the local semicircle law  for individual matrix elements \eqref{Gijest 2},
 the matrix elements of all Green functions $R$,  $S_{\gamma-1}, S_\gamma$  are bounded  
by $CN^{2\e}$. Although \eqref{Gijest 2} is not directly applicable
to $\eta\ge N^{-1+\e}$, it is easy to show that
$$
  |G_{ij}(E+i\eta)|\le \max_i |G_{ii}(E+i\eta)| \le \frac{\eta'}{\eta}  \max_i |G_{ii}(E+i\eta')|
$$
so choosing $\eta'\sim N^{-1+\e}$ we can prove a bound for $\eta$ slightly below
$1/N$ at the  expense of a factor $\eta/\eta'$. The estimates of the related resolvents
$R$,  $S_{\gamma-1}, S_\gamma$ are similar.

By assumption \eqref{4},  the difference between 
the expectation of matrix elements of  $S_{\gamma-1}$ and $S_\gamma$
 is of order $N^{-2-\delta + C \e}$.
 Since the number of steps, $\gamma(N)$ is of order $N^2$,  the
difference in \eqref{tel} is of order $N^{2}N^{-2 - \delta+ C \e}\ll 1$, and
this proves Theorem \ref{comparison}  for a single resolvent. 
It is very simple to turn this heuristic argument into a rigorous proof
and to generalize it to the product of several resolvents.
 The real difficulty is the input that the resolvent entries can
be bounded for a general class of Wigner matrices down to the almost optimal scale
$\eta\sim 1/N$.

\subsection{Universality for generalized Wigner matrices: putting it together}\label{sec:put}

In this short section we put the previous information together to prove 
Theorem \ref{bulkWigner}.  We first focus on the case when $b_N$ is independent of $N$. 
Recall that Theorem~\ref{thm:DBM} states that the correlation 
functions of the Gaussian divisible ensemble,
\be\label{matrixdbm1}
H_t = e^{-t/2} H_0 + (1-e^{-t})^{1/2}\, U,
\ee
where $H_0$  is the initial  Wigner matrix and 
$U$ is an independent standard GUE (or GOE) matrix,
are given by the corresponding GUE (or GOE) for $t\ge  N^{-2\fa+\e}$ 
provided that the a-priori estimate \eqref{assum3}  holds
for the solution $f_t$ of the forward equation \eqref{dy} with some exponent $\fa>0$.
Since the rigidity  of eigenvalues \eqref{rigwig} holds uniformly for all generalized Wigner matrices, 
we have proved  \eqref{assum3} for $\fa = 1/2 - \e$  with any $\e>0$.

{F}rom the evolution of the OU process \eqref{zij} for $v_{ij}  = N^{1/2} h_{ij}$ we have 
\be\label{smom}
\big | \E v_{ij}^s(t) -  \E v_{ij}^s(0) \big | \le C t = C N^{-1 + 3 \e}
\ee
for  $s=3, 4$ and with the choice of
 $t = N^{-1 + 3\e}$. 
 Furthermore, $\E h_{ij}^s(t)$ are independent of $t$ for $s=1, 2$
due to $\E v_{ij}(0) = 0$ and $\E v_{ij}^2(t) = 1$. 
Hence \eqref{4} is satisfied  for the matrix elements of $H_t$ and $H_0$ and 
we can thus use Theorem \ref{com} to conclude that the correlation functions of $H_t$ and $H_0$
are identical at the scale $1/N$. 
Since the correlation functions of $H_t$ are given by the corresponding Gaussian 
case, we have proved Theorem \ref{bulkWigner} under the condition
 that the probability distribution of the matrix elements  decay subexponentially. 
Finally, we need a technical cutoff argument to relax the decay condition
to \eqref{4+e} which  we omit here  (see Section 7 in \cite{EKYY2}).

 The argument for $N$-dependent $b=b_N$ in the range
 $b_N\ge N^{-1+\xi}$, $\xi>0$, is slightly different. For such a small 
$b_N$, \eqref{abstrthm} could be established only for relatively large times, $t\ge N^{-\xi/8}$.
We cannot therefore compare $H_0$ with $H_t$ directly, since the deviation of the third
moments of $v_{ij}(0)$ and $v_{ij}(t)$ in \eqref{smom} would not satisfy
\eqref{4}. Instead, we construct an auxiliary Wigner matrix $\wh H_0$ such that  
up to the third moment {\it its} time evolution $\wh H_t$ under the OU flow \eqref{matrixdbm1}
 matches exactly the {\it original} matrix $H_0$ and the fourth moments are close
 even for $t$ of order $N^{-\xi/8}$
(see Lemma 3.4 of \cite{EYY2}).
Theorem~\ref{thm:DBM} will then be
applied for $\wh H_t$, and Theorem~\ref{comparison} can be used to compare  $\wh H_t$
and $H_0$. This completes the sketch of the proof of Theorem \ref{bulkWigner}.

\section{Universality of the correlation functions for $\beta$-ensembles} \label{sec:unibeta}

In this section we outline the proof of Theorem~\ref{bulkbeta}.
We will use the notation $\mu=\mu_{\beta, V}^{(N)}$ for the probability measure defining
the  general $\beta$-ensemble with a potential $V$ on $N$ ordered real points
$\lambda_1\leq \ldots\leq \lambda_N$, see  \eqref{01}.
We let $\P^\mu$ and $\E^\mu$  denote the probability and the expectation with respect 
to $\mu$. The equilibrium density is denoted by $\varrho=\varrho_V$
and its Stieltjes transform by 
$$
  m(z)=m_V(z)= \int_\bR \frac{\varrho_V(x)}{x-z}\rd x.
$$
The classical location of the $k$-th point will be denoted by  $\gamma_k=\gamma_{k, V}$,
 see \eqref{defgammagen}. Note that in this section $\mu$, $\varrho$, $m$ and $\gamma_k$
refer to the quantities related to the general $V$ and not the Gaussian one
as in the previous sections.  This avoids carrying the $V$ subscripts all the time
as in Section~\ref{sec:lsc} we dropped the subscripts {\it sc} referring to
the semicircle law.

\subsection{Rigidity  estimates}   \label{beta}

For simplicity of  presentation we assume that 
 the potential $V$ is convex,  i.e., 
\begin{equation}\label{eqn:LSImu}
\varpi :=\frac{1}{2} \inf_{x\in\RR}V''(x) > 0,
\end{equation}
the equilibrium density  $\varrho(s)$ is supported on a single interval 
$[A,B]\subset \bR$ and satisfies \eqref{equilibrium}
(for the general case, see \cite{BEY2}). 
The  Gaussian case corresponds to 
$V(x)=x^2/2$,
in which case  the equilibrium density is
the semicircle law, $\varrho_{sc}$, given by  \eqref{sc}.
Our  main result concerning the universality is Theorem \ref{bulkbeta} and 
similar  statement holds for the 
universality of the averaged
gap distributions directly. In fact, the proof of Theorem \ref{bulkbeta} goes via 
the averaged gap distribution as we now  explain.

The first step to prove Theorem \ref{bulkbeta}  is the following theorem 
which provides a rigidity estimate  on the location of each individual point.
The precision in the bulk is almost down to the optimal scale $1/N$, the
estimate is weaker near the edges.
In the following, we will denote $\llbracket x,y\rrbracket=\NN\cap[x,y]$.

\begin{theorem}\label{thm:accuracy}\cite[Theorem 3.1 and Lemma 3.6]{BEY}
Fix any $\alpha, \e >0$ and assume that \eqref{eqn:LSImu} holds. Then there are constants
$\delta,c_1,c_2>0$ such that for any $N\geq 1$ and $k\in\llbracket \alpha N,(1-\alpha) N\rrbracket$,
\be\label{bulkacc}
\P^\mu\left(|\lambda_k-\gamma_{k}|> N^{-1+\e}\right)\leq c_1e^{-c_2N^\delta}.
\ee
The following weaker bound is valid close to the spectral edges:
\be\label{edgeacc}
\P^\mu \left ( |\lambda_k-\gamma_k|\ge  N^{-4/15+\epsilon}\right)  \le c_1 e^{ - c_2  N^\delta}  \; .
\ee
for  any
$ k\in \llbracket N^{3/5+ \e},  N - N^{3/5+ \e}\rrbracket $. Finally, the bound
\be\label{allacc}
\P^\mu \left ( |\lambda_k-\gamma_k|\ge \e \right)  \le c_1 e^{ - c_2  N^\delta}  \; .
\ee
holds for any $k\in \llbracket 1, N \rrbracket$.
\end{theorem}

\bigskip 

We explain some ideas of the proof of \eqref{bulkacc}, the arguments
for \eqref{edgeacc} are similar and \eqref{allacc} follows from an easy large deviation bound, see \cite{BEY}.
The first ingredient to prove \eqref{bulkacc} is an analysis of the loop equation
following Johansson \cite{Joh} and Shcherbina \cite{Sch}. 
 The equilibrium density $\varrho$, for a convex potential $V$,  is given by 
 \begin{equation}\label{eqn:rho}
\varrho(t) =\frac{1}{\pi}r(t)\sqrt{(t-A)(B-t)}\mathds{1}_{[A,B]}(t),
\end{equation}
where $r$ is a real function  that
can be extended to an analytic function in $\CC$  and $r$ 
has no zero in $\RR$.
Denote by  $s(z):=-2r(z)\sqrt{(A-z)(B-z)}$ where the square root is defined such
 that its asymptotic value 
is $z$ as $z\to\infty$. Recall that the density
is the one-point correlation function which is characterized  by
\be
\int_\bR \rd \lambda_1   O(\lambda_1)  p^{(1)}_N(\lambda_1) =
\int_{\bR^N}  O(\lambda_1) 
 \rd\mu_{\beta, V}^{(N)}(\lambda), \qquad \lambda = (\la_1, \la_2, \ldots, \la_N).
\label{def:dens}
\ee
  Let  $\bar m_N$ and $m$ be the Stieltjes transforms of the  density $p_N^{(1)}$
and the equilibrium density $\varrho$, respectively.  Notice that
in Section~\ref{sec:lsc}  we have used $m_N$ to denote 
the Stieltjes transform of the  empirical measure \eqref{mNdef}; 
here $\bar m_N$ denotes the ensemble average of the analogous quantity.

Define the analytic functions 
$$
b_N(z):=\int_{\RR}\frac{V'(z)-V'(t)}{z-t}(p_1^{(N)}-\varrho)(t)\ind t 
$$
and  
$$
   c_N(z):=\frac{1}{N^2}k_N(z)+\frac{1}{N}\left(\frac{2}{\beta}-1\right)\bar m_N'(z), \quad \mbox{with}
\quad k_N(z):=\var_\mu\left(\sum_{k=1}^N\frac{1}{z-\lambda_k}\right).
$$
Here for  complex random variables $X$ we use the definition that 
$\var(X)=\E(X^2)-\E(X)^2$.

 The equation used by Johansson 
(which can be obtained by a change of variables
in \eqref{def:dens}
 \cite{Joh} or by integration by parts \cite{Sch}), is
 a variation of the loop equation (see, e.g., \cite{Ey})
used in the physics literature and it takes the form
\begin{equation}\label{eqn:firstLoop}
( \bar m_N-m)^2+s(\bar  m_N-m)+b_N=c_N.
\end{equation}

Equation  \eqref{eqn:firstLoop} can be used to  express the difference $\bar m_N - m$ 
in terms of $(\bar m_N-m)^2$, $b_N$ and $c_N$.
In the  regime where $|\bar m_N - m|$  is small, we 
can neglect the quadratic term.  The term $b_N$ is of the same order as
$|\bar m_N-m|$ and is  difficult  to treat. As observed in \cite{APS,Sch}, for analytic  $V$,
  this term vanishes
when we perform a contour integration. So we have roughly the relation
\be\label{55}
( \bar m_N-m) \sim \frac 1 { N^2} \var_\mu\left(\sum_{k=1}^N\frac{1}{z-\lambda_k}\right),
\ee
where we dropped the less important error involving $\bar m_N'(z)/N $ due to the extra $1/N$ factor.
 In the convex setting,
the variance can be estimated by  the logarithmic Sobolev
inequality and we immediately obtain an estimate on $\bar m_N - m$.
We then use the Helffer-Sj\"ostrand formula, see \eqref{1*}, 
 to  estimate the locations of the particles. 
 This will provide us with an accuracy of order $N^{-1/2}$
for $\E^\mu \lambda_k-\gamma_k$.  This argument gives only
an estimate on the expectation of the locations of the particles  since 
we only have  information
 on the  averaged quantity, $\bar m_N$. 
 Although it is tempting
to use  this new accuracy information on the particles
 to estimate the variance again in \eqref{55}, the information on 
 the expectation on $\lambda_k$ alone is very difficult to use 
in a bootstrap argument.
To estimate the variance of a non-trivial function of $\lambda_k$
 we need high probability estimates 
on  $\lambda_k$.

The key idea in this section is the observation that the 
accuracy information on the $\lambda$'s can be used to  improve the local convexity
of the measure $\mu$ in  the direction involving the {\it differences} of $\lambda$'s.
To explain this idea, we compute the  Hessian of the Hamiltonian of $\mu$:
$$
\Big\langle \bv , \nabla^2  \cH(\la)\bv\Big\rangle
\ge   \varpi  \,   \|\bv\|^2 + \frac{1}{N}
 \sum_{i<j} \frac{(v_i - v_j)^2}{(\la_i-\la_j)^2}. 
$$
The naive lower bound on $\nabla^2  \cH$ is $\varpi$, but for a
 typical $\bla=(\la_1, \la_2, \ldots, \la_N)$ it is in fact much better
in most directions. To see this effect,  suppose  
we know $ |\la_i-\la_j| \lesssim M /N$ 
 with some $M$ for any  $i, j  \in I_k^M$, where
$I_k^M := \llbracket k-M, k+M\rrbracket$. 
Then for $\bv = (v_{k-M}, \ldots, v_{k+M})$
with $\sum_j v_j = 0$ we have 
\be\label{convex2}
\Big\langle \bv , \nabla^2  \cH(\bla)\bv\Big\rangle
\ge \frac{N}{M^2}
 \sum_{i, j \in I_k^M}  (v_i - v_j)^2 \ge C \frac N M \sum_j v_j^2.  
\ee
 This improves the convexity of the Hessian to $N/M$  on the hyperplane $\sum_j v_j = 0$.
Let 
$$
\lambda_k^{[M]}: = \frac{1}{ |I_k^M|}\sum_{j\in I_k^M}\lambda_j = \frac{1}{2M+1}\sum_{j\in I_k^M}\lambda_j
$$
denote the block average 
of the locations of particles and rewrite 
$$ 
 \lambda_k - \lambda_k^{[N^{1-\e}]}=  \sum_j \Big( \lambda_k^{[M_j]}- \lambda_k^{[M_{j+1}]}\Big)
$$
as a telescopic sum with an appropriate sequence of $M_1=0$, $M_2, \ldots$ with the property
that $M_j/M_{j-1}\le N^\e$.
We can now use the improved concentration on the hyperplane $\sum_j v_j = 0$ to the 
variables $\lambda_k^{[M_j]}- \lambda_k^{[M_{j+1}]}$ to
 control the fluctuation of $ \lambda_k - \lambda_k^{[N^{1-\e}]}$. 
Since the fluctuation of $  \lambda_k^{[N^{1-\e}]}$ is very small for small $\e$, 
we finally arrive at the estimate 
\be\label{concen1}
\P^\mu \left(|\la_k-\E^\mu(\la_k)|> a \right)\leq C e^{- C  N^2 a^2/ M }.
\ee

{F}rom \eqref{concen1} we thus  have that $|\la_k - \E^\mu \la_k| \lesssim \sqrt M /N$
with high probability. This improves the starting accuracy 
$ |\la_i-\la_j| \lesssim M /N$  for $i, j  \in I_k^M$  to $ |\la_i-\la_j| \lesssim M' /N$
with some $M'\ll M$,
 provided that we can prove  that  $ |\E^\mu (\la_i-\la_j )| \ll  M' /N$.
 But the last inequality 
involves only expectations and it
will follow from the analysis of the loop equation \eqref{eqn:firstLoop} 
we just mentioned  above.  Starting from $M=N$, this procedure can
 be repeated by decreasing $M$ step by step until we get the optimal 
accuracy, $M \sim O(1)$.
The implementation of this argument in \cite{BEY} is somewhat
different from this sketch due to various technical issues, but it follows the same basic idea.

\subsection{The local equilibrium measure }\label{sec:loceq}

Having completed the first step, the rigidity
  estimate, we now focus on the second
step, i.e. on the uniqueness of the local Gibbs measure.
Let $0<\kappa<1/2$. Choose $q\in [\kappa, 1-\kappa]$ and set $L=[Nq]$ (the integer part).
 Fix an integer  $K=N^k$ with $k<1$.
 We will study the local spacing statistics
of $K$ consecutive particles
$$
  \{ \lambda_j\; : \; j\in I\}, \qquad I=I_L:=
 \llbracket L+1, L+K \rrbracket.
$$
These particles are typically located near
$E_q$ determined by the relation
$$
  \int_{-\infty}^{E_q} \varrho(t) \rd t = q.
$$
Note that $|\gamma_L- E_q|\le C/N$.

We will distinguish the inside and outside particles
by renaming them as
\be\label{35}
(\lambda_1, \lambda_2, \ldots,
\lambda_N):=(y_{1}, \ldots y_{L}, x_{L+1},  \ldots, x_{L+K}, y_{L+K+1},
  \ldots y_{N}) \in \Xi^{(N)},
\ee
but note that they keep their original indices.
The notation $\Xi^{(N)}$ refers to the simplex
$\{\bz \; :\; z_1<z_2< \ldots < z_N\}$ in $\RR^N$.
In short we will write
$$
\bx=( x_{L+1},  \ldots, x_{L+K} ), \qquad \mbox{and}\qquad
 \by=
 (y_{1}, \ldots, y_{L}, y_{L+K+1},
  \ldots, y_{N}),
$$
all  in increasing order, i.e. $\bx\in \Xi^{(K)}$ and
$\by \in \Xi^{(N-K)}$.
We will refer to the $y$'s as   {\it external
points}  and to the $x$'s as  {\it internal points}.

We will fix the external points (also called
as boundary conditions) and study
conditional measures on the internal points.
We  define  the
{\it local equilibrium measure} on $\bx$ with fixed boundary condition  $\by$ by
\begin{equation}\label{eq:muydeold}
 \quad
\mu_{\by} (\rd\bx)  = \mu_\by(\bx) \rd \bx, \qquad
\mu_\by(\bx):=  \mu (\by, \bx) \left [ \int \mu (\by, \bx) \rd \bx \right ]^{-1}.
\end{equation}
Note that for any fixed $\by\in \Xi^{(N-K)}$,  
the measure $\mu_\by$ is supported on configurations of $K$ points 
$\bx=\{ x_j\}_{j\in I}$ 
 located in the interval $[y_{L}, y_{L+K+1}]$.

The Hamiltonian $\cH_\by$ of the measure $\mu_\by (\rd \bx) \sim \exp(-\beta N\cH_\by(\bx))\rd\bx$
 is given by 
\be\label{24}
\cH_{\by} (\bx) :=
 \sum_{i\in I}  \frac{1}{2}V_\by (x_i)
-  \frac{1}{N} \sum_{i,j\in I\atop i< j}
\log |x_{j} - x_{i}| 
\qquad \mbox{with}\qquad
V_\by (x): = V(x) - \frac{ 2 }{ N} \sum_{j \not \in I}
\log |x - y_{j}|.
\ee
We now define the set of {\it good boundary configurations}
with  a parameter $\delta=\delta(N)>0$
\begin{align}\label{goodset}
  \cG_{\delta}=\cG:=
\Big\{   \by \in \Xi^{(N-K)}\; :\; |y_j-\gamma_j|\le \delta, 
\; \forall \, j\in
\llbracket N\kappa/2, L\rrbracket \cup \llbracket L+K+1,
N(1-\kappa/2)\rrbracket
\Big\},
\end{align}
where $\kappa$ is a small constant to cutoff points near the spectral edges.  Some 
rather weak  additional conditions
 for $\by$ near the spectral edges will also be needed. They can be built
in the definition of $\cG$ based upon the bounds \eqref{edgeacc} and \eqref{allacc}
 but  we will neglect
this issue here. 

Let   $\sigma $ and $ \mu$ be two measures of the form \eqref{01}  
with  potentials $W$ and $V$ and 
 densities $\varrho=\varrho_W$ and $\varrho_V$, respectively.
For our purpose $W(x)=x^2/2$,  i.e., $\sigma$ 
is the Gaussian $\beta$-ensemble
and  its density $\varrho_W(t) =\frac{1}{2\pi}(4-t^2)^{1/2}_+$ is the Wigner semicircle law.
Let  the sequence $\gamma_j$ be the classical locations for $\mu$ and
the sequence $\theta_j$ be the classical locations for $\sigma$.
Similarly to the construction of the  measure $\mu_{\by}$,
 for any positive integer $L' \in \llbracket 1, N-K \rrbracket$  we can construct 
the measure $\sigma_\bt$ conditioned that the particles outside are given by 
the classical locations $\theta_j$ 
for $j \notin
\llbracket  L',  L'+K \rrbracket$. 
More precisely,  we define a {\it reference
local Gaussian  measure} $\sigma_\th \sim \exp(-\beta N\cH_\bt(\bx))\rd\bx$
 on the set $[\th_{L'}, \th_{L'+K+1}]$ via the  Hamiltonian
\be\label{241}
\cH_{\bt} (\bx) =
 \sum_{i \in I'}  \Big [  \frac{1}{4}   x_i^2   - \frac{ 1 }{N} \sum_{j \not \in I'}
\log |x_i - \th_{j} | \Big ]
-  \frac{ 1 }{N} \sum_{ i,j \in I'\atop i<j }
\log |x_{j} - x_{i}| ,
\ee
where  $I':= \llbracket L'+1, L'+K\rrbracket$.
Since $L'$ will not play an active role, we will 
abuse the notation and set $L'=L$. 

The measure $\mu_{\by}$ lives on the interval 
$[y_L, y_{L+K+1}]$ while the measure $\sigma_{\bt}$ lives on the interval 
$[\theta_L, \theta_{L+K+1}]$ and it is difficult to compare them. 
But after an appropriate translation and dilation, they will live on the same interval and from now on 
we assume that  $[y_L, y_{L+K+1}]= [\theta_L, \theta_{L+K+1}]$. 
The parameter $K=N^k$ has to be sufficiently small since  $\varrho_V$ and 
$\varrho_W$ are not constant functions
and we have 
to match these two densities quite precisely  in the whole interval. 
There are some other subtle  issues related to the rescaling, 
but  we will neglect them here to concentrate on the main ideas.  
Our main result is the following theorem
which is essentially a combination of Proposition 4.2 and Theorem 4.4 from  \cite{BEY}.

\begin{theorem}\label{thm:mi} 
Let $0<\varphi\le \frac{1}{38}$.
 Fix $K=N^k$, $\delta = N^{-1+\varphi}$
and $k=\frac{39}{2}\varphi$.
Then for $\by\in \cG_\delta$  we have 
\be\label{381}
\Bigg|  \E^{ \mu_{\by} } \frac 1 K \sum_{i \in I} O\Big( N(x_i-x_{i+1}) \Big)
   -\E^{ \sigma_{\th} }
 \frac 1 K \sum_{i \in I} O\Big( N(x_i-x_{i+1}) \Big)\Bigg|
\to 0
\ee
as $N \to \infty$
for any smooth and compactly supported test function $O$.
A similar formula holds for more complicated observables of the form
\eqref{cG}.
\end{theorem}

{F}rom the rigidity estimate, Theorem~\ref{thm:accuracy}, it
follows that $\cG_\delta$ has an overwhelming probability, so
the expectation $\E^{\mu_\by}$ can be changed to $\E^\mu$ 
and similarly for the reference measure. Once \eqref{381}
is proven for all observables of the form \eqref{cG}, we get
that the locally averaged gap statistics for $\mu$ coincide
with those of the reference Gaussian case, hence they are universal. 
Since averaged gap statistics identifies locally averaged correlation
functions, we obtain Theorem~\ref{bulkbeta}.

\medskip

It remains to prove Theorem \ref{thm:mi}.
The basic idea  is to use the 
 Dirichlet  form
 inequality \eqref{diff}.
Although  \eqref{diff}  was stated for an infinite volume measure, it holds for any measure 
with  repulsive logarithmic interactions in a finite volume 
and with  the parameter $\tau^{-1}$ being 
the lower bound on the Hessian of the Hamiltonian.
In our setting, 
we denote by $\tau_\sigma^{-1}$  the  lower bound for $\nabla^2\cH_\bt$,
 and the  
Dirichlet form inequality becomes 
\be\label{relax}
   \Bigg| \big [\E^{ \mu_{\by}} -\E^{ \sigma_{\bt}}  \big  ]
\frac 1 K \sum_{i \in I} O\Big( N(x_i-x_{i+1}) \Big)\Bigg|
\le C \Big(  \frac { \tau_\sigma  N^{\e}}{ K}
 D \big  (\mu_{\by} | \sigma_{\bt}   \big  ) \,    \Big)^{1/2}
+ Ce^{-cN^{\e}} \sqrt{S (\mu_{\by} | \sigma_{\bt}   \big  )},
\ee
where 
\be
  D(\mu_\by\mid \sigma_\bt) : =
 \frac{1}{2 N} \int \Big|\nabla \sqrt{ \frac{\rd\mu_\by}{\rd\sigma_\bt}}\Big|^2\rd \sigma_\bt.
\ee
Thus our task is to prove  that 
\be\label{d1}
 \tau_\sigma  N^{\e} \frac {   D(\mu_\by\mid \sigma_\bt)  }{K}  \to 0 .
\ee

By definition,
\[  \frac{\tau_\sigma}{K}  D \big(\mu_\by \mid \sigma_\bt)
 \le
\frac{\tau_\sigma N}{ K} \int  \sum_{ L+1\le  j \le L+ K }  Z_j^2 \rd  \mu_{\by},
\]
where $Z_j$ is defined as
\begin{align}
Z_j := &
 \frac{\beta}{2}V'(x_j) -
  \frac \beta N \sum_{    k < L\atop k >L+K }
   \frac 1 {x_j- y_{k}}
 -  \frac{\beta}{2}W'(x_j)  +
 \frac \beta N \sum_{  k <  L\atop k>L+K}   \frac 1 {x_j-  \th_{k}}.
 \label{Zdef}
\end{align}
Using  the equilibrium relation \eqref{equilibrium}
between the potentials $V$, $W$ and the densities $\varrho_V$, $\varrho_W$,
 we have 
\begin{align}
Z_j = &
 \beta\int_\bR  \frac{\varrho_V(y)}{x_j-y}\rd y -
  \frac \beta N \sum_{    k < L\atop k >L+K }
   \frac 1 {x_j- y_{k}}
 - \beta\int_\bR  \frac{\varrho_W(y)}{x_j-y}   \rd y +
 \frac \beta N \sum_{  k <  L\atop k>L+K}   \frac 1 {x_j-  \th_{k}}.
\nonumber
\end{align}
Hence  $Z_j$ is the sum of the error terms,
\begin{align}
A_j  : &= \int_{ y\not \in [y_{L}, y_{L+K+1} ] }    \frac {\varrho_V(y) } {x_j- y }  \rd y
- \frac{1}{N} \sum_{    k < L\atop k >L+K }
   \frac 1 {x_j- y_{k}} ,
\label{Adef}\\
B_j & : = \int_{y_{L}}^{ y_{L+K+1} }    \frac {\varrho_V(y) - \varrho_{W}(y)  } {x_j- y } \rd y,
\label{fiveom}
\end{align}
and there is a term  similar to $A_j$ with $y_j$ replaced by $\theta_j$ and $\varrho_V$
 replaced by $\varrho_W$.

 With our convention,  the total numbers of particles 
in the interval $[y_{L+K+1}, {y_{L}}]$ are equal and thus 
$$
\int_{y_{L}}^{ y_{L+K+1} }  \varrho_V(y)   \rd y = \int_{y_{L}}^{ y_{L+K+1} }  \varrho_{W}(y)  \rd y.
$$
Since the  densities  $\rho_V$ and $\rho_W$ are $C^1$ functions away from the endpoints  $A$ and $B$
and $  y_{L+K+1}- {y_{L}}$ is small, $ |\rho_V-\rho_W|$ is small in the interval $[y_{L+K+1}, {y_{L}}]$ and
thus   $B_j$ is small.
For estimating $A_j$, we can replace the integral 
$$
\int_{ -\infty}^{y_{L}}     \frac {\varrho_V(y) } {x_j- y }  \rd y \quad \mbox{by}\qquad 
\frac 1 N  \sum_{  k< L   }  \frac 1 {x_j- \gamma_{k}}
$$ with negligible errors,
 at least for $j$'s away from the edges, 
$j\in \llbracket L+N^\e, L+K-N^\e\rrbracket$.
 Thus
\be\label{b7}
|A_j |  \le   \frac C N  \Big | \sum_ { k <  L\atop k>L+K} T^k_j
  \Big|,  \qquad T^k_j :=  \frac 1 {x_j- y_{k}} -  \frac 1 {x_j- \gamma_{k}} ,
\ee
and $T^k_j$ can be estimated by the assumption
  $|y_k -\gamma_k |\le \delta$ from $\by\in \cG_\delta$. The same argument works if $j$ is 
close to the edge, but $k$ is away  from the edges, i.e. $k\le L-N^\e$ or $k\ge L+K+N^\e$.
The edge terms,  $ T^k_j$ for $|j-k|\le N^\e$,  are difficult to estimate 
due to the singularity in the denominator  and  the event that  many $y_k$'s
with $k<L$ 
may pile up near $y_L$. 
 To resolve this difficulty, 
 we show that   the averaged local statistics 
of the measure $\mu_\by$  
 are insensitive to the change of the 
boundary conditions for $\by$ near the edges.  This can be achieved  by 
the simple inequality 
\be\label{Diff1}
\Big|\frac 1 K \sum_{i \in I}   \int O\big( N(x_i-x_{i+1} )\big)   [ \rd \mu_{\by'} - \rd  \mu_{\by}]
\Big|
\le C   \int |\rd \mu_{\by'} - \rd  \mu_{\by} |  \le C  \sqrt {   S(\mu_{\by'}|   \mu_{\by} ) }
\ee
for any two boundary conditions
$\by$ and $\by'$.   Although we still have to estimate the entropy
that includes a logarithmic singularity, this can be done much more easily
since entropy is less sensitive to singularities than Dirichlet form.
Therefore, we can replace the boundary condition $y_k$ with $y_k'=\theta_k$ for 
$|j-k|\le N^\e$  and  then the most singular 
edge terms in \eqref{Zdef} cancel out.

We note that  we can 
perform this replacement only 
for a small number of index pairs $(j,k)$,
 since
estimating the gap distribution by 
the  total entropy, as noted in \eqref{entbound} in Section~\ref{sec:DBM},
is not as efficient as the estimate using  the Dirichlet form per particle. 
Thus we can afford to use this argument only for the edge terms, $|j-k|\le N^\e$.
For all other index pairs $(j,k)$ we still have
to estimate $T_j^k$ by exploiting that $\by$  is a good configuration,
i.e. $y_k-\gamma_k$ is small.

 Unfortunately,   even with the optimal 
accuracy $\delta\sim N^{-1 + \e'}$ in \eqref{goodset} as an input, 
the relation \eqref{d1} 
still cannot be satisfied  for any choice of $N^{c \e'}
 \le K \le N^{ 1 - c\e'}$. 
To understand this problem, we remark that while  the edge
terms become a smaller  percentage 
of the total  terms in \eqref{Diff1} as $K$ gets bigger, 
the relaxation time to equilibrium for $\sigma_\bt$,  determined by the convexity  of
$\cH_\bt''$,   increases at the same time. 
  At the end of our calculation, 
there is no good regime for the choice 
of $K$.  Fortunately, this can be resolved by  using the idea of
the local relaxation  measure as in \eqref{defW}, i.e., we add 
a quadratic term $\frac{1}{2 \tau } (x_j -\gamma_j)^2$ to the Hamiltonian of the measure
 $\mu_\by$ and $\frac{1}{2 \tau_\sigma } (x_j -\th_j)^2$
for the measure $\sigma_\bt$.  With these ideas, we can complete the proof of Theorem \ref{thm:mi}.

\section{Single gap universality}\label{sec:sg}

In this section we outline the proofs of Theorem~\ref{thm:sg} and \ref{thm:beta}
following closely \cite{EYsinglegap}.
Both proofs rely on the single gap universality for the locally conditioned
measure $\mu_\by$ introduced already in \eqref{eq:muydeold}.
This will be stated in  Theorem~\ref{thm:local},  whose
proof takes up most of this section. At the end, in Section~\ref{sec:loctoglob}
we complete the proofs of  Theorem~\ref{thm:sg} and \ref{thm:beta}.

\subsection{Statements on local equilibrium measures} 
\label{sec:loc}

\subsubsection{Definition}

We work in the bulk spectrum and we consider the local equilibrium
measure on $\cK:=2K+1$ points  which is the conditional measure
after fixing all other points. To define it precisely,
we fix two small positive numbers, $\al, \delta>0$ and choose two positive integer
 parameters $L, K$ such that
\be\label{K}
L\in \llbracket \alpha N,  (1-\alpha)N\rrbracket, \qquad
N^\delta\le K\le N^{1/4}.
\ee
All results will hold for any sufficiently small $\al, \delta$ 
and for any sufficiently large $N\ge N_0$, where  the threshold $N_0$ depends
on $\al, \delta$  and maybe on other parameters of the model.

Denote $ I = \nc I_{L, K}:= \llbracket L-K, L+K \rrbracket$ the set of  $\cK$ consecutive indices in the bulk.
As in Section~\ref{sec:loceq}, we will distinguish external and internal points
by renaming them as
\be\label{renamexy}
(\lambda_1, \lambda_2, \ldots,
\lambda_N):=(y_{1}, \ldots y_{L-K-1}, x_{L-K},  \ldots, x_{L+K}, y_{L+K+1},
  \ldots y_{N}) \in \Xi^{(N)},
\ee
the only  difference is that here the internal particles are labelled
symmetrically to $L$. This discrepancy is only notational, but we
prefer to follow the notations of the original papers.
In short we will write
$$
\bx=( x_{L-K},  \ldots x_{L+K} )\in \Xi^{(\cK)} , \qquad \mbox{and}\qquad
 \by=
 (y_{1}, \ldots y_{L-K-1}, y_{L+K+1},
  \ldots y_{N}) \in\Xi^{(N-\cK)}.
$$
As in \eqref{eq:muydeold} we again
 define  the
local equilibrium measure
 on $\bx$ with boundary condition  $\by$ by
\begin{equation}\label{eq:muyde}
 \quad
\mu_{\by} (\rd\bx) : = \mu_\by(\bx) \rd \bx, \qquad
\mu_\by(\bx):=  \mu (\by, \bx) \left [ \int \mu (\by, \bx) \rd \bx \right ]^{-1},
\end{equation}
where $\mu=\mu(\by, \bx)$ is the (global) equilibrium measure \eqref{01}.
For a fixed $\by$, this measure can  also be written as a Gibbs measure,
\be\label{muyext}
   \mu_\by = \mu_{\by, \beta,V}= Z_\by^{-1} e^{-N\beta \cH_\by},
\ee
with Hamiltonian
\be\label{Vyext}
   \cH_\by (\bx): = \sum_{i\in I} \frac{1}{2}  V_\by (x_i) -\frac{1}{N}
   \sum_{i,j\in I\atop i<j} \log |x_j-x_i|, \qquad
   V_\by(x) :=  V(x) - \frac{2}{N}\sum_{k\not\in I} \log |x-y_k|.
\ee
Here
 $V_\by(x)$ can be viewed as the external potential of
a $\beta$-log-gas of the points $\{ x_i \; : \; i\in I\} $
in the configuration interval $J=J_\by:=(y_{L-K-1}, y_{L+K+1}).$

\subsubsection{Universality of the local gap statistics for $\mu_\by$}

Our main technical result, Theorem~\ref{thm:local} below,  asserts that
the local gap statistics 
is essentially independent
of $V$ and $\by$ as long as  the boundary conditions $\by$ are regular.
This property is expressed by defining the following set of
  ``good'' boundary conditions 
 with some given positive parameters $\xi, \nc \al$
(thet set $\cG$ in \eqref{goodset} played exactly the same role)
\begin{align}\label{yrig}  
\cR= \cR_{L,K}(\xi,\al):= & \{ \by:  
     |y_k-\gamma_k|\le N^{-1}K^\xi,
 \quad k \in \llbracket\alpha N, (1-\alpha)N\rrbracket \setminus I_{L, K} \} \\ \nonumber
  & \cap 
  \{ \by:   |y_k-\gamma_k|\le N^{-4/15}K^\xi, 
 \quad k \in \llbracket N^{3/5}K^\xi,  N- N^{3/5}K^\xi\rrbracket 
  \} \\
 \nonumber  &  \cap 
  \{ \by: |y_k-\gamma_k|\le 1, \;\; k \in\llbracket 1, N\rrbracket\setminus I_{L,K} \} .
\end{align}
This definition is taylored to 
the rigidity bounds for the $\beta$-ensemble, see
Theorem~\ref{thm:accuracy}. 
 Note that $\cR$ has a 
key parameter, the exponent $\xi$, which will be chosen as an arbitrary
small positive number in the applications. We will not follow its dependence
precisely and we will often neglect it from the notation, i.e. we will talk
about ``good'' boundary conditions $\by\in \cR$.

Good boundary conditions give rise to a regular potential  $V_\by$. More precisely, if $\by\in\cR$,
then
\begin{align}\label{Jlength}
    |J_\by| & =   \frac{\cK}{N \varrho(\bar y) } + O\Big(\frac{K^\xi}{N}\Big), 
\\
\label{Vby1}
   V_\by'(x) & = \varrho(\bar y) \log \frac{d_+(x)}{d_-(x)}
   + O\Big(\frac{K^\xi}{N d(x)}\Big),   \qquad x\in J_\by,
\\
\label{Vbysec}
   V_\by''(x) & \ge \inf V'' +\frac{c}{ d(x)},   \qquad x\in J_\by.
\end{align}
Here 
$$
   \bar y: = \frac{1}{2}( y_{L-K-1}+y_{L+K+1})
$$
denotes the midpoint of   the configuration interval,
$d(x): = \min\{ |x-y_{L-K-1}|, |x- y_{L+K+1}|\}$ is the distance
of $x$ to the boundary of the configuration interval $J=(y_{L-K-1}, y_{L+K+1})$
and $d_-(x)$ and $d_+(x)$ are regularized versions of the distances
of $x$ to the closest and to the farthest endpoints
of  $J$, respectively. The key point is that the leading term of $V_\by'(x)$ 
depends on the boundary conditions only through the density in the center, $\varrho(\bar y)$.
We also introduce
\be\label{aldef}
\alpha_j: = \bar y + \frac{j-L}{\cK+1}|J|, \qquad j\in I_{L,K},
\ee
 to denote the $\cK$ equidistant points
within the interval $J$.

\medskip

\begin{theorem} [Gap universality for local measures]  \label{thm:local} 
Fix $L, \wt L$ and $\cK= 2 K+ 1$ satisfying \eqref{K}
with an exponent $\delta>0$. Consider two
boundary conditions $\by, \wt\by$ such that the configuration intervals coincide,
\be\label{J=J}
   J = (y_{L-K-1}, y_{L+K+1}) = (\wt y_{\wt L-K-1}, \wt y_{\wt L+ K+1}).
\ee
We consider the
measures $\mu = \mu_{ \by, \beta, V}$ and $\wt\mu = \mu_{ \wt\by, \beta, \wt V}$
defined as in \eqref{muyext},
with possibly two different external potentials $V$ and $\wt V$.
 Let $\xi>0$ be a small constant. 
 Assume that  $|J|$ satisfies 
\be\label{Jlen}
  |J|  =   \frac{\cK}{N \varrho(\bar y) } + O\Big(\frac{K^\xi}{N}\Big).
\ee
Suppose that   $\by, \wt\by\in \cR$  
and that 
\be\label{Ex}
  \max_{j\in I_{L,K}}  \Big| \E^{ \mu_\by} x_j -  \alpha_j\Big| +  
 \max_{j\in I_{\wt L,K}} 
\Big| \E^{\wt \mu_{\wt\by}} x_j -  \alpha_j\Big|\le CN^{-1}K^\xi 
\ee
holds.
Let the integer number $p$ satisfy
$|p| \le K-K^{1-\xi^*}$ for some   small $\xi^*>0$.  Then 
there exists $\xi_0 > 0$, depending on
$\delta$,   such that if $\xi,\xi^* \le \xi_0$ 
 then  for   any $n$ fixed and  any bounded smooth observable 
$O:\bR^n\to \bR$ with compact support we have 
\begin{align}\label{univ}
 \Bigg|  \E^{\mu_\by} 
  O\big(  N(x_{L+p}- x_{L+p+1}),  \ldots &  N(x_{L + p}-x_{L +p+n} ) \big) \\ \non
& -   \E^{ \wt\mu_{\tilde \by } } 
  O\big(  N(x_{\wt L + p}- x_{\wt L + p+1}), \ldots  N(x_{\wt L + p}-x_{\wt L +p+n} ) \big)
 \Bigg|\le C K^{-\e} 
\end{align}
for some $\e > 0$ depending on $\delta,\al$ 
and  for some $C$ depending
on $O$. This holds for any $N\ge N_0$ sufficiently large,
where $N_0$ depends on the parameters $\xi,\xi^*, \al$,  and  $C$ in \eqref{Ex}.
\end{theorem}

\subsubsection{Rigidity and level repulsion of $\mu_\by$}

In the following two theorems  we establish
rigidity and level repulsion estimates for the local log-gas $\mu_\by$ 
with  good boundary conditions $\by$. 
 While both rigidity and level repulsion  are basic questions for log gases
and are interesting in themselves,
our main motivation to prove these theorems is to use them in the proof of Theorem \ref{thm:local}. 

We remark that an almost optimal rigidity estimate in the bulk
was given in Theorem~\ref{thm:accuracy} and some level repulsion bound
was given in (4.11) of \cite{BEY}, these results hold with respect
to the global measure $\mu$. For the proof  of Theorem \ref{thm:local}
we need their local versions with respect to $\mu_\by$, at least for most $\by$. 
Naively, this looks as a simple conditioning argument, but there is a subtle point.
{F}rom the estimates w.r.t. $\mu$,
one can  conclude that $\mu_\by$ has a good rigidity bound for 
a set of boundary conditions with high probability  w.r.t. the global measure $\mu$.  
This will be sufficient for the proof of Theorem~\ref{thm:beta}, but not for Theorem~\ref{thm:sg}.
In  the proof  for the gap universality of  Wigner matrices
we will need  a rigidity estimate  for  $\mu_\by$ 
for a set of $\by$'s with  high probability with respect to 
 by the time evolved 
measure $f_t \mu$  which may be
asymptotically singular to $\mu$ for large $N$. 
The following result asserts that a rigidity estimate holds 
for $\mu_\by$ provided that $\by$ itself satisfies a
 rigidity bound and  an extra  condition, \eqref{Exone},  holds. 
This condition will have to be verified
with different methods in the Wigner case.

\begin{theorem}  [Rigidity estimate for local measures]  \label{thm:omrig} 
For   $\by \in \cR $  consider
the local equilibrium measure $\mu_\by$ defined in \eqref{muyext}
  and  assume that
\be\label{Exone}
   \Big| \E^{ \mu_\by} x_k -  \alpha_k\Big| 
\le CN^{-1}K^\xi, \quad  k\in I=I_{L,K},
\ee
is satisfied.
Then there are positive constants $C, c$, depending on $\xi$,
such that  for any $k \in  I $ and $u>0$, 
\be\label{rig}
   \P^{\mu_\by}\Big( N\big| x_k - \al_k\big| \ge u K^{ \xi }\Big)\le C e^{-c u^2 }.
\ee
\end{theorem}

The proof of this result is similar to that  of the concentration estimate
 \eqref{concen1} in Theorem~\ref{thm:accuracy}. To estimate $x_k -\E^{ \mu_\by} x_k$, we again use
 a multiscale argument of local averages for which stronger
convexity bounds are available. The analogue of the
 accuracy estimate controlling $\E^{ \mu_\by} x_k -\gamma_k$
in  Theorem~\ref{thm:accuracy} is replaced by
the assumption \eqref{Exone}. Notice that, unlike for the global measure $\mu$, a direct accuracy control
via the  loop equation  is not available for
$\mu_\by$ since the potential $V_\by$ is not analytic.

\medskip

Now we state the level repulsion  estimates.

\begin{theorem}  [Level repulsion estimate for local measures] 
 \label{lr2} 
For  $\by \in \cR$ we have the  following estimates: 

\noindent 
i)  [Weak form of level repulsion]  For any $s>0$ 
 we have 
\be\label{k521}
\P^{ \mu_\by} [  N( x_{i+1} - x_{i}) \le s   ] \le
  C \left ( N  s \right ) ^{\beta + 1},  \qquad i\in\llbracket L-K-1, L+K\rrbracket .
\ee
\noindent 
ii)   [Strong form of level repulsion] 
Suppose that there exist  positive constants   $C, c$ such that the following rigidity estimate holds
for any $k\in I$:  
\be\label{weakrig}
   \P^{\mu_\by}\Big( N|x_k-\al_k|\ge CK^{\xi^{ 2}}\Big) \le C \exp{(-K^{c})}.
\ee
Then there exists small a constant $\theta$, depending on $C, c$ in \eqref{weakrig},
such that for any $s\ge \exp (- K^{\theta})$.
we have 
\be\label{k52}
\P^{ \mu_\by} [ N( x_{i+1} - x_{i}) \le s   ] \le
  C \left ( K^{\xi } s  \log N \right ) ^{\beta + 1}, \qquad i\in\llbracket L-K-1, L+K\rrbracket .
\ee
\end{theorem}

 The level repulsion bounds  will mostly be used in the following estimate
which trivially follows from  Theorem~\ref{lr2}

\begin{corollary}\label{cor:mom} Let $\by\in \cR$, then for any  $p< \beta + 1$  we have
\be\label{expinv}
   \E^{\mu_\by} \frac{1}{\big[ N|x_i-x_{i+1}|\big]^p} \le C_p K^{C_3\xi}, \qquad   i \in \llbracket L-K-1, L+K\rrbracket.
\ee
\qed
\end{corollary}

\subsubsection{Sketch of the proof of the level repulsion}

The proof of part (ii) of Theorem~\ref{lr2} goes in three steps. For simplicity, we consider
\eqref{k52} only for the first gap, i.e. $i= L-K-1$ and we also assume that $\bar y=0$
by a simple shift. 

\medskip

{\it Step 1.} In this  step we prove 
\be\label{weakerrep} 
  \P^{ \mu_{\by}} ( x_{L-K}-y_{L-K-1}\le s/N)\le CKs \log N,
\ee
which is essentially
\eqref{k52} but with factor $K$ instead of $K^\xi$
and with the exponent $\beta+1$ replaced with  one. The proof of \eqref{weakerrep}
 is dilation argument. For a nonnegative parameter
$\varphi$,  we define
\begin{align}
Z_\varphi :=  & \int\ldots\int_{-a+ a \varphi }^{a- a \varphi}   \rd \bx
 \prod_{i,j\in I\atop i < j} (x_i-x_j)^\beta
e^{- N\frac{\beta}{2} \sum_j V_\by (x_j)  } \nonumber \\
 & = (1-\varphi)^{ K+\beta K(K-1)/2} \int\ldots\int_{-a }^{a}   \rd \bw
\prod_{i < j} (w_i-w_j)^\beta e^{- N \frac{\beta}{2}\sum_j V_\by ((1-\varphi) w_j)},
\label{Zphi}
\end{align}
where we set 
\be\non
a:= -y_{L-K-1}, \qquad w_j:=(1-\varphi)^{-1}x_{L+j}, \qquad \rd \bx =  \prod_{|j|\le K} \rd x_{L+j}\qquad
\rd\bw =   \prod_{|j|\le K} \rd w_j.
\ee
Clearly $Z_{\varphi=0}$ is the normalization constant of the measure $\mu_\by$ and we have
\be\label{ZZ}
    \P^{ \mu_{\by}} ( x_{L-K}-(-a)\ge a\varphi )\ge \frac{Z_\varphi}{Z_0}.
\ee
The multiple integral on the r.h.s of \eqref{Zphi} is almost the same
as $Z_0$, except that the argument of $V_\by$ is rescaled by $1-\varphi$. 
This effect can be estimated from the explicit formula \eqref{Vyext} for $V_\by$.
The external potential $V$ in \eqref{Vyext} is unproblematic since it is smooth.
Due to $\by\in \cR$, the points $y_j$ are regularly spaced on  scales at least
$K^\xi/N$, thus the sum of the interaction terms $\log |x-y_k|$  
for $k$'s away from the edges of $I^c$, i.e. $k\le L-2K$ or $k\ge L+2K$, is a regular function of $x$ and the effect
of dilation  can be well approximated by Taylor expansion. For nearby $k$'s 
right below the lower edge, i.e. $L-2K\le k\le L-K$,
we use the trivial bound $(1-\varphi) x- y_k \ge (1-\varphi) (x - y_k)$. From
these estimates it follows that
\be\label{Z/Z}
\frac{Z_\varphi}{Z_0} \ge 1- CK^2\varphi\log N.
\ee
Since $a\sim K/N$, together with \eqref{ZZ} it implies \eqref{weakerrep}. 

\medskip

{\it Step 2.}
 Now we consider an auxiliary measure which 
are  slightly modified  version of the local equilibrium measures:   
\be\label{mu0}
 \mu^{(0)} :=
 Z^{(0)}   (x_{L-K} - y_{L-K-1})^{-\beta} \mu_\by; 
 \ee
where  $Z^{(0)}$  are chosen for normalization.  
In other words, we drop the term $(x_{L-K} - y_{L-K-1})^\beta$ from  
the measure $ \mu_\by$.
Setting $X:=  x_{L-K} - y_{L-K-1}$
for brevity, we have
\be\label{k4}
\P^{ \mu_\by} [ X \le s/N   ] =  
 \frac { \E^{ \mu^{(0)}}  [  1 ( X \le s/N  )  X^\beta ] }
{ \E^{ \mu^{(0)} }  [  X^\beta ]}.
\ee
The estimate \eqref{weakerrep} also holds for $ \mu^{(0)}$ and thus
$$
\E^{ \mu^{(0)} } [  {\bf 1}
 ( X \le s/N  )  X^\beta ] \le C  (s/N)^\beta K s  \log N 
$$
 and with the choice $s=cK^{-1}  (\log N)^{-1}$ in \eqref{weakerrep} we also have  
$$
\P^{ \mu^{(0)} }\left  (  X \ge    \frac{c }{N K \log N }  \right )  \ge 1/2 
$$
with some positive constant $c$.
This implies that  
$$
\E^{ \mu^{(0)} }  [  X^\beta ] \ge \frac 1 2  \left ( \frac c {N K \log N }\right )^\beta.
$$
Combining with \eqref{k4}, we have thus proved that 
\be\label{lrb}
\P^{ \mu_\by} [  X \le s/N   ] \le 
 C \left ( { Ks \log N}  \right ) ^{\beta + 1},
\ee
 i.e. we obtained \eqref{weakerrep} but with an exponent $\beta+1$ in the r.h.s.

\medskip

{\it Step 3.}  We now improve the constant $K$ to $K^\xi$ in the r.h.s of \eqref{weakerrep}.
The factor $K$ originated from the number of particles in $\mu_\by$. We can
further condition the measure $\mu_\by$ on the points
$$
    z_j:= x_j\; \qquad  j\ge L-K+K^\xi,
$$
and we let $\mu_{\by,\bz}$ denote the conditional measure
on the remaining $x$ variables $\{ x_j\; : \;  L-K \le j\le L-K+K^\xi\}$. From the rigidity estimate  \eqref{weakrig} 
 we have $(\by, \bz) \in \cR$ with a
very high probability w.r.t. $\mu_\by$.  
 We will now apply  \eqref{lrb}   to the measure $\mu_{\by,\bz}$ to
 obtain
\be\label{k51}
\P^{ \mu_{\by, \bz}} [ X \le s/N   ] \le
  C \left ( K^{\xi } s  \log N \right ) ^{\beta + 1}.
\ee
This holds for all $z$ 
with a high $\mu_\by$-probability.
 The subexponential lower bound on $s$, assumed in part ii) of Theorem~\ref{lr2}, 
 allows us to include the probability of
the complement of $\cR$ in the estimate, we thus have proved \eqref{k52}.

 The proofs of the weaker bound  \eqref{k521} 
 for any $s>0$
use similar arguments that have led to \eqref{weakerrep}, but 
without assuming $\by\in \cR$ which yields that one factor of $K$ has to be replaced with $N$
in \eqref{Z/Z}. The assumption that the boundary conditions are good needs to be dropped since in Step 3
of the above argument, \eqref{weakerrep} is also used after additional conditioning on $\bz$, distributed
according to $\mu_\by$,  and without  
\eqref{weakrig} there is no rigidity result available for $\mu_\by$.

\subsection{Proof of Theorem~\ref{thm:local}}

In this section, we start to compare gap distributions  of two local log-gases
 on the same configuration interval but with different external potential and 
boundary conditions.  For simplicity,
we consider only an observable of a single gap; a few consecutive gaps
can be handled similarly.
{F}rom now on, we use microscopic  coordinates, i.e. we replace $x_j$ with $x_j/N$,
  and we also relabel the indices so that the coordinates 
of  $x_j$ are  $j \in I=\{-K, \ldots, 0, 1, \ldots K\}$.
 This will have the advantage that $K$ remains
the only large parameter; $N$ disappears.

 The local equilibrium measures and their Hamiltonians 
will be denoted by the same symbols, $\mu_\by$ and $\cH_\by$, as before,
but with a slight abuse of notations we redefine them now  to the microscopic scaling, i.e. 
\be\label{Vz} 
   \cH_\by (\bx): = \sum_{i\in I} \frac{1}{2} V_{ \by} ( x_i) -
   \sum_{i,j\in I\atop i<j} \log |x_j-x_i|, \qquad
   V_\by (x) :=  N V(x/N) - 2\sum_{j\not\in I} \log |x-y_j|,
\ee
The other Hamiltonian $\wt H_{\wt \by}$ is defined in a similar way with $V$ in \eqref{Vz} replaced
 with another external potential  $\wt V$. 
We also rewrite \eqref{Ex}  in the microscopic coordinate as 
\be\label{Exm}
   | \E^{ \mu_\by} x_j -  \alpha_j| +  | \E^{\wt \mu_{\wt \by}} x_j -  \alpha_j |\le C  K^{\xi}, \quad 
\ee
where 
$\alpha_j: = \frac{j}{\cK +1} |J| $ 
is the rescaled version of the definition given in \eqref{aldef},
but we keep the same notation.  The concept of ``good'' set $\cR$ is also rescaled
accordingly.

Suppose that $\by,\wt\by\in \cR$ and define the interpolating  measures 
\be\label{omd}
\om_{\by, {\wt \by}}^r =    Z_r e^{-\beta r (\wt V_{\wt \by} (\bx) - V_\by (\bx) )}  \mu_\by,
 \qquad  r\in[0,1], 
\ee
so that $ \om_{\by, {\wt \by}}^1 =\wt\mu_{\wt \by}$ and $\om_{\by, {\wt \by}}^0= \mu_\by$
 ($Z_r$ is a normalization constant). 
This is again a local log-gas with Hamiltonian
\be\label{Hyy}
  \cH_{\by, \wt\by}^r (\bx)= \frac{1}{2}\sum_{i\in I} V_{\by, \wt\by}^{r}(x_i) 
  -\frac{1}{N}\sum_{i<j} \log |x_i-x_j|, \qquad V_{\by, \wt\by}^{r}(x) : =   (1-r)  V_\by(x)+ r \wt V_{\wt\by}(x).
\ee

For any fixed $r$,  the measure $\om_{\by, {\wt \by}}^r$ inherits all relevant properties
of $\mu_\by$. In particular the rigidity bound in the form
\be\label{rigi}
 \P^{\om}\big( \big| x_i- \alpha_i\big| \ge C K^{ C\xi }\big)\le C e^{- K^{\theta}}, 
 \quad i \in I,
\ee
the level repulsion bounds 
\eqref{k521}--\eqref{k52} and their consequence in \eqref{expinv} 
hold w.r.t. the measure $\om= \om_{\by, {\wt \by}}^r$ as well (in the new microscopic coordinates
there are no $N$ factors in the left hand sides of these inequalities).
The proofs  are basically parallel with the arguments for $\mu_\by$; the only
nontrivial step is to show that  \eqref{Exm} implies the analogous bound 
$$
  \big|\E^\om x_k -\al_k\big| \le CK^\xi
$$
w.r.t. $\om= \om_{\by, {\wt \by}}^r$ as well. Although $\om$ appears to be some easy combination of
$\mu_\by$ and $\wt\mu_{\wt\by}$, this conclusion is nontrivial. It requires
comparing $\om$ and $\mu_\by$ via the entropy inequality, which
 involves controlling the exponential moment of $|x_k-\al_k|$  w.r.t. $\mu_\by$.
At this point the Gaussian tail proven in \eqref{rig} is necessary.

The right hand side of \eqref{univ} with $n=1$,
in the rescaled coordinates and with $L=\wt L=0$, is estimated by
\be\label{rinteg}
  \Big| [\E^{\mu_\by} - \E^{\wt \mu_{\wt\by}}] O(x_{p}-x_{p+1} ) \Big|
  \le \int_0^1 \rd r  \frac{\rd}{\rd r}\E^{\om_{\by, {\wt \by}}^r } O(x_{p}-x_{p+1} ).
\ee
For any bounded smooth function $O$ with compact support 
\be\label{rcorr}
\frac{\rd}{\rd r} \E^{\om_{\by, {\wt \by}}^r }   O  (x_{p}-x_{p+1} ) =  
    \beta  \langle h_0;  O  (x_{p}-x_{p+1} )  \rangle_{\om_{\by, {\wt \by}}^r },
\ee
where
   \be\label{h0def}
h_0 = h_0(\bx)= \sum_{i \in I}   ( V_\by(x_i) - \wt V_{\wt \by}(x_i) )
\ee
and $\langle f ; g\rangle_\om : = \E^\om fg - (\E^\om f)(\E^\om g)$ denotes the covariance. 
Thus Theorem~\ref{thm:local} follows immediately from the 
following   estimate on the gap covariance function. \qed

\begin{theorem}\label{cor}   
 Consider two smooth potentials $V, \wt V$
and two good boundary conditions, $\by, \wt\by\in \cR$,   such that the configuration intervals coincide,
$J_\by = J_{\wt\by}$.  For any $r\in [0,1]$ let $\om = \om_{\by, {\wt \by}}^r$ be the interpolating measure
defined in \eqref{omd}.
 Assume that
 \eqref{Exm} holds
  for both boundary conditions $\by,\wt\by$. 
 Fix $\xi^*>0$. Then there exist  $\e> 0$ and $C>0$, 
depending on $\xi^*$, such that for
any sufficiently small $\xi$, 
for $ |p| \le K^{1-\xi^*}$  we have 
\be\label{eq:cor}
 | \langle h_0(\bx);  O(x_{p}-x_{p+1} ) 
 \rangle_\om      | \le   K^{C  \xi} K^{-\e}
\ee
 for any smooth function $O:\R \to \R$ with compact support provided that  
$K$  is large enough.
\end{theorem}

Theorem~\ref{cor} is our key technical result. The main difficulty behind it
is due to the fact that 
the covariance function  of two points, $\langle x_i;  x_j  \rangle_{\om}$,  decays only logarithmically.
 In fact, for the GUE, Gustavsson 
proved that  (Theorem 1.3 in \cite{Gus}) 
\be\label{GUS}
\langle x_i; x_j \rangle_{GUE}  \sim 
     \log \frac {N}{  [ |i-j|+1]},
\ee
and a similar formula is expected for $\om$.    
Although $h_0(\bx)$ depends strongly only on points near the boundary 
and $x_p$ is away from the boundary, it is
still very difficult to prove 
Theorem \ref{cor} based on this slow logarithmic decay. However,
 the covariance function of the type 
\be
\langle g_1(x_i); g_2( x_j-x_{j+1}) \rangle_\om
\ee
decays much faster  in $|i-j|$.  Since  the second factor $g_2(x_j-x_{j+1})$ 
 depends only on the difference of two neighboring points, it is expected that the decay
is the (discrete) derivative in $j$ of the covariance \eqref{GUS}, i.e. it
is $|i-j|^{-1}$.  The actual result \eqref{eq:cor} is much weaker, but it still provides
a  power-law decay in $K$ instead of a logarithmic decay.
 Covariances of the form 
$\langle g_1(x_i-x_{i+1}); g_2( x_j-x_{j+1}) \rangle_\om$ are expected to decay even faster
but we have not pursued this direction further.

We point out that
the  fact that observables of differences 
of particles behave much nicer was a basic observation in  DBM analysis 
(Theorem~\ref{thm3}),
 see the explanation  around \eqref{entbound}.

\subsection{Decay of correlation functions: Proof of Theorem~\ref{cor}}

We will express the difference of gap distributions between two measures 
in terms of  random walks in  time dependent 
random environments. The  decay of correlation functions will be 
translated into a partial regularity property
 of the corresponding parabolic equation. This partial regularity
is a discrete version of the De Giorgi-Nash-Moser theory but
with a long range elliptic part.

\subsubsection{Random Walk Representation}\label{sec:rw}

 In this section we  derive a random walk  representation for  the gap correlation function
on the left hand side of \eqref{eq:cor}. We will apply it for the interpolating measure
$\om= \om^r_{\by, \wt\by}$ \eqref{omd} and for the
 function $h_0$ given in \eqref{h0def},
but the representation formula (Proposition~\ref{prop:repp} below) is valid for any $\om$ and $h_0$.

Let $  \cL^\om$   
be the reversible generator given by the Dirichlet form 
\be\label{LK}
D^\om(f)=  -\int f \cL^\om f \rd \om   =  \sum_{|j| \le K}  \int   (\partial_j f )^2 
\rd \om.
\ee
This process can also be characterized by the following SDE  
\be
  \rd x_i = \rd B_i + \beta \Big[ - \frac{1}{2}  (V_{\by,\wt\by}^{r})'   ( x_i ) +
 \frac{1}{2}\sum_{j\ne i} \frac{1}{(x_i-x_j)}\Big] \rd t,
\label{SDE}
\ee
where $\{ B_i\; : \; |i|\le K\}$ is a family of independent standard real Brownian motions.
 Let $\E_\bx$ denote the expectation for this process with initial point $\bx(0)=\bx$. 
 The expectation with respect to the process starting from equilibrium is
$\E^\om[\cdot ] = \int \E_\bx[\cdot]  \om (\rd\bx)$. 
 With a slight abuse of notations,  when we talk about the process, 
we will  use $\P^\om$ and $\E^\om$ also to denote the probability 
and expectation w.r.t. this dynamics with initial data distributed w.r.t. $\om$, i.e., in equilibrium.

Suppose $h(t)=h(t, \bx)$ is the solution of the equation $\partial_t h = \cL^\om h$
with an initial condition  $h_0$.
Introduce the notation
\be\label{vjdef}
\bv(t, \bx) = \nabla_\bx h(t, \bx), \quad \mbox{i.e.}
 \quad v_j (t,  \bx) : =  \partial_{x_j} h (t, \bx ).
\ee
By integrating the time derivative of $\langle h(t,\bx);  O(x_{p}-x_{p+1}) \rangle_{\om }$
and using the equation $h(t, \bx)$ satisfies, we
have 
\be\label{83}
\langle h_0(\bx);  O(x_{p}-x_{p+1}) \rangle_{\om }   =    \int_0^\infty   \rd \sigma \int  
    O'( x_p-x_{p+1})
[v_p (\sigma, \bx)  -  v_{p+1}  (\sigma, \bx)  ]  \rd \omega  (\bx).
\ee
For any fixed $\sigma$,
the inner integral on the right hand side can be expressed by a random walk representation.
Fix a path $\{ \bx(s)\; : \; s\in [0,\sigma]\}$.  
Define the following operators on $\bR^\cK$ 
\be\label{61}
\cA(s) := \cB(s) + \cW(s)
\ee
\be\label{Bdef}
  [\cB(s)\bv]_j =  - \sum_k  B_{jk}(\bx (\si-s)) (v_k- v_j), \quad 
  B _{jk}(\bx) = \frac 1 { (x_j-x_k)^2} \ge 0
\ee
$$
[\cW(s)\bv]_j = \cW_j(s)v_j, \qquad  \cW_j(s):= [V_{\by,\wt\by}^r]'' (x_j(\sigma- s)).
$$
Clearly $\cB(s)$ is diffusion operator with random rates and $\cW(s)$ is a potential representing
a random environment.  These operators depend on the whole path $\bx(s)$, but we omit this fact from the notation.

With these notations we have the following representation:
\begin{proposition}\label{prop:repp} For any smooth function $h_0: J^\cK \to \R$, 
for any  $p \in I$, $-K \le p \le  K-1 $,  we have   
\begin{align}\label{reppgen}
\langle h_0; & O(x_{p}-x_{p+1}) \rangle_{\om }  = \int_0^{\infty}   \rd \sigma \int  
    O'( x_p-x_{p+1})
\E_\bx  [w_p (\sigma, \bx(\cdot); \sigma )  -  w_{p+1} (\sigma,  \bx(\cdot); \sigma) ] \omega (\rd \bx). 
\end{align}
Here, for any $\sigma>0$ and for any fixed path  $\{ \bx(s)\; : \; s\in [0,\sigma]\}$
we let
$\bw$ denote the solution of the evolution equation 
\be\label{veq2}
\partial_s\bw(s; \bx(\cdot),  \sigma)  = - { \cA}(s) \bw(s; \bx(\cdot),  \sigma)  ,
\ee
with initial data 
$\bw (0; \bx(\cdot), \sigma) := \nabla h_0 (\bx(\si))$.  
\end{proposition} 
 This representation in a slightly different setting  already appeared in  Proposition 2.2 of
 \cite{DGI} (see also Proposition 3.1 in \cite{GOS}), which was a probabilistic formulation  of the idea 
 of Helffer and Sj\"ostrand \cite{HS} and Naddaf and Spencer \cite{NS}. The proof relies on
taking the gradient of the equation $\partial_t h = \cL^\om h$.
A direct computation of the commutator $[\nabla, \cL^\om]$ 
yields that 
\be\label{ge}
\partial_t \bv (t,  \bx)  = \cL^\om \bv (t,  \bx)  -  \wt \cA ( \bx ) \bv(t, \bx) , 
\ee
with initial condition  $\bv_0(\bx)=\bv(0, \bx)= \nabla h_0(\bx)$.
Here $\wt\cA(\bx) = \cB(\bx) + \cW(\bx)$, where
$\cB(\bx)$ is the operator given by the matrix $B_{jk}$ in \eqref{Bdef} and
 $\cW(\bx)$ is
the diagonal multiplication operator by 
$ [V_{\by,\wt\by}^r]''(x_j)$. Since $\cL^\om$ generates the process $\bx(t)$, 
we can represent the solution to \eqref{ge} by the Feynman-Kac formula
which can be written in the form \eqref{reppgen}.

When applying this proposition to our case, we will choose the initial condition
 $h_0$ be given by \eqref{h0def}. The initial condition for the random walk \eqref{veq2} is given
by $\nabla h_0$. Notice that the leading
term in $\partial_j h_0(\bx) =  V_\by'(x_j) -  \big[ \wt V_{\wt\by}\big]'(x_j)$ cancel; this is because
 the leading term
in \eqref{Vby1} depends only on the density $\varrho(\bar y)$ which
is matched for $\by$ and $\wt\by$ by $J_\by=J_{\wt\by}$, see \eqref{Jlength}.  We thus have 
\be\label{gradh0}
  | \partial_j h_0 (x)| \le \frac{CK^\xi}{d(x_j)},
\ee
i.e. initially $\bw$ is small away from the boundary and for the small
$\sigma$ regime the inner integral in the r.h.s. of \eqref{reppgen} is small.
After very long time $\bw$ becomes constant, but then the right hand side of \eqref{reppgen}
is zero.  The analysis of \eqref{reppgen} 
requires to monitor what happens to $\bw$ for coordinates $p$ away from the boundary at intermediate times.

In the following sections we make a few preparations that exclude
irrelevant regimes. First, 
it is easy to see that the regular spacing of $\by, \wt\by \in \cR$
implies that $W_j(s)\ge cK^{-1}$, which means that the $L^1$-norm of the solution
to \eqref{veq2} decays at a rate of order $K$. Thus the integral in \eqref{reppgen}
can be truncated at $\sigma\le CK\log K$.

\subsubsection{Preparation for the De Giorgi-Nash-Moser bound: Restriction to the good  paths}

 The representation \eqref{reppgen} expresses the
covariance function in terms of the discrete spatial derivative
of the solution to \eqref{veq2}.
To estimate  $w_p (\sigma, \bx(\cdot); \sigma )  -  w_{p+1}  (\sigma,  \bx(\cdot); \sigma)$
in \eqref{reppgen}, 
we will now study the H\"older continuity of the solution $\bw (s, \bx(\cdot); \sigma )$ to \eqref{veq2}
at time $s=\sigma$ and at the spatial point $p$. 
We will do it for each fixed
path $\bx(\cdot)$, with the exception of a set of ``bad'' paths that will have a  small probability. 

Notice that if all points $x_i$  were approximately regularly spaced in the interval $J$,
then the operator $\cB$ had a kernel $\cB_{ij}\sim (i-j)^{-2}$, i.e. it 
were essentially a discrete version of the operator $|p|=\sqrt{-\Delta}$ (in one dimension).
H\"older continuity will thus be the consequence of the De Giorgi-Nash-Moser bound
for the parabolic equation \eqref{veq2}. 
However, we need to control
the coefficients in this equation, which  depend on the random walk $\bx(\cdot )$.

For the De Giorgi-Nash-Moser theory we need both upper and lower bounds on
the time dependent kernel $\cB_{ij}(s)$. The rigidity bound  \eqref{rigi} guarantees 
a lower bound on $\cB_{ij}$, up to a factor $K^{-C\xi}$. Since the subexponential probabilistic estimate
in \eqref{rigi} is very strong, one can easily guarantee a very similar estimate uniformly in time, i.e.
\be\label{unifrig}
\P^\om \Big\{\bx (s)\; : \;  \sup_{0 \le s \le CK\log K} \;\sup_{ |j| \le K } |x_j(s)-\alpha_j| \le K^{C\xi}  \Big\} 
\ge 1- e^{-K^\theta}
\ee
(maybe after reducing $\theta$ from \eqref{rigi}). This follows from the fact that $\om$ is
invariant under the dynamics and $\bx(t)$ has some stochastic continuity.

The level repulsion estimate
implies  certain upper bounds  on $\cB_{ij}$, but these estimates not particularly
strong. Even in the $\beta>1$ case, the bound \eqref{expinv} implies only that
$$
   \E^\om \cB_{i, i+1}(s) = \E^\om \frac{1}{(x_{i+1}-x_i)^2} \le K^{C\xi}
$$
is finite. In the $\beta=1$ borderline case even the expectation of $\cB_{i,i+1}$ is
infinite. Such a weak control does not allow us to guarantee
an effective simultaneous bound on $\cB_{i,i+1}$ for all $i$ and for all time.
Instead of  supremum bounds, we  control these coefficients only
in an average sense and we can show that for any fixed index $Z\in I$,
time $s$ and parameter $M$, we have 
\be\label{PQ}
  \P^\om\Big\{\bx (s)\; : \;
 \frac{1}{1+s} \int_0^s \rd a \;\frac{1}{M} \sum_{|i-Z|\le M} \sum_j 
B_{ij}(\bx(\si-a)) \le K^\rho\Big\} \ge 1- K^{C\xi-\rho}.
\ee
Here $\rho$ will be chosen as large constant times $\xi$.
The summation over $j$ is harmless since for $|i-j|\ge K^\xi$ the rigidity estimate can be used to bound $\cB_{ij}$.
By a dyadic choice of the parameters $s, M$, it is easy to upgrade \eqref{PQ} to hold
for any $M\le K$ and $s\le CK\log K$. But it is essential that a reference point $Z$ be fixed,
one cannot guarantee that none of the gaps closes.

The expectation over the paths,
$\int \E_\bx[\, \cdot\,] \om(\rd \bx)$, in \eqref{reppgen} will be restricted to the sets 
given in \eqref{unifrig} and \eqref{PQ}.
Due to the strong subexponential bound, the restriction to the set in \eqref{unifrig} is unproblematic. 
However, the estimate \eqref{PQ} is quite weak; the probability of  the ''bad''  paths is bounded
only by a small negative power of $K$. This is not sufficient to 
compensate the time integration   in  \eqref{reppgen} even after the upper cutoff
$\si\le CK\log K$. We will need to use that the heat kernel of the equation \eqref{veq2} has an
  $L^1\to L^\infty$ decay of order $1/s$ after time $s$. Thus the solution $w_p(\sigma, \bx (\cdot); \sigma)$ decays
as $1/\sigma$ which renders the $\rd\sigma$ integration in \eqref{reppgen} harmless. 

For completeness, we state the $L^p\to L^q$ heat kernel decay estimate in a general form. Notice that we only assume
 a lower bound in $\cB_{ij}$ to guarantee sufficient ellipticity;
there is no upper bound required for these bounds.
\begin{proposition}
Consider the evolution equation 
\be\label{ve}
\partial_s  \bu (s) =  - \cA(s)  \bu (s), \qquad \bu(s)\in\R^\cK
\ee
and fix $\si>0$. 
Suppose that for
 some
constant $b$ we have 
\be\label{B}
  \cB_{jk}(s) \ge   \frac b   { (j-k)^2}, \quad 0 \le s \le \si, \quad j\ne k,
\ee  
and 
\be\label{W1}
\cW_j (s)  \ge  \frac b { d_j } ,   \qquad d_j := \big| |j|-K\big|+1,
\quad 0 \le s \le \si .
\ee
Then  for any $1\le p\le q\le \infty$  we have the decay estimate 
\be\label{decay}
\| \bu(s) \|_q  \le  ( sb)^{-(  \frac{1}{p}  - \frac 1 q)} 
    \| \bu(0) \|_{ p }, \qquad 0<s\le\si.
\ee
\end{proposition} 
The proof  relies on the usual
Nash argument and uses the following critical Gagliardo-Nirenberg-type inequality 
for the discrete version of the operator $\sqrt{-\Delta}$:
\begin{proposition} There exists a positive constant $C$ such that 
\be\label{s}
\| f \|_{L^4(\Z)}^4  \le C \| f \|_{L^2(\Z)}^2 
\sum_{i \not = j  \in \Z}   \frac { |f_i - f_j|^2} { |i-j|^2}  
\ee 
holds for any function $f: \Z \to \R$.
\end{proposition}
The continuous version of this inequality,
$ \| \phi \|_4^4 \le C \| \phi \|_2^2  \langle\phi,  |p|  \phi\rangle$,
 was first proven in \cite{OO}.

\subsubsection{Preparation for the De Giorgi-Nash-Moser bound: Finite speed of propagation}

The H\"older continuity of the parabolic equation \eqref{veq2} emerges
only after a certain time, thus for the small $\sigma$ regime in
the integral \eqref{reppgen} we need a different argument.
Since we are interested in the H\"older continuity around
the middle of the interval $I$  (note that $|p|\le K^{1-\xi^*} $ in
Theorem~\ref{cor}), and
the initial condition $\nabla h_0$ is small in this region (see \eqref{gradh0}),
a finite speed of propagation estimate  guarantees
that $w_p(\si; \bx(0), \si)$ is small if $\si$ is not too large.

Since  \eqref{veq2} is linear, for the finite speed of propagation
it is sufficient to consider the fundamental solution.  For
 $a$ fixed, let $\bu^a(s)$ denote the solution  
\be\label{ve2}
\partial_s  \bu^a (s) = - \cA(s)  \bu^a (s) , \quad  u^a_j (0) = \delta_{aj}.
\ee
with a delta function as initial data.
We will assume that  the coefficients of $\cA$ satisfy, for some 
fixed   $|Z|\le  K/2$  
and  $\rho> 0$,  the bound 
\be\label{Kass}
   \sup_{0 \le s \le \si}\sup_{0 \le M \le K} \frac{1}{ 1+ s} \int_0^s \frac{1}{M}
 \sum_{i\in I\, : \, |i-Z| \le M}\sum_{j\in I\, : \, |j-Z| \le M} 
 \cB_{ij}(\si-a)
 \rd a \le CK^{\rho}.
\ee
 Notice that \eqref{Kass} is satisfied on the set of good path given by \eqref{unifrig} and \eqref{PQ}.
The following lemma provides a finite speed   of propagation  estimate for the equation \eqref{ve2}
under the condition of \eqref{Kass}. 
This estimate is not optimal, but it is sufficient for our purpose.

\begin{lemma}\label{lem-finite}  [Finite Speed of Propagation Estimate] 
 Fix $a\in I$ 
 and $\sigma\le CK\log K$. 
 We assume
 that the coefficients of $\cA$  satisfy
\eqref{B} and \eqref{W1} with $b= K^{-\xi}$.
Assume that \eqref{Kass} is satisfied for some fixed  $Z$, $|Z|\le K/2 $.
Then for the fundamental solution \eqref{ve2} we have the  estimate for any $s \le \si$ and $p\in I$
\be\label{finite}
|u_p^a(s)| \le  \frac {C K^{  \rho  +  2 \xi + 1/2 } \sqrt  {s  +1 }  } {   |p- a |}.
\ee
\end{lemma}

For the proof, we split the operator $\cA=\cS+\cR$  into a short range and a long range part, where
the short range part
$\cS(s)$ is defined by
\be
(\cS(s) \bv)_j :=  -  \sum_{k\; : \; |j-k| \le \ell } \cB_{jk}(s) (v_k-v_j)  + \cW_j(s) v_j 
\ee
with some cutoff parameter $\ell$. The norm of the long range part in any $L^p$
 is bounded by $\ell^{-1}$ and it is treated as a perturbation via Duhamel formula.
For the short range part, we  control the exponentially weighted norm
of the solution of $\pt_s \br(s) = - \cS(s)\br(s)$, i.e. we derive a
Gronwall bound for
$$
  f(s)  =    \sum_{j\in I}  e^{ | j - a |/\theta}  r_j^2 (s).
$$
The result is
$$
f(s) \le  \exp \left [ C\theta^{-2} \ell^2  \int_0^s  \sum_{k,  j:  |j-k| \le \ell } \cB_{k j}  (s')  \rd s'   \right ] f(0).
$$
The exponent is estimated by \eqref{Kass} with $M=K$. The optimization of
the lengthscale $\theta$ together with the cutoff parameter $\ell$  yields
 \eqref{finite}. \qed

\medskip

Inserting the estimate \eqref{finite} into \eqref{reppgen} and using 
the estimate \eqref{gradh0} on the initial data $\nabla h_0$, we obtain that the contribution
of the short time regime, $\sigma\le K^{1/4}$, is negligible if $p$
is away from edge, $|p|\le K^{1-\xi^*}$ for some $\xi^*>0$. This allows us
to disregard the $\sigma\le K^{1/4}$ regime in \eqref{reppgen}
and focus on $\sigma\in [K^{1/4}, CK\log K]$.

\subsubsection{A discrete De Giorgi-Nash-Moser bound}

We will now treat the main part of the integral
\eqref{reppgen}  by parabolic regularity. The preparations
in the previous sections ensure that is is sufficient to consider the
integration regime $\sigma\in [K^{1/4}, CK\log K]$ and we can assume
that the path $\bx(\cdot)$ is good in the sense of the estimates
\eqref{unifrig} and \eqref{PQ}. In particular, the rigidity estimate implies not only lower bounds but also upper
   bounds for distant indices; more precisely we have
\be\label{far1}
\cB_{ij} (s) \le \frac{C}{(i-j)^2}
\ee
for any $|i-j|\ge  CK^{\xi}$ and $ 0\le s\le CK\log K$; and similarly
 \be\label{g4}
\cW_i(s)   \le  \frac{K^{\xi}}{d_i}, \; \; \text { if } \quad d_i \ge K^{C \xi}.
\ee   
The following regularity theorem
combined with  \eqref{gradh0} completes the estimate of \eqref{reppgen}
and completes the proof of Theorem~\ref{cor}. \qed

\begin{theorem}  [Parabolic partial regularity with singular coefficients] 
 \label{holderg} 
 Let  $\bu$ be a solution to \eqref{ve2}, where $\bu= \bu^a$ for any choice of $a$. 
Suppose that   the coefficients of $\cA$
satisfy the lower bounds
\eqref{B} and \eqref{W1} with $b= K^{-\xi}$, the upper bounds \eqref{far1}, \eqref{g4}
for distant indices and the upper bound in
\eqref{Kass} in average sense for all indices.
Let $\si \in [K^{c_1}, C_1K\log K]$ be fixed, where $c_1>0$ is an arbitrary positive constant. 
Then for any  $0<q'<1$  there exists $q>0$ 
  so that for any 
 $|Z| \le  K/2$ 
\be\label{HC}
\sup_{ \max(|j-Z|,  |j'-Z|) \le \si^{ 1- q'} } 
   | u_j (\si)  -  u_{j'}  (\si)  |  \le C \sigma^{-1-q} ,
\ee
where $\bu= \bu^a$ for any choice of $a$. 
\end{theorem}

Notice that this result is deterministic, all probability 
estimates are comprised in verifying  the conditions.   We also remark
 that if we define the rescaled function   $v(j/K, t):= t u_j (t)$, 
then \eqref{HC} can be interpreted as a type of H\"older regularity 
of $v$  on scale $\sigma^{1-q'}K^{-1}\ll 1$ at the point $Z/K$:
\be
 |v(x,  \sigma) - v(y,  \sigma  ) | \le  \sigma^{-q} \le |x-y|^{c_1q}
\label{ss}
\ee
for $ 1/K \le |x-y| \le \si^{ 1- q'}/K$ and $x, y$  near $Z/K$. 
The  H\"older exponent is thus  at least $c_1 q$.

Although the statement of Theorem~\ref{holderg} seems to be complicated, 
the underlying mechanism is that there is  a positive 
 exponent  $q$ in \eqref{HC}, 
 which to a great degree is an universal constant. 
This exponent provides an extra smallness factor in addition to
the natural size of $u_j(\sigma)$, which is $\sigma^{-1}$ from the
$L^1\to L^\infty$ decay. As \eqref{ss} indicates, this gain
comes from a H\"older regularity on the relevant scale.  

\medskip

Our equation
\eqref{ve2} is of the type considered in \cite{C}, but it
is discrete and in a finite interval. The key difference, however, is
that the coefficient $\cB_{ij} = (x_i- x_j)^{-2}$  in the elliptic part of \eqref{ve2} 
can be singular if gaps close, even temporarily, while \cite{C} assumed the uniform 
bound $\cB_{ij}\le C/|i-j|^2$.
The only control we have for the singular behavior of $\cB_{ij}$ is the
estimate \eqref{Kass} which is very weak. This estimate essentially says that 
the space-time maximum function of $\cB_{i, i+1}(t)$   at a fixed space-time point $(Z,0)$ 
is bounded by $K^\rho$. Our main task is to show that this condition
is sufficient for proving H\"older continuity at the same point. 
 Our strategy follows  the approach of  Caffarelli-Chan-Vasseur \cite{C}.
The main new feature of our argument is  the derivation 
of a local energy dissipation  estimate for parabolic equation with singular coefficients satisfying  \eqref{B}, \eqref{W1}
as lower bounds 
and only  \eqref{Kass} as an upper bound. The analogous result in \cite{C},
 called the {\it first De Giorgi lemma}, is proved under uniform bounds on the coefficients.
For our proof, roughly, we have to run the argument of the first De Giorgi lemma
 twice; first we get a bound only in $L^2(\Z)$ then
using this information we upgrade it to an $L^\infty(\Z)$ bound.
This concludes the sketch of the proof of  Theorem~\ref{holderg}. \qed

\subsection{From local measures to Wigner matrices and $\beta$-ensembles}\label{sec:loctoglob}

Given Theorem~\ref{thm:local}, the proofs of Theorem~\ref{thm:sg} and \ref{thm:beta}
follow relatively standard ideas from previous results,
some of them were reviewed in Sections~\ref{sec:uniwig} and \ref{sec:unibeta}.
  The key inputs are to
verify the condition \eqref{Ex} and to ensure 
that the configuration intervals coincide \eqref{J=J}.

For the $\beta$-ensemble,  \eqref{Ex} simply follows from conditioning
the global rigidity estimate in Theorem~\ref{thm:accuracy}. 
For matching the configuration intervals, first we match the local density
 by  scaling and translation that guarantees that $|J_\by|\sim |J_{\wt\by}|$, see \eqref{Jlen}.
 Then, with a second scaling, we fine tune the slight discrepancy
between the lengths of $J_\by$ and $J_{\wt\by}$. This finishes the proof of
Theorem~\ref{thm:beta}. \qed

\medskip

In the Wigner case, we always work on the same configuration interval,
so matching of $J$ is automatic.
The proof of \eqref{Ex}, however,  requires a bit more effort than
for the $\beta$-ensemble, but it will relatively easily follow
from other information we already collected along the
 three step strategy described in Section~\ref{sec:genrem}.
As we explained in the proof of Theorem~\ref{bulkWigner},
the averaging over the energy was really needed only in the second step,
where the closeness of the local statistics of $f_t\mu$ and $\mu$ was shown for small $t$,
where $f_t$ is the evolution of the DBM.
As a byproduct of this step, we obtain bounds on the global entropy and
 Dirichlet form, see \eqref{1.3}. In particular, the local Dirichlet
form w.r.t $\mu_\by$   can be estimated by the global one, which 
then can be used to compare expectations w.r.t. the conditional
measures $f_{t,\by}\mu_\by$ and $\mu_\by$;
\be
   \E^{f_{t,\by}\mu_\by}O(\bx) - \E^{\mu_\by}O(\bx).
\label{OOO}
\ee
 We are especially
interested in controlling the difference
\be\label{EEE}
  \big|\E^{f_{t,\by}\mu_\by}x_j - \E^{\mu_\by}x_j\big|\le CK^\xi N^{-1}.
\ee
Since $\E^{f_t\mu}x_j$ is close to its classical location $\gamma_j$ by 
rigidity \eqref{rig} for Wigner matrices, after conditioning, we
obtain that $\E^{f_{t,\by}\mu_\by}x_j$ is also close to $\gamma_j$, at least for
most $\by$ w.r.t. $f_t\mu$.
Combining this information with \eqref{EEE} yields \eqref{Ex}.
Therefore  Theorem~\ref{thm:local} applies and we will use
it for a Gaussian case, $V(x)=\wt V(x)= x^2/2$ but with two different
boundary conditions $\by, \wt\by\in \bR$. Since this holds for
most $\wt\by$ w.r.t the measure $\mu$, it also holds 
for $\mu$ itself, i.e. the gap statistics of $\mu_\by$ and $\mu$ 
coincide. On the other hand, the 
estimate \eqref{OOO} applied to the 
observable $O(x_j-x_{j+1})$ implies  directly that the single gap
distribution w.r.t. $f_{t,\by}\mu_\by$ and $\mu_\by$ coincide for most
of the $\by$ w.r.t. $f_t\mu$.
Finally, the gap statistics of $f_{t,\by}\mu_\by$ and $f_t\mu$ coincide
for most $\by$ w.r.t. $f_t\mu$ by conditioning. Putting these relations together
we obtain that the gap statistics of $f_t\mu$ and $\mu$ coincide,
i.e. the local measures, that played an important auxiliary role, are eliminated.

Finally, the small Gaussian component present in $f_t\mu$ for small but non-zero $t$ can be removed by the
Green function comparison theorem, Theorem~\ref{comparison}. Although the direct
application of the Green functions give information only on eigenvalues
around a fixed energy and not on an eigenvalue with a fixed label,
the estimates are strong enough to transfer fixed energy information
to fixed label. The main reason for this flexibility is that
 Theorem~\ref{comparison} allows for very small $\eta\sim N^{-1-\e}$, 
i.e. well below the typical spacing.
Indeed, Theorem 1.10 from \cite{KY} implies  that if the first four  moments of
two generalized Wigner ensembles, 
 $H^{\f v}$ and $H^{\f w}$, are the same, then  we have
\be\label{12}
 \lim_{N\to \infty} \big [ \E^{\f v} - \E^{{\f w}}   \big ] 
 O\big( N(x_j-x_{j+1}), N(x_{j}-x_{j+2}), \ldots , N(x_j - x_{j+n})\big)   \;=\; 0.
\ee
 Roughly speaking, the proof of \eqref{12} in \cite{KY} was based on  Theorem~\ref{comparison}.
  In order to convert fixed energy 
to  a fixed eigenvalue  index, one needs to know that the total number of eigenvalues up to a fixed energy is the same 
for the two ensembles. The total number of eigenvalues up to a fixed energy $E$ can be expressed in terms of integration of 
imaginary part of  the trace of  Green functions, i.e.,
$$
  \int^{E}_{-\infty} \rd y \; \im \, \tr \frac 1 { H -  (y +i\eta)}
$$
with  an $\eta$ slightly   smaller than $1/N$.  
 Thus the basic idea of the Green function comparison theorem can be employed and this  leads to \eqref{12}. 
 This completes the proof of Theorem~\ref{thm:sg}. \qed

\end{document}